\documentclass[12pt]{article}
\usepackage{latexsym,amssymb,amsmath} 
\usepackage{cite} 
\usepackage{path}
\usepackage{url}
\usepackage{verbatim}
\usepackage[pdftex]{graphicx}
\usepackage{natbib}
\usepackage[symbol]{footmisc}
\usepackage{float}
\usepackage{bm}

\usepackage{multicol}
\usepackage{hyperref}
\usepackage{amsthm}
\usepackage{subfigure}
\usepackage{booktabs}
\usepackage{paralist}
\usepackage{dsfont}
\usepackage{tikz}
\usepackage{caption}
\usepackage[font=scriptsize,labelfont=bf]{caption}

\numberwithin{equation}{section}
\numberwithin{table}{section}
\numberwithin{figure}{section}

\newcommand{\half}{{\textstyle \frac{1}{2}}}
\newcommand{\eps}{\varepsilon}

\newcommand{\bmu}{\boldsymbol{\mu}}
\newcommand{\btheta}{\boldsymbol{\theta}}
\newcommand{\bDelta}{\boldsymbol{\Delta}}
\newcommand{\bbeta}{\boldsymbol{\beta}}
\newcommand{\boldeta}{\boldsymbol{\eta}}
\newcommand{\bY}{\boldsymbol{Y}}

\newcommand{\bX}{\boldsymbol{X}}

\newcommand{\bOmega}{\boldsymbol{\Omega}}
\newcommand{\bomega}{\boldsymbol{\omega}}
\newcommand{\bZ}{\textbf{Z}}
\newcommand{\bu}{\boldsymbol{u}}
\newcommand{\beps}{\boldsymbol{\epsilon}}

\newcommand{\hmu}{\hat{\bmu}}
\newcommand{\htheta}{\hat{\btheta}}
\newcommand{\hbeta}{\hat{\bbeta}}
\newcommand{\hDelta}{\hat{\bDelta}}
\newcommand{\heps}{\hat{\beps}}
\newcommand{\inSigma}{\Sigma^{-1}}

\newcommand{\hSigma}{\hat{\Sigma}}
\newcommand{\tSigma}{\tilde{\Sigma}}
\newcommand{\dSigma}{\dot{\Sigma}}
\newcommand{\indSigma}{\dSigma^{-1}}
\newcommand{\inhSigma}{\hat{\Sigma}^{-1}}
\newcommand{\inTSigma}{\tilde{\Sigma}^{-1}}

\newcommand{\IR}{\mathbb{R}}

\newcommand{\E}{\mathds{E}}
\newcommand{\norm}[2]{\left\|{#1}\right\|_{{}_{#2}}}
\newcommand{\abs}[1]{\left|{#1}\right|}

\newcommand{\der}[2]{\frac{\partial {#1}}{\partial {#2}}}

\newcommand{\junk}[1]{{}}

\newtheorem{thm}{Theorem}
\newtheorem{asu}{A}
\newcounter{subassumption}[asu]
\newtheorem{lemma}{Lemma}
\newtheorem{corollary}{Corollary}

\renewcommand{\thesubassumption}{(\textit{\roman{subassumption}})}
\makeatletter
\renewcommand{\p@subassumption}{\theasu}
\makeatother
\newcommand{\subasu}{
  \refstepcounter{subassumption}%
  \thesubassumption~\ignorespaces}
\newlength{\fwtwo} 

\begin{document}
\DeclareGraphicsExtensions{.jpg}

\begin{center}
\textbf{\Large High Dimensional Classification for Spatially Dependent Data with Application to Neuroimaging} \\[6pt]
  Yingjie Li$^\ast$, Liangliang Zhang$^\dag$ and Tapabrata Maiti$^\ast$ \\[6pt]
  $^\ast$Department of Statistics and Probability,
  Michigan State University  \\[6pt]
   $^\dag$Department of Biostatistics,
   The University of Texas MD Anderson Cancer Center\\[6pt]
    liangliangzhang.stat@gmail.com
\end{center}

\begin{abstract}
{
Discriminating patients with Alzheimer's disease (AD) from healthy subjects is a crucial task in the research of Alzheimer's disease.
The task can be potentially achieved by linear discriminant analysis (LDA), which is one of the most classical and popular classification techniques.
However, the classification problem becomes challenging for LDA because of the high-dimensionally and the spatial dependency of the brain imaging data. 
To address the challenges, researchers have proposed various ways to generalize LDA into high-dimensional context in recent years. 
However, these existing methods did not reach any consensus on how to incorporate spatially dependent structure. 
In light of the current needs and limitations, we propose a new classification method, named as Penalized Maximum Likelihood Estimation LDA (PMLE-LDA).
The proposed method uses $Mat\acute{e}rn$ covariance function to describe the spatial correlation of brain regions.
Additionally, PMLE is designed to model the sparsity of high-dimensional features. 
The spatial location information is used to address the singularity of the covariance.
Tapering technique is introduced to reduce computational burden.
We show in theory that the proposed method can not only provide consistent results of parameter estimation and feature selection, but also generate an asymptotically optimal classifier driven by high dimensional data with specific spatially dependent structure.
Finally, the method is validated through simulations and an application into ADNI data for classifying Alzheimer's patients.}
\end{abstract}

\section{Introduction}
\label{sec_introSpatialLDA}
{
This paper is motivated by discriminating patients with Alzheimer's disease (AD) from healthy subjects using structural Magnetic Resonance Imaging (sMRI) data. We also would like to identify the key sMRI features that differentiate the two groups. We translate the real needs as a technical problem of using classification method and selecting features.    
However, the problem becomes challenging because of the complexity of the data. First, brain imaging data is spatially dependent, which means that the dependence between various voxels (pixels) can be depicted by their proximity. Second, brain imaging data is high-dimensional, because a single subject can produce hundreds of 3D MRI scans and a single 3D MRI scan can generate millions of voxels. 

Many existing methods could handle the challenges of brain imaging data, but they did not solve a classification problem. In brain research, the assignment of functional regions has been mainly based on certain assumptions and conceptualizations. In particular, conceptualization of spatial partition and correspondence is widely used in brain imaging analysis (see \citep{worsley2003developments,smith2007spatial,lindquist2008statistical,musgrove2016fast,Bowman2014}). These methods exploited spatial dependence of brain imaging data, but they were not designed to classify and identify spatial features. In addition to brain imaging, spatial analysis has been broadly applied across agriculture, geology, soil science, oceanography, forestry, meteorology and climatology. Traditionally, these applications are not necessarily high-dimensional. But an increasing trend of big data shows that high-dimensional data with spatially dependent structure attracts more and more interest. For recent developments of regularized models dealing with spatially dependent data, please refer to \citep{hoeting2006model,huang2007optimal,zhu2010selection,chu2011penalized,reyes2012selection,fu2013estimation,nandy2017additive,feng2016variable}. One can simply adapt these available procedures, but again they did not consider how to classify and identify spatial features. 

Many high-dimensional approaches are not preferable due to the concern of high variance and overfitting issues \citep{friedman2001elements}. So it is necessary to incorporate regularization techniques into the classification method for high dimensional data. Fisher's linear discriminant analysis (LDA) is one of the most classical and popular classification techniques. The simplicity and flexibility of LDA has allowed itself to be extended to many complex and high dimensional applications. Researchers have proposed many ways to generalize LDA into high-dimensional context. However, these existing methods did not reach any consensus on how to incorporate spatial dependence structure. In light of the current needs and limitations, we propose a new LDA procedure accommodating both the complex dependent structure and high-dimensionality. 

{Before introducing the proposed procedure, we first give a review of the existing LDA methods for high dimensional data. Let us consider the $p$-dimensional discriminant problem between two classes $\mathcal{C}_1$ and $\mathcal{C}_2$. According to some classification rule $T(\bOmega): R^p \to \{1,2\}$, a new observation $\bOmega$ can be classified into class $\mathcal{C}_1= \{ \bOmega : T(\bOmega)=1\}$ or $\mathcal{C}_2 = \{ \bOmega : T(\bOmega)=2\}$. Given that $\bOmega\in \mathcal{C}_1$, the misclassification rate is the conditional probability of that $\bOmega$ is classified into class $\mathcal{C}_2$, i.e. $P(T(\bOmega)=2|\bOmega\in \mathcal{C}_1)$. Similarly, $P(T(\bOmega)=1|\bOmega\in \mathcal{C}_2)$ denotes the misclassification rate when $\bOmega\in\mathcal{C}_2$.}

{The optimal classifier obtained by minimizing the posterior probability is known as the Bayes rule, which classifies the new observation into the most probable class (Chapter 2 in \cite{friedman2001elements}). Suppose that $f_k(\bomega)$ denotes the density of the misclassification rate that an observation $\bomega$ is classified into $\mathcal{C}_{k}$, ($k=1,2$). Let $\pi_k$ be the prior probability of class $k$ with $\pi_1+\pi_2=1$. According to Bayes theorem, the posterior probability of an observation $\bOmega=\bomega$ in each class is 
$
P(\bOmega\in \mathcal{C}_k|\bOmega=\bomega)=\frac{f_k(\bomega)\pi_k}{f_1(\bomega)\pi_1+f_2(\bomega)(1-\pi_1)}.
$
}

A typical way of modeling the class densities is that they are assumed to be multivariate Gaussian $N(\bmu_1, \Sigma)$ and $N(\bmu_2,\Sigma)$ respectively, where $\bmu_k$ ($k=1,2$) are the class mean vectors and $\Sigma$ is the common positive definite covariance matrix. Then the density of an observation $\bOmega=\bomega$ from $\mathcal{C}_k$ can be written as
$
f_k(\bomega)=\frac{1}{(2\pi)^{p/2}}|\Sigma|^{1/2}e^{-{1\over 2} (\bomega-\bmu_k)^T\inSigma(\bomega-\bmu_k))}.
$
}
Under this assumption, the Bayes rule assigns $\bOmega=\bomega$ into $\mathcal{C}_1$ if $\pi_1f_1(\bomega)\ge \pi_2 f_2(\bomega)$. Equivalently, $\bomega$ is assigned to $\mathcal{C}_1$ if
$
\text{log}\frac{\pi_1}{\pi_2}+(\bomega-\bmu)^T\inSigma (\bmu_1-\bmu_2)\ge 0,
$
where $\bmu=(\bmu_1+\bmu_2)/2$. Notice that this classifier is linear in $\bomega$.

{In practice, the parameters of the Gaussian distribution should be estimated using the training data.} Suppose that $\bY_{k1},..., \bY_{kn_k}$ are training data from class $\mathcal{C}_k$, where $k\in\{1,2\}$ and $\bY_{kj}\in \mathbb{R}^p$ are independent and identically distributed as $N_p(\bmu_k,\Sigma(\btheta_0))$, where $\bmu_k=(\mu_{k1},...,\mu_{kp})^{T}$, $n_k$ is the sample size for class $\mathcal{C}_k$. $\Sigma(\btheta)$ is the covariance matrix with parameter $\btheta=\btheta_0$. Assume that $\hat{\bmu}_1$, $\hat{\bmu}_2$, $\hSigma$ and $\hDelta$ (generated from $\bY_{k1},..., \bY_{kn_k}$) are estimates of $\bmu_1$, $\bmu_2$, $\Sigma$ and $\bDelta$, where $\bDelta=(\Delta_1,...,\Delta_p)^T=\bmu_1-\bmu_2$ is the difference of the two classes in mean.

Let $n=n_1+n_2$ be the total sample size. Assume that  $\frac{n_1}{n}\to \pi$, $0<\pi<1$ as $n\to \infty.$ $p$ depends on $n$. Assume that the two classes have equal prior probabilities, i.e. both the probabilities that a new observation comes from $\mathcal{C}_1$ and $\mathcal{C}_2$ are $\half$. Then we obtain the classification rule $\hat{\delta}$:
\begin{align}\label{LDA}
\hat{\delta}(\bOmega)=(\bOmega-\frac{\hat{\bmu}_1+\hat{\bmu}_2}{2})^T\hat{\Sigma}^{-1}\hat{\bDelta}.
\end{align}
A new observation $\bomega$ is classified into class $\mathcal{C}_1$ if $\hat{\delta}(\bomega)>0$ and $\mathcal{C}_2$ otherwise. If the new observation $\bOmega$ comes from $\mathcal{C}_1$, then the  conditional misclassification rate of $\hat{\delta}$ is
\begin{align}\label{W1}
W_1(\hat{\delta})=P(\hat{\delta}(\bOmega)\le 0|\bOmega\in \mathcal{C}_1, \bY_{ki},i=1,2,...,n_k,k=1,2)=1-\Phi(\Psi_1),
\end{align}
where
\begin{align}\label{Psi1}
\Psi_1=\frac{(\bmu_1-\hmu)^T\hSigma(\hmu_1-\hmu_2)}{\sqrt{(\hmu_1-\hmu_2)^T\inhSigma\Sigma\inhSigma(\hmu_1-\hmu_2)}}.
\end{align}
Similarly, we can define the error rate for observations from $\mathcal{C}_2$. If a new observation $\bOmega$ comes from class $\mathcal{C}_2$, the  conditional misclassification rate of $\hat{\delta}$ is:
\begin{align}\label{W2}
W_2(\hat{\delta})=\mathbb{P}(\hat{\delta}(\bOmega)>0|\bOmega\in \mathcal{C}_2, \bY_{ki},k=1,2;i=1,...,n_k)=\Phi(\Psi_2),
\end{align}
where
\begin{align}\label{Psi2}
\Psi_2=\frac{(\bmu_2-\hmu)^T\hSigma(\hmu_1-\hmu_2)}{\sqrt{(\hmu_1-\hmu_2)^T\inhSigma\Sigma\inhSigma(\hmu_1-\hmu_2)}}.
\end{align}

As we assume the equal prior probability for the two classes, the overall misclassification rate is defined as
\begin{align}\label{W}
W(\hat{\delta})=\half(W_1(\hat{\delta})+W_2(\hat{\delta})).
\end{align}

If $\bmu_1$, $\bmu_2$ and $\Sigma$ are known, the optimal classification rule is Bayes rule, which classifies a new observation $\bOmega=\bomega$ into class $\mathcal{C}_1$ if
\begin{align}\label{BayesRule}
\delta(\bomega)=(\bomega-\frac{\bmu_1+\bmu_2}{2})^{T}\Sigma^{-1}\bDelta>0.
\end{align}

Bayes rule has the smallest misclassification rate. If there's a new observation $\bOmega$ from class $\mathcal{C}_1$, since $\bOmega$ has normal distribution $N(\bmu_1,\Sigma(\btheta))$, we can calculate that the conditional misclassification rate of Bayes rule $\delta$ is
\begin{align}
W_1(\delta)=W_2(\delta)=1-\Phi(\frac{\sqrt{C_p}}{2}),
\end{align}
where $C_p=\bDelta^T\inSigma(\btheta)\bDelta$ and $\Phi(\cdot)$ is the standard Gaussian distribution function.

We obtain the overall misclassification rate of Bayes rule as
$W(\delta)=1-\Phi(\frac{\sqrt{C_p}}{2})$.
{
Since Bayes rule has the smallest misclassification rate, we write $W_{OPT}=1-\Phi(\frac{\sqrt{C_p}}{2})$ as the optimal misclassification rate. Under certain conditions, we could have $C_p\to C_0$, then $W_{OPT}\to 1-\Phi(\frac{\sqrt{C_0}}{2})$, where $C_0$ is a constant.}

Using training data, we can estimate the parameters with the sample mean and covariance
\begin{align}
&\hat{\bmu}_k=\sum_{i=1}^{n_k} \bY_{ki}/n_k=\bar{\bY}_{k\cdot}, \\
&\hat{\Sigma}=\sum_{k}\sum_{i}(\bY_{ki}-\hat{\bmu}_k)^T(\bY_{ki}-\hat{\bmu}_k)/(n_1+n_2-2), k=1,2.
\end{align}
Then LDA classifies $\bOmega$ into class $\mathcal{C}_1$ if
\begin{align}
\hat{\delta}_{LDA}(\bOmega)=(\bOmega-\frac{1}{2}(\bar{\bY}_{1\cdot}+\bar{\bY}_{2\cdot}))^{T}\hSigma^{-1}(\bar{\bY}_{1\cdot}-\bar{\bY}_{2\cdot})>0
\end{align}

LDA is an asymptotically optimal classifier under traditional large sample scenario, that is, the dimension of variables ($p$) is fixed and the sample size ($n$) tends to infinity.  However, this is not true in the high dimensional context. \cite{bickel2004some} demonstrated that LDA asymptotically did not perform better than random guessing if $p/n\to \infty$.

{The asymptotic theory of LDA does not hold under high dimensional setting because of two reasons.} First, the sample covariance matrix $\hat{\Sigma}$ is singular. It is difficult to estimate the precision matrix $\Omega=\inSigma$. To resolve this, the independence rule (IR) ignores the correlations among features and use diagonal of $\hSigma$ to replace $\hSigma$. \cite{bickel2004some} showed in theory that IR leads to a better classification result than the naive LDA, where the Moore-Penrose inverse is used to replace $\inhSigma$. Another similar way to resolve this issue is the nearest shrunken centroid classifier \citep{tibshirani2002diagnosis}. \cite{fan2008high} proposed the feature annealed independence rule (FAIR) that performs feature selection by t-test in addition to IR. The above methods made LDA applicable for high dimensional classification. However, they ignored the covariance structure of the features and the classifiers from those methods were not asymptotically optimal.
{
Some methods have been proposed for covariance matrix estimation or precision matrix estimation in high dimension scenario \citep{bickel2008covariance, bickel2008regularized, rothman2008sparse, cai2016minimax, cai2016estimating}.
The covariance matrix estimated by these methods can be directly used in LDA and address the singularity of $\hSigma$. However, an accurate estimate of $\Sigma$ does not necessarily lead to better classification. \cite{fan2008high} and \cite{shao2011sparse} showed that even though the true covariance matrix is known, the classification could be no better than random guess because of the noise accumulated from estimating the means. This arouses the second challenge for high dimensional LDA. That is, the noise introduced by the estimation of many non-informative features would lead to poor classification performance. Therefore, the regulation of features is needed.}

{To address the second challenge stated above, \cite{witten2011penalized} proposed penalized LDA by applying penalties on the feature vectors. \cite{cai2011direct},\cite{mai2012direct} and \cite{fan2012road} assumed sparsity and introduced penalization on the discriminant direction $\bm \zeta=\Sigma^{-1}(\bmu_2-\bmu_1)$, which was also adopted by \cite{cai2018high}. These methods regularized the estimated discriminant direction $\bm \zeta$ directly, avoiding of estimating $\Sigma^{-1}$. The advantages are obvious. The penalization reduced the noise accumulated in estimation on high dimensional features. The obtained classifiers incorporated the covariance structure among features. However, the disadvantage is that the panelized results were not straightforward to interpret. In particular, the panelized directions did not convey any information on which features should be selected. Because we know feature selection is more relevant to the research question. In this paper, we adopt another types of work that assume sparsity on the feature difference $\bDelta=\bmu_2-\bmu_1$ (the difference of means between the two classes). \cite{shao2011sparse} assumed sparsity and put hard threshold on both the feature difference $\bDelta$ and the covariance matrix $\Sigma$. This method provided an asymptotically optimal classifier. The double threshold method in \cite{shao2011sparse} was then extended to quadratic discriminant analysis in \cite{li2015sparse}. \cite{xu2014covariance} proposed a covariance-enhanced method to achieve feature selection for linear discriminant analysis. However, this method did not work directly on selecting informative features. Most recently, \cite{cannings2017random} used random projection to perform dimension reduction to address this issue. 
However, this method did not work well when the data is sparse in ultrahigh dimensional settings, say $p$ is in the thousands while sample size $n$ ranges from $50$ to $1000$.}

{
In this article, we develop a new classification procedure, named as Penalized Maximum Likelihood Estimation LDA (PMLE-LDA), for high dimensional data with spatially dependent features. We assume that the features follow multivariate normal distribution. We structure the covariance matrix by a spatial covariance function (e.g.  $Mat\acute{e}rn$ covariance function). By introducing the spatial structure, the covariance matrix can be estimated by maximum likelihood estimation no matter how large the number of features is, compared to the sample size. This estimated covariance matrix can address the first challenge of covariance singularity. To address the second challenge of mean misidentification, we assume that the feature difference between two classes is sparse, which indicates that only a fraction of the $p$ features contribute in differentiating the two classes. Given a training data set, we use Penalized Maximum Likelihood Estimation (PMLE) method to perform parameter estimation. The resulting estimates are plugged into LDA model to construct a new classifier. We show in theory that the proposed procedure can not only provide consistent results of parameter estimation and feature selection, but also generate an asymptotically optimal classifier driven by high dimensional data with specific spatially dependent structure. To the best of our knowledge, the proposed method is the first to use spatial correlations to adjust the classification rules. Although we develop PMLE-LDA under the linear framework for two classes classification, it could be potentially extended to other classification methods such as quadratic discriminant analysis (QDA) and multi-classes classification problems.}

The rest of the paper is organized as follows. In section \ref{sec_MLELDA},
we show that MLE-LDA is asymptotically optimal if $p/n\to 0$ under some regularity conditions. We also show that in high dimensional setting ($p/n \to C>0$), the MLE-LDA performs poorly (no better than random guess) even if the true covariance is known unless the signals are very strong. This indicates the necessity of penalization for the MLE.
In section \ref{sec-PLDA}, we propose to estimate the parameters by penalized MLE (PMLE) by applying a penalty on $\bDelta=\bmu_1-\bmu_2$, which measures the difference of the two classes in mean. We assume the sparsity of $\bDelta$. Then we derive and prove the parameter estimation consistency and feature selection consistency of PMLE. In the end, PMLE-LDA is constructed. We show that it is asymptotically optimal even if $p/n\to C>0$. Simulation study and real data analysis are conducted in section \ref{sec-simulation} and \ref{sec-data}.

\section{Classification using maximum likelihood estimate (MLE-LDA)}
\label{sec_MLELDA}
\subsection{Spatial models}
\label{sec_introSpatial}
In this section, we introduce necessary terminologies and assumptions in spatial statistics. For a spatial domain of interest $D$ in $\mathbb{R}^d$, we consider two classes of spatial processes $\{y_k(s): s\in D,k=1,2\}$, ($k=1,2$), such that
\begin{align}
y_k(s)=\mu_k(s)+\epsilon(s),
\end{align}
where $\mu_k(s)$ is the mean effect function and $\epsilon(s)$ is the corresponding random noise. Assume that the error process $\{\epsilon(s):s\in D\}$ is a Gaussian process with mean zero and a covariance function
\begin{align}\label{covfunction}
\gamma(s,s\ensuremath{'};\btheta)=cov(\epsilon(s),\epsilon(s\ensuremath{'})),
\end{align}
where $s,s\ensuremath{'}\in D$ and ${\btheta}$ is a $q \times 1$ vector of covariance function parameters. We assume that the spatial domain is expanding as the number of samples on the domain is increasing.

\begin{asu}\label{asu_domain}
Assume the sample set $D\in \IR^d$ ($d\ge 1$) is predetermined and non-random with the restriction $\norm{s_i-s_j}{2}\ge \epsilon>0$, for $s_i,\ s_j\in D$ for all pairs $i,j=1,2,...,p$ to ensure that the sampling domain increases in extent as $p$ increases.
\end{asu}

Assume for any sample of the spatial processes, there are observations at $p$ discrete sites $s_1,...,s_p\in D$. Suppose $y_{ki}(s)\ (i=1,2,...,n_k)$ is from class $\mathcal{C}_k$ ($k=1,2$). Let $Y_{kij}=y_{ki}(s_j)$ be the observation at $j$th site for the $i$th sample of spatial process $y_k(s)$, where $k=1,2$, $i=1,2,...,n_k$, $j=1,2,...,p$, then the $j$th observation for sample $i$ can be represented by
\begin{align}
Y_{kij}=\mu_{kj}+\epsilon_{kij},
\end{align}
where $\mu_{kj}=\mu_k(s_j)$ is the mean effect at $j$th location in class $\mathcal{C}_k$ and $\epsilon_{kij}=\epsilon_{ki}(s_j)$ is the corresponding Gaussian random noise for $i$th sample at $j$th location. In matrix notation, the above model can be written as
\begin{align}\label{Yki}
\bY_{ki}=\bmu_k+\boldsymbol{\epsilon}_{ki},
\end{align}
where $\bY_{ki}=(Y_{ki1},...,Y_{kip})^{T}$, $\bmu_k=(\mu_{k1},...,\mu_{kp})^{T}$ is the mean vector of class $\mathcal{C}_k$ and $\boldsymbol{\epsilon}_{ki}=(\epsilon_{ki1},...\epsilon_{kip})^{T}$ has multivariate nomal distribution $N(\textbf{0},\Sigma)$.
As $\epsilon(s)$ has a covariance function (\ref{covfunction}), the covariance matrix $\Sigma$ can be represented by $\Sigma(\btheta)=[\gamma(s_i,s_j;\btheta)]_{i,j=1}^{p}$, i.e. $\gamma(s_i,s_j)$ is the $(i,j)$th entry. From (\ref{Yki}), we have
\begin{align}\label{Y_dis}
\bY_{ki}\sim N(\bmu_k,\Sigma({\btheta})).
\end{align}

Assume $\btheta_0$ be the true parameter in (\ref{covfunction}). If $\btheta=\btheta_0$, we write $\Sigma(\btheta)$ as $\Sigma$ for simplicity. Next, we make some assumptions on the covariance function $\gamma(s_i,s_j;\btheta)$:
\begin{asu} \label{asu_covfunc}
\subasu Let $\Xi$ be the parameter space for $\btheta$. Assume the covariance function $\gamma(s,s^{\prime};\btheta)$ is stationary, isotropy, and twice differentiable with respect to $\btheta$ for all $\btheta \in \Xi$ and $s,s^{\prime}\in D$. \\
\subasu $\gamma(s,s^{\prime};\btheta)$ is positive-definite in the sense that for every finite subset $\{s_1,s_2,...,s_p\}$ of $D$, the covariance matrix $\Sigma=[\gamma(s_i,s_j;\btheta)]$ is positive-definite.
\end{asu}

Under the stationary and isotropic assumption, $\Sigma(\btheta)$ could be written as $\Sigma(\btheta)=[\gamma(h_{ij};\btheta)]_{i,j=1}^{p}$, where $h_{ij}=\norm{s_i-s_j}{2}$ is the Euclidean distance between sites $s_i$ and $s_j$.

There are many ways to model the covariance function $\gamma(h;\btheta)$. A widely used family of covariance function is the $Mat\acute{e}rn$ covariance function. It is defined as:
\begin{align}\label{matern}
\gamma(h;\sigma^2,c,\nu,r):=\sigma^2(1-c)\frac{2^{1-\nu}}{\Gamma(\nu)}({h}/r)^{\nu}K_{\nu}({h}/r)
\end{align}
where $K_{\nu}(\cdot)$ is a modified Bessel function of the second kind and $\sigma^2>0$ is the variance, $0\le c\le 1$ is a nugget effect, $\nu>0$ is the scale and smoothness parameter \citep{cressie1992statistics}. First, the $Mat\acute{e}rn$ covariance function is isotropic and the correlation decreases when the distance $h$ increases. Second, when $\nu$ increases, the smoothness of the random field increases. Moreover, the $Mat\acute{e}rn$ covariance function converges to Gaussian covariance function $\gamma(h;\sigma^2,c,r)=\sigma^2(1-c)exp(-h^2/r^2)$ as $\nu\to\infty$. Last, if $\nu=\half$, (\ref{matern}) is reduced to the well known exponential covariance function $\gamma(h;\sigma^2,c,r)=\sigma^2(1-c)exp(-h/r)$. $r$ is called the range parameter since it measures the distance at which the correlation have decreased below certain threshold.

\subsection{Classification using MLE-LDA}
\label{sec-MLE-LDA}
Under the setting of spatial statistics, the covariance structures $\Sigma$ and means $\bmu_1$, $\bmu_2$ and $\bDelta=\bmu_1 - \bmu_2$ can be estimated by maximum likelihood estimation (MLE). Plugging the MLE into (\ref{LDA}) results in a MLE-LDA classifier. In this section, we investigate the properties of MLE-LDA. We prove that MLE-LDA is asymptotically optimal if $p/n \to 0$ but perform poorly while $p/n \to C > 0$ even if the true covariance is known.  

Let $\bY=(\bY_{11}^{T},...,\bY_{1n_1}^{T},\bY_{21}^{T},...,\bY_{2n_2}^{T})^{T}$, where $\bY_{ki}$ ($k=1,2$ and $i=1,2,...,n_k$) is defined as in section \ref{sec_introSpatial}. As defined in (\ref{Y_dis}), $\bY$ is a $np$ dimension vector and follows a multivariate normal distribution. Then we have the log-likelihood function for $\bmu_k$ and $\btheta$
\begin{align}\label{loglik}
L(\btheta, \bmu_1,\bmu_2; \bY)=&-\frac{p(n_1+n_2)}{2}log(2\pi)-\frac{n_1+n_2}{2}log|\Sigma(\btheta)|\\\nonumber
&-\frac{1}{2}\sum_{k=1}^{2}\sum_{i=1}^{n_k}(\bY_{ki}-\bmu_k)^{T}\Sigma(\btheta)^{-1}(\bY_{ki}-\bmu_k).
\end{align}
We can estimate $\bmu_1$, $\bmu_2$ and $\btheta$ by MLE even in high dimensional settings. According to the setting of the spatial model in section \ref{sec_introSpatial}, the resulting $\Sigma(\htheta)$ is a positive definite matrix. This resolves the first challenge of sample covariance singularity.

We denote the resulting estimates as $\hmu_{1MLE}, \hmu_{2MLE}, \htheta_{MLE}$, which can be  plugged in (\ref{LDA}) and get the MLE-LDA classifier $\hat{\delta}_{MLE}:$
\begin{align}\label{MLE-LDA}
\hat{\delta}_{MLE}(\bOmega)=(\bOmega-\frac{\hmu_{1MLE}+\hmu_{2MLE}}{2})^T\inSigma(\htheta_{MLE})(\hmu_{1MLE}-\hmu_{2MLE}).
\end{align}

In this section, we investigate the consistency of parameter estimation of (\ref{loglik}). Also we investigate the classification performance of $\hat{\delta}_{MLE}.$   Let $\bmu_1=(\mu_{11},\mu_{12},...,\mu_{1p})$, $\bmu_2=(\mu_{21},\mu_{22},...,\mu_{2p})$ and $\btheta_0=(\theta_{01},\theta_{02},...,\theta_{0q})$ be the true parameters. Let $\Sigma_k(\btheta), k=1,2,...,q$ be the partial derivative of the matrix $\Sigma(\btheta)$ with respect to $\theta_k$, i.e. $\frac{\partial}{\partial \theta_k}\Sigma(\btheta)=\Sigma_{k}(\btheta)$. Also let $\Sigma^k(\btheta), k=1,2,...,q$ denote the partial derivative of the matrix $\Sigma(\btheta)^{-1}$ with respect to $\theta_k$, i.e. $\frac{\partial}{\partial \theta_k}\Sigma^{-1}(\btheta)=\Sigma^{k}(\btheta)$. Also, denote $\Sigma_{kj}(\btheta)=\frac{\partial \Sigma(\btheta)}{\partial \theta_k \partial \theta_j}$ and $\Sigma^{kj}(\btheta)=\frac{\partial \Sigma^{-1}(\btheta)}{\partial \theta_k \partial \theta_j}$. We are going to simplify the notation if $\btheta=\btheta_0$, i.e. we write $\Sigma(\btheta_0)$ as $\Sigma$, $\inSigma(\btheta_0)$ as $\Sigma^{-1}$, $\Sigma_k(\btheta_0)$ as $\Sigma_k$ and $\Sigma^k(\btheta_0)$ as $\Sigma^k$. For a square matrix $A$, denote the set of all the eigenvalues by $\lambda(A)$. Moreover, denote the maximum and minimum eigenvalues by $\lambda_{max}(A)$ and $\lambda_{min}(A)$, respectively.

We need to assume some regularity conditions for Theorem \ref{ThconsistMLE}.
\begin{asu} \label{asu_eigen_sig}
$\lim sup_{p \to \infty}\lambda_{\max}(\Sigma)<\infty$, $\lim  inf_{p \to \infty}\lambda_{\min}(\Sigma)>0$
\end{asu}

\begin{asu}\label{asu_covnorm_inv}
$\norm{\Sigma_k}{F}^{-2}=O_p(p^{-1})$, {where $\norm{\Sigma_k}{F}=\sum_{i,j=1}^{p}\gamma_k^2(h_{ij};\btheta)$, where $\gamma_k(h_{ij};\btheta)=\der{\gamma(h_{ij};\btheta)}{\theta_k}$, $k=1,2,...,q$ and $\btheta$ is a $k$ dimensional parameter.}
\end{asu}

\begin{asu}\label{asu_tij_lim}
Assume $\lim_{p\to\infty}a_{ij}$ exist, where $a_{ij}=\frac{t_{ij}}{t_{ii}^{1/2}t_{jj}^{1/2}}$ and $t_{ij}=tr(\inSigma\Sigma_i\inSigma\Sigma_j)$.
\end{asu}

\begin{asu}\label{asu_cov_star}
There exists an open subset $\omega$ that contains the true parameter point $\btheta_0$ such that for all $\btheta^*\in\omega$, we have:\\
\subasu \label{asu_star_coveigen}
$-\infty < \lim_{p \to \infty}\lambda_{\min}(\Sigma_k(\btheta^*))<\lim_{p \to \infty}\lambda_{\max}(\Sigma_k(\btheta^*))<\infty$. \\
\subasu \label{asu_star_coveigen2}
$-\infty< \lim_{p\to \infty}\lambda_{\min}(\Sigma_{kj}(\btheta^*))<\lim_{p\to \infty}\lambda_{\max}(\Sigma_{kj}(\btheta^*))<\infty.$\\
\subasu \label{asu_star_covdif3}
$\parallel \frac{\partial t_{ij}(\btheta^*)}{\partial\btheta} \parallel_{2}=O_p(p)$, where $t_{ij}(\btheta^*)=tr(\inSigma(\btheta^*)\Sigma_i(\btheta^*)\inSigma(\btheta^*)\Sigma_j(\btheta^*))$
\end{asu}

Since $\Sigma^k=-\Sigma\Sigma_k\Sigma$ and $\Sigma^{kj}=\Sigma^{-1}(\Sigma_k\inSigma\Sigma_j+\Sigma_j\inSigma\Sigma_k-\Sigma_{kj})\Sigma^{-1}$, by  A\ref{asu_eigen_sig} and A \ref{asu_cov_star}, we have
$$-\infty < \lim_{p \to \infty}\lambda_{\min}(\Sigma^k(\btheta^*))<\lim_{p \to \infty}\lambda_{\max}(\Sigma^k(\btheta^*))<\infty$$ and
$$-\infty< \lim_{p\to \infty}\lambda_{\min}(\Sigma^{kj}(\btheta^*))<\lim_{p\to \infty}\lambda_{\max}(\Sigma^{kj}(\btheta^*))<\infty.$$

Notice that for any $p\times p$ matrix $A$ we have $\norm{A}{F}\le \sqrt{p}\norm{A}{2}=\sqrt{p}\lambda_{\max}(A)$, then from A \ref{asu_cov_star} we have:
\begin{itemize}
\item[(1)] $\norm{\Sigma^k(\btheta^*)}{F}=O_p(\sqrt{p})$;
\item[(2)] $\norm{\Sigma^{kj}(\btheta^*)}{F}=O_p(\sqrt{p}).$
\end{itemize}

First we have the following theorem about MLE consistency of (\ref{loglik}).
\begin{thm}\label{ThconsistMLE}
Assume A\ref{asu_covfunc}-A\ref{asu_cov_star} hold .  Let $(\bmu_1,\bmu_2,\btheta_0)$ be the true parameter. The maximum likelihood estimate (MLE) of (\ref{loglik}) is: $\hmu_{1MLE}=\bar{\bY}_{1\cdot}$, $\hmu_{2MLE}=\bar{\bY}_{2\cdot}$, $\hat{\btheta}_{MLE}$, where $\bar{\bY}_{k\cdot}=\sum_{i=1}^{n_k}\bY_{ki}/n_k$. Also,
\begin{itemize}
\item[(i)] If $p/n\to 0$, $\norm{\hat{\btheta}_{MLE}-\btheta_0}{2}=O_p(\frac{1}{\sqrt{np}})$;
\item[(ii)] If $p/n\to C$ with $0<C \le \infty$ and $\sqrt{p}/n\to 0$, $\norm{\hat{\btheta}_{MLE}-\btheta_0}{2}=O_p(\frac{1}{n})$.
\end{itemize}
\end{thm}
\begin{proof}
See Supplementary Materials.
\end{proof}

{In Theorem \ref{ThconsistMLE}, A\ref{asu_covfunc} is necessary to ensure the good property for covariance matrix $\Sigma(\btheta)$. A\ref{asu_eigen_sig} and A\ref{asu_tij_lim} are assumed in \cite{mardia1984maximum} for existence of MLE for spatial regression for fixed $p$ when $n \to \infty$. A\ref{asu_covnorm_inv} and A\ref{asu_cov_star} are necessary conditions for the parameter consistency when $p\to\infty$. It is obvious that $Mat\acute{e}rn$ covariance function satisfies A\ref{asu_covfunc}. In the remark section \ref{remarks_asu}, we verified the $Mat\acute{e}rn$ covariance function satisfy the first part of A\ref{asu_eigen_sig} and A\ref{asu_covnorm_inv}-A\ref{asu_cov_star}. $Mat\acute{e}rn$ covariance function also satisfy the second part of A\ref{asu_eigen_sig} if A\ref{asu_domain} holds. } Theorem \ref{ThconsistMLE} shows that under the spatial statistical model, all the parameters ($\bmu_1,\bmu_2, \btheta$) can be estimated consistently by the MLE for either $p/n\to 0$ or $p/n$ goes to a positive constant or $\infty$. Therefore we obtain a positive definite covariance matrix estimate of $\Sigma(\btheta)$.  We can then plug-in the MLE's,  $\hat{\bmu}_1=\hat{\bmu}_{1MLE}$, $\hat{\bmu}_2=\hat{\bmu}_{2MLE}$ and $\hat{\Sigma}=\Sigma(\hat{\btheta}_{MLE})$ into (\ref{BayesRule}) to build up the classification function \ref{MLE-LDA}.

Then a new observation $\bomega$ of $\bOmega$ would be classified into class $\mathcal{C}_1$ if $\hat{\delta}_{MLE}(\bomega)>0$ and $\mathcal{C}_2$ otherwise. Using the same notations in section \ref{sec_introSpatialLDA}, the conditional misclassification rate is defined by (\ref{W1}) and (\ref{W2}). For simplicity, We are going to use $\hat{\bmu}_1$ to denote $\hat{\bmu}_{1MLE}$, $\hat{\bmu}_2$ to denote $\hat{\bmu}_{2MLE}$, $\hat{\btheta}$ to denote $\hat{\btheta}_{MLE}$ in this section.

We will see in Theorem \ref{ThErrpLessn} that the approximate optimal error rate can be achieved while $p/n\to 0$. However, Theorem \ref{ThErrp>n} shows if $p/n\to C$ with $0<C\le\infty$, the error rate would be no better than random guessing even if we know the true covariance, due to the error accumulated in the estimation of $\bmu_1$ and $\bmu_2$.

\begin{thm}\label{ThErrpLessn}
Let $C_p=\bDelta^T\Sigma(\btheta)\bDelta$. Assume ${p\over n}\to 0$, $C_p \to C_0$ with $0\le C_0\le \infty$, and $nC_p\to \infty$ as $n,p\to \infty$.
\begin{itemize}
\item[(1)] The overall misclassification rate $W(\hat{\delta}_{MLE})$ is asymptotically sub-optimal. In other words, $W(\hat{\delta}_{MLE}) \to 1-\Phi(\frac{\sqrt{C_0}}{2})$.
\item[(2)] Moreover, if $C_p\to C_0$ with $0\le C_0<\infty$ or if $C_p\to\infty$ and $C_p \frac{p}{n}\to 0$, then $W(\hat{\delta}_{MLE})$ is asymptotically optimal, i.e. $\frac{W(\hat{\delta}_{MLE})}{W_{OPT}}\stackrel{P}{\to} 1$.
\end{itemize}

\end{thm}

\begin{proof}
See Supplementary Materials.
\end{proof}

The following theorem shows that while $\frac{p}{n}$ goes to a positive constant or $\infty$, even though the true covariance is known, the error accumulated in the estimation of $\bmu_1$ and $\bmu_2$ would cause biased misclassification rate unless the signal levels ($C_p$) are extremely high. This discovery suggests that even though there's no problem in parameter estimation in our model even in high dimensional case, it is still necessary to select important features for classification.
\begin{thm}\label{ThErrp>n}
Assume the true covariance $\Sigma$ is known, denote the classifier function as
\begin{align}
\delta_{\hmu}(\bomega)=(\bomega-\hmu)^T\Sigma^{-1}(\hmu_1-\hmu_2)
\end{align}
where $\hat{\bmu}_1$, $\hat{\bmu}_2$ are MLE in (\ref{loglik}) and $\hat{\bmu}=\frac{\hat{\bmu}_1+\hat{\bmu}_2}{2}$. Assume $p/n\to C$ with $0<C\le \infty$, $C_p\to C_0$ with $0\le C_0\le \infty$. Assume $n_1\ne n_2$ and $n_k>\frac{n}{4}$ ($k=1,2$), then
	\begin{itemize}
	\item[(1)] For $\frac{C_p}{p/n}\to \infty$, then $W(\hat{\delta}_{\hmu})\stackrel{P}{\to} 0$ and $W_{OPT} \stackrel{P}{\to} 0$ but $\frac{W(\hat{\delta}_{\hmu})}{W_{OPT}}\stackrel{P}{\to} \infty$.
	\item[(2)] For $\frac{C_p}{p/n}\to c$ with $0<c<\infty$,
	\begin{itemize}	
	\item[(i)] if $\frac{p}{n}\to C <\infty$, then $lim_P W(\hat{\delta}_{\hmu})>1-\Phi(\frac{\sqrt{C_0}}{2})$;
	\item[(ii)] if $\frac{p}{n} \to \infty$, then $W(\hat{\delta}_{\hmu})\stackrel{P}{\to} 0$ and $W_{OPT} \stackrel{P}{\to} 0$, but $\frac{W(\hat{\delta}_{\hmu})}{W_{OPT}}\stackrel{P}{\to} \infty$.
	\end{itemize}			
	\item[(3)] For $\frac{C_p}{p/n}\to 0$, then $W(\hat{\delta}_{\hmu})\stackrel{P}{\to} \half$.	
	\end{itemize}
\end{thm}
\begin{proof}
See Supplementary Materials.
\end{proof}

\begin{corollary}\label{CorErrp>n} With all conditions are  same as in Theorem \ref{ThErrp>n}, and if $n_1=n_2$, then
\begin{itemize}
\item[(1)] If $\frac{C_p}{\sqrt{p/n}}\to \infty$, then $W(\hat{\delta}_{\hmu})\stackrel{P}{\to} 0$ and $W_{OPT} \stackrel{P}{\to} 0$, but $\frac{W(\hat{\delta}_{\hmu})}{W_{OPT}}\stackrel{P}{\to} \infty$;
\item[(2)] If $\frac{C_p}{\sqrt{p/n}}\to c$ with $0<c<\infty$,
\begin{itemize}
\item[(i)] If $\frac{p}{n}\to C$, then $W(\hat{\delta}_{\hmu})\stackrel{P}{\to} 1-\Phi(\frac{c}{2\sqrt{4+c/\sqrt{C}}})$ and $W_{OPT}\stackrel{P}{\to} 1-\Phi(\frac{\sqrt{c\sqrt{C}}}{2})$
\item[(ii)] If $\frac{p}{n}\to \infty$, then $W(\hat{\delta}_{\hmu})\stackrel{P}{\to} 1-\Phi(\frac{c}{4})$, and $W_{OPT}\stackrel{P}{\to} 0$;
\end{itemize}
\item[(3)] If $\frac{C_p}{\sqrt{p/n}}\stackrel{P}{\to} 0$, we have $W(\hat{\delta}_{\hmu})\stackrel{P}{\to} \half$.
\end{itemize}
\end{corollary}

Theorem \ref{ThErrp>n} and the Corollary \ref{CorErrp>n} show that while $p/n\to C$ with $0<C\le\infty$, $\hat{\delta}_{\hmu}$ is never asymptotically optimal. It is asymptotically sub-optimal only if $C_p\to \infty$. It reveals that though there's no difficulty in applying LDA on spatial dependent data while estimating the parameters by MLE, however, in high dimensional case ($p/n\to C$ with $0<C\le \infty$), the classification performance may be poor due to noise accumulated in the estimation of $\bmu_1$ and  $\bmu_2$ (see \cite{fan2008high} and \cite{shao2011sparse}). Therefore, feature selection is still critical for classification with high dimension. \cite{fan2008high} seeks to extract salient features by two-sample t-test and proved that t-test can pick up all important features by choosing an appropriate critical value once the features are assumed to be independent. \cite{shao2011sparse} proposes to select features by threshold. For the spatially correlated features, we can use penalized maximum likelihood estimates (PMLE) for feature selection.

\section{Classification using penalized maximum likelihood estimate (PMLE-LDA)}
\label{sec-PLDA}

\subsection{The penalized maximum likelihood estimation (PMLE) }
\label{sec_PMLE_intr}
{
In this section, we consider feature selection for the high dimensional classification problem (i.e. $p/n \to C$ with $0<C\le \infty$ as $p\to \infty$ and $n\to \infty$). We recall the notation used in section \ref{sec_introSpatialLDA}. We use $(\Delta_1,...,\Delta_2)=(\mu_{21}-\mu_{11},...,\mu_{2p}-\mu_{1p})$ to denote the differences of the mean between class $\mathcal{C}_1$ and $\mathcal{C}_2$. The vector form is denoted as $\bDelta=\bmu_2-\bmu_1$, which is a $p$ dimensional vector. Define the signal set $S=\{j: \Delta_j\neq 0\}$. Let $s$ be the number of non zero elements in $\bDelta$. The important features are contained in the set $S$. Instead of assuming the sparsity of discriminant direction \citep{cai2011direct,fan2012road, mai2012direct}, we assume the sparsity of feature difference $\bDelta$ (i.e. $s\ll n$ and $s/n\to 0$). Next, we derive the penalized likelihood function based on the assumption that the observations $\bY_{ki}$ are normally distributed $\bY_{ki}\sim N(\bmu_k,\Sigma(\btheta_0))$ for $k=1,2$ and $i=1,2,...,n_k$.}

{First, we define two matrix forms that help simplify notations in subsequent derivations. Let $I_p$ be a $p\times p$ identity matrix. We denote the diagonal block matrix for square matrix $A$ as $diag_n(A)$. We denote the block matrix for identity matrix $I_p$ as $\tilde{J}_{n,p}$. Their definitions are given as follows.}
\begin{align*}
diag_{n}(A) = \underbrace{\left(   \begin{array}{cccc}
A   &  0     &  \cdots  & 0\\
0     &  A   &  \cdots  & 0\\
\vdots& \vdots &  \ddots  & \vdots\\
0     & 0      &  0       &A
\end{array}
\right)}_{n\times n \ blocks},
\  \
\tilde{J}_{n,p} = \underbrace{ \left(\begin{array}{cccc}
I_p   &  I_p     &  \cdots  & I_p\\
I_p   &  I_p     &  \cdots  & I_p\\
\vdots& \vdots   &  \ddots  & \vdots\\
I_p   &  I_p     &  I_p     &I_p
\end{array}
\right)}_{n\times n\ blocks}.
\end{align*}
Both $diag_n(A)$ and $\tilde{J}_{n,p}$ consist of $n\times n$ blocks. Thus we have $\tilde{I}_{n,p}=diag_n(I_p)$.

{
Recall that $\bY=(\bY_{11}^T,\cdots, \bY_{1n_1}^T,\bY_{21}^T,\cdots,\bY_{2n_2}^T)^T$ is a $np$ dimensional vector. In order to estimate $\bDelta=\bmu_1-\bmu_2$, we transform $\bY$ by letting $\bZ=\textbf{V}\bY$, where $\textbf{V}$ is a $(n-1)p\times np$ matrix made up of the first $(n-1)p$ rows of $\tilde{I}_{n,p}-\frac{1}{n}\tilde{J}_{n,p}$. Then $\bZ=(\bZ_1^T\ \bZ_2^T\cdots \bZ_{n-1}^T)^T$, where $\bZ_i=\bY_{1i}-\bar{\bY}$ for $i=1,2,...,n_1$, $\bZ_i=\bY_{2(i-n_1)}-\bar{\bY}$ for $i=n_1+1,n_1+2,...,n-1$ and $\bar{\bY}=\frac{1}{n}\sum_{k=1}^{2}\sum_{i=1}^{n_k}\bY_{ki}$. Note that transformed data $\bZ$ becomes $(n-1)p$ dimensional instead of $np$ dimensional, because it is known that the freedom of a centered transformation is $n-1$ (if performed on $n$ observations). Then the distribution of $\bZ$ is given as} 
$$
\bZ_i\sim\left\{
\begin{array}{l}
N(-\tau_2\bDelta, \frac{n-1}{n}\Sigma(\btheta_0)), i=1,2,\ldots,n_1, \\
N( \tau_1\bDelta, \frac{n-1}{n}\Sigma(\btheta_0)), i=n_1+1,\ldots,n-1.\\
\end{array}\right.
$$
where $\tau_1=\frac{n_1}{n}$ and $\tau_2=\frac{n_2}{n}$. The covariance is $cov(\bZ_i,\bZ_j)=-\frac{1}{n}\Sigma$ for $i\neq j$. We define $\bX^{(1)}$ and $\bX^{(2)}$ as
\begin{align*}
\bX^{(1)}=\left(
\begin{array}{cccc}
-\tau_2   &  0     &  \cdots  & 0\\
0     &  -\tau_2    &  \cdots  & 0\\
\vdots& \vdots &  \ddots  & \vdots\\
0     & 0      &  0       &-\tau_2
\end{array}
\right)_{p\times p},
\ \ \
\bX^{(2)}=\left(
\begin{array}{cccc}
\tau_1   &  0     &  \cdots  & 0\\
0     &  \tau_1    &  \cdots  & 0\\
\vdots& \vdots &  \ddots  & \vdots\\
0     & 0      &  0       &\tau_1
\end{array}
\right)_{p\times p}.
\end{align*}
We further define that $\bX_i=\bX^{(1)}$ for $i=1,2,...,n_1$ and $\bX_i=\bX^{(2)}$ for $i=n_1+1,,...,n-1$. Then we have $\bX=(\bX_1^T, \bX_2^T,\cdots,\bX_{n-1}^T)^T$.

{Because PMLE is traditionally used in a linear regression setup, we rewrite $\bbeta=\bDelta$. This keeps the consistency with the traditional notations, which researchers have been familiar with.}  Then the $(n-1)p\times 1$ vector $\bZ$ as a multivariate normal distribution $N(\bX\bbeta,\dot{\Sigma})$, where $\dot{\Sigma}=(\tilde{I}_{n-1,p}-\frac{1}{n}\tilde{J}_{n-1,p})diag_{n-1}(\Sigma)$. Denote all the unknown parameters by $\boldeta=(\bbeta,\btheta)\in \IR^{p+q}$. Based on the fact that
$\abs{\dSigma}=\abs{\tilde{I}_{n-1,p}-\frac{1}{n}\tilde{J}_{n-1,p}}\abs{diag_{n-1}\Sigma(\btheta))}=(\frac{1}{n})^p\abs{\Sigma(\btheta)}^{n-1}$ and $(\tilde{I}_{n-1,p}-\frac{1}{n}\tilde{J}_{n-1,p})^{-1}=\tilde{I}_{n-1,p}+\tilde{J}_{n-1,p}$, we can write the penalized log-likelihood function of $\bbeta$ and $\btheta$ as
\begin{align}\label{PMLE_G}
Q(\btheta,\bbeta;\bZ)
=&-\frac{np}{2}log(2\pi)-\half log\abs{\dSigma}-\half(\bZ-\bX\bbeta)^T\dot{\Sigma}^{-1}(\bZ-\bX\bbeta)-n\sum_{j=1}^{p}P_{\lambda}(|\beta_j|) \nonumber \\
=&C_{n,p}-\frac{n-1}{2}log\abs{\Sigma}-\half(\bZ-\bX\bbeta)^T diag_{n-1}({\inSigma}) (\tilde{I}_{n-1,p}+\tilde{J}_{n-1,p})(\bZ-\bX\bbeta) \nonumber \\
&-n\sum_{j=1}^{p}P_{\lambda}(|\beta_j|),
\end{align}
where $C_{n,p}=-\frac{(n-1)p}{2}\log\pi+\frac{p}{2}\log n$. $P_{\lambda}(x)$ is a generic sparsity-inducing penalty, which could be the lasso penalization or folded concave penalization (such as the SCAD and the MCP). We will elaborate the choice of penalization later in this paper.

{By observing the joint likelihood \ref{PMLE_G}, we can see that $\btheta$ and $\bbeta$ play different roles because one is included in the mean and the other is included in the covariance. So it is difficult to obtain the estimation of them simultaneously. The exact solution of \ref{PMLE_G} should be achieved through numerous iterations before convergence. However, to save computational time, we adopt the one-step estimation procedure to estimate $\btheta$ and $\bbeta$ though iterative updates \citep{chu2011penalized}. The procedure is shown as follows.
\begin{center}
\colorbox[gray]{0.95}{
\begin{minipage}{0.85\textwidth}
\textbf{One-step PMLE (PMLE$_{ose}$) computing procedure:}
\begin{enumerate}[1.]
\itemsep0em
\item Initialize $\bbeta$ by minimizing $R(\bbeta)=(\bZ-\bX\beta)^T(\bZ-\bX\beta)+n\sum_{j=1}^{p}P_{\lambda}(|\beta_j|)$ with respect to $\bbeta$. Denote the initialization by $\hbeta^{(0)}$;
\item With $\bbeta=\hbeta^{(0)}$, estimate $\btheta$ by maximizing $Q(\btheta,\hbeta^{(0)};\bZ)$ in (\ref{PMLE_G}) with respect to $\btheta$. Denote the estimate by $\htheta^{(0)}$;
\item With $\btheta=\htheta^{(0)}$, update $\bbeta$ by maximizing $Q(\htheta^{(0)},\bbeta;\bZ)$ in (\ref{PMLE_G}) with respect to $\bbeta$. Denote the estimate by $\hbeta^{(1)}$;
\item With $\bbeta=\hbeta^{(1)}$, estimate $\btheta$ by maximizing $Q(\btheta,\hbeta^{(1)};\bZ)$ in (\ref{PMLE_G}) with respect to $\btheta$. Denote the estimate by $\htheta^{(1)}$.
\end{enumerate}
\end{minipage}}
\end{center}
}
Then $\htheta_{ose}=\htheta^{(1)}$ and $\hbeta_{ose}=\hbeta^{(1)}$ are the obtained estimates. We call $\htheta_{ose}$ and $\hbeta_{ose}$ as the one-step PMLE.
Mean parameters $\bmu_1$ and $\bmu_2$ can be estimated by $\hat{\bmu}_{1,ose}=\bar{\bY}-\tau_2 \hbeta_{ose}$ and $\hat{\bmu}_{2PMLE}=\bar{\bY}+\tau_1 \hbeta_{ose}$. Besides, we estimate the covariance as $\hat{\Sigma}=\Sigma(\htheta_{ose})$. The $(i,j)$th element of $\hSigma$ is 
$
\hat{\sigma}_{i,j}=\gamma(h_{ij};\htheta_{ose}),
$
where $h_{ij}=\norm{s_j-s_i}{2}$ is the Euclidean distance between site $s_i$ and $s_j$.

\subsubsection{Consistency of {one-step PMLE}}
Penalty function largely determines the sampling properties of the penalized likelihood estimates. Some additional assumptions about the penality function and tuning parameter $\lambda$ are needed:

\begin{asu}\label{asu_penfunc1}
Assume $a_{n}=O_p(\frac{1}{\sqrt{n}})$, where $a_{n}=\max_{1\le j\le p}\{p^{'}_{\lambda_{n}}(|\beta_{0j}|),\beta_{0j}\ne 0\}$
\end{asu}

\begin{asu}\label{asu_penfunc2}
$b_{n}\to 0$ as $n\to \infty$, where $b_{n}=\max_{1\le j\le m}\{p^{''}_{\lambda_{n}}(|\beta_{0j}|),\beta_{0j}\ne 0\}$
\end{asu}

\begin{asu}\label{asu_sparsity_1}
$\lambda_{n}\to 0$ and $\lambda_{n} /\sqrt{\frac{s}{n}}\to \infty$.
\end{asu}

\begin{asu}\label{asu_sparsity_2}
$\lim\inf_{\substack{n\to \infty\\p\to \infty}}\lim\inf_{\theta\to 0+}P_{\lambda_{n}}^{'}(|\theta|)/\lambda_{n}>0$
\end{asu}

A \ref{asu_penfunc1} ensures the unbiasedness property for large parameters and the existence of the consistent penalized likelihood estimator. A \ref{asu_penfunc2} ensures that the penalty function does not influence the penalized likelihood estimators more than the likelihood function itself. A\ref{asu_sparsity_2}  ensures the penalized likelihood estimators possess the sparsity property. A\ref{asu_sparsity_1} leads to the variable selection consistency.

Smoothly Clipped Absolute Deviation (SCAD) penalty satisfies all these assumptions. We adopt SCAD penalization in this paper. \cite{fan2001variable} proposed the SCAD penalty function and claimed that it has three good properties: unbiasedness, sparsity and continuity. Unbiasedness means that there is no over-penalization of large features to avoid unnecessary modeling biases. Sparsity means that the insignificant parameters are set to 0 by a thresholding rule to reduce model complexity. Continuity means that the penalized likelihood produces continuous estimators. The SCAD penalty function is defined as
\begin{align*}
p_{\lambda}(\beta)=
\begin{cases}
\lambda \abs{\beta} & \text{if } \abs{\beta}\le \lambda \\
-\frac{\beta^2-2a\lambda \beta + \lambda^2}{2(a-1)} & \text{if } \lambda<\abs{\beta}\le a\lambda \\
\frac{(a+1)\lambda^2}{2}  & \text{if } \abs{\beta}>a\lambda
\end{cases}
\end{align*}
for some $a>0$. More details can be found in \cite{fan2001variable}. We first illustrate the property of PMLE of (\ref{PMLE_G}) by the following theorem.

{As demonstrated in \cite{zou2008one}, the one-step method is as efficient as the fully iterative method both empirically and theoretically, provided that the initial estimators are reasonably good. We will see in the proof of Theorem \ref{Th_consist_PMLE}, in the one step estimation Algorithm stated in section \ref{sec_PMLE_intr}, the initial estimators for $\bbeta$ ($\hbeta^{(0)}$) is obtained by minimizing penalized regression function $R(\bbeta)=(\bZ-\bX\beta)^T(\bZ-\bX\beta)+n\sum_{j=1}^{p}P_{\lambda}(|\beta_j|)$, which is consistent and has oracle property. The initial estimator for $\btheta$ ($\htheta^{(0)}$) is also a consistent estimate from MLE. As a result, the one-step PMLE (PMLE$_{ose}$) has good property as demonstrated in Theorem \ref{Th_consist_PMLE}.  }

Recall that the true parameter $\bbeta_0$ is a parameter vector of size $p$, and $\btheta_0=(\theta_{01},\theta_{02},...,\theta_{0q})$ is a $q-$dimensional parameter in covariance function. We define the sparsity of $\bbeta_0$ as follows. Without loss of generality, we can write $\bbeta_0=(\beta_{1,0}^T,\beta_{2,0}^T)^T$, where $\beta_{1,0}\in\IR^{s}$ stands for non-zero components, and $\beta_{0,2}=\textbf{0}_{(p-s)\times 1}$ stands for zero components. The number of nonzero components suffices that $\frac{s}{n}\to 0$ as $n, p, s \rightarrow \infty$. So
we can write $\bX_i=(\bX^1_{i},\bX^2_{i}), i=1,2,..,n$, where $\bX^1_{i}$ is the $p\times s$ submatrix of $\bX_i$ made up of nonzero columns in $supp(\bbeta_0)$ and $\bX^2_{i}$ is the $p\times (p-s)$ complement matrix. We have the following Theorem for PMLE$_{ose}$:

\begin{thm}\label{Th_consist_PMLE}
Assume conditions A \ref{asu_covfunc}-A \ref{asu_sparsity_2} hold. Assume $\bbeta_0=(\bbeta_{1,0}^T,\bbeta_{2,0}^T)^T$, where $\bbeta_{1,0}\in\IR^{s}$ is non-zero component, $\bbeta_{2,0}=\textbf{0}_{(p-s)\times 1}$ is the zero component of $\bbeta_0$ with $\frac{s}{n}\to 0$, $\frac{p}{n}\to C$ with $0<C\le\infty$ as $n, p, s \rightarrow \infty$. The one-step PMLE of (\ref{PMLE_G}) from the one step procedure in section  \ref{sec_PMLE_intr} (PMLE$_{ose}$) is $\hat{\boldeta}_{ose}=(\hbeta_{ose},\htheta_{ose})$  with $\hbeta_{ose}=(\hbeta_{1,ose}^T,\hbeta_{2,ose}^T)^T$ and $\hbeta_{1, ose}$ is a sub-vector of $\hbeta_{ose}$ formed by nonzero components in $supp(\bbeta_0)$. Then $\hat{\boldeta}_{ose}$ satisfy:
\begin{itemize}
\item[a. \textbf{(consistency)}] $\norm{\htheta_{ose}-\btheta_0}{2}=O_p(\frac{1}{\sqrt{np}})$ and $\norm{\hbeta_{ose}-\bbeta_0}{2}=O_p(\sqrt{\frac{s}{n}})$.
\item[b. \textbf{(sparsity)}] $\hbeta_{2, ose}=0$ with probability tending to $1$ as $n\to \infty$.
\end{itemize}
\end{thm}
\begin{proof}
See Supplementary Materials.
\end{proof}

\subsubsection{Covariance tapering and one-step PMLE}
\label{sec_tapper}
When the number of features is large ($p$ is large) for each realization of the spatial process, calculating the likelihood can be computationally infeasible (requiring $\mathcal{O}(p^3)$ calculation). Covariance tapering can be used to approximate the likelihood. When the covariance matrix is replaced with a tapered one, the resulting matrices can then be manipulated using efficient sparse matrix algorithms which would reduce computational burden effectively.

In section \ref{sec_introSpatialLDA}, the covariance matrix is defined as $\Sigma(\btheta)=[\gamma(s_i,s_j)]_{i,j=1}^{p}$. Under A\ref{asu_covfunc}, we can simply write it as $\Sigma=[\gamma(h_{ij})]_{i,j=1}^{p}$, where $h_{ij}=\norm{s_i-s_j}{2}$ is the Euclidean distance between sites $s_i$ and $s_j$. Let $K_T(h,w)$ denote a tapering function, which is an isotropic autocorrelation function when $0<h<w$ and 0 when $h\geq w$ for a given threshold $w>0$.  We use a simple tapering function from \cite{wendland1995piecewise},
\begin{align}\label{tapperfunction}
K_T(h,\omega)=\left[(1-h/w)_{+}\right]^2
\end{align}
where $x_{+}=\max({x,0})$ meaning that the correlation is $0$ when the lag distance $h$ is greater than the threshold distance $w$. Let $\textbf{K}(w)=[K_{T}(h_{ii{'}},w)]_{i,i^{'}=1}^{p}$ denote the $p\times p$ tapering matrix. Then a tapered covariance of $\Sigma$ is defined as $\Sigma_{T}=\Sigma \circ \textbf{K}(w)$, where $\circ$ is the Schur product (i.e. elementwise product). By the properties of the Schur product (\cite{neumaier1992horn}, chap. 5), the tapered covariance matrix would keep the positive definiteness thus it is still a valid covariance matrix. When $p$ is large, we approximate the  penalized log-likelihood (\ref{PMLE_G}) by replacing $\Sigma$ with $\Sigma_T$ and obtain a covariance tapered penalized log-likelihood:

\begin{align}\label{penloglik_g_taper}
Q_T(\btheta,\bbeta;\bZ)
=&-\frac{np}{2}log(2\pi)-\half log\abs{\dSigma_T}-\half(\bZ-\bX\bbeta)^T\dot{\Sigma}_T^{-1}(\bZ-\bX\bbeta)-n\sum_{j=1}^{p}P_{\lambda}(|\beta_j|)\nonumber\\
=&C_{n,p}-\frac{n-1}{2}log\abs{\Sigma_T}-\half(\bZ-\bX\bbeta)^T diag_{n-1}({\inSigma_T}) (\tilde{I}_{n-1,p}+\tilde{J}_{n-1,p})(\bZ-\bX\bbeta)\nonumber\\
&-n\sum_{j=1}^{p}P_{\lambda}(|\beta_j|)
\end{align}
where $C_{n,p}=-\frac{(n-1)p}{2}\log\pi+\frac{p}{2}\log n$.

We keep all the notations the same as in (\ref{PMLE_G}), except that $\Sigma$ is replaced by $\Sigma_T$. {We follow the one-step PMLE procedure.} Let $\hbeta_{T, ose}=\hDelta_{T, ose}$ and $\htheta_{T, ose}$ be the one-step PMLE with tapered covariance (PMLE$_{T, ose}$). Next, we prove the consistency of PMLE$_{T, ose}$.
Let $\gamma_k(\btheta,h)=\der{\gamma(\btheta,h)}{\theta_k}(\btheta)$ and $\gamma_{jk}(\btheta,h)=\frac{\partial^2\gamma(\btheta,h)}{\partial \theta_k\theta_j}$. Two additional assumptions are made here for regularization.

\begin{asu}\label{asu_tapperRange}
Assume $0 <\inf_{p}\{\frac{w_{p}}{p^{\delta}}\}< \sup_{p}\{\frac{w_{p}}{p^{\delta}}\}  <\infty$, where $w_{p}$ is the threshold distance in the tapering function for some $\delta>0$.
\end{asu}
\begin{asu}\label{asu_covfunc_int}
Let $d\ (d\ge1)$ be the dimension of the domain, i.e. $D\subset \IR^d$. Assume for all $\btheta\in \Xi$ and $1\le k,j\le q$, we have $ \gamma(\btheta,h), \gamma_k(\btheta,h), \gamma_{jk}(\btheta,h)$ belong to the function space $\pounds$, where $\pounds=\{f(h):\int_{0}^{\infty}h^d f(h)dh<\infty\}.$
\end{asu}
Let $\Sigma$ be the covariance matrix and $\Sigma_T$ be the tapered covariance matrix. $\Sigma_{k,T}=\der{\Sigma_T}{\theta_k}$ and $\Sigma_{jk,T}=\frac{\partial^2 \Sigma_T}{\partial \theta_j \partial \theta_k}$. By using the tapering function (\ref{tapperfunction}), we have the following result for {PMLE$_{T,ose}$}.

{Similar to Theorem \ref{Th_consist_PMLE}, since the initial estimates have good properties, the one-step PMLE with tapering (PMLE$_{T,ose}$) also has good properties hence we have the following theorem:
\begin{thm}\label{Th_consist_PMLE_T}
Assume conditions \ref{asu_covfunc}-\ref{asu_covfunc_int} hold. Assume $\bbeta_0=(\bbeta_{1,0}^{T},\bbeta_{2,0}^{T})^{T}$, where $\bbeta_{1,0}\in\IR^s$ is non-zero component, $\bbeta_{2,0}=\textbf{0}_{(p-s)\times 1}$ is the zero component of $\bbeta_0$ with $\frac{s}{n}\to 0$, $\frac{p}{n}\to C$ with $0<C\le\infty$ as $n,p,s\to \infty$. The one-step PMLE estimates of (\ref{penloglik_g_taper}) from Algorithm in section \ref{sec_PMLE_intr} (PMLE$_{T,ose}$) is $\hat{\boldeta}_{T, ose}=(\hbeta_{T,ose},\htheta_{T,ose})$  with $\hbeta=(\hbeta_{1,T,ose}^T,\hbeta_{2,T,ose}^T)^T$ and $\hbeta_{1,T,ose}$ is a sub-vector of $\hbeta_{T,ose}$ formed by non-zero components in $supp(\bbeta_0)$. Then $\hat{\boldeta}_T$ satisfy:
\begin{itemize}
\item[a. \textbf{(consistency)}] $\norm{\htheta_{T,ose}-\btheta_0}{2}=O_p(\frac{1}{\sqrt{np}})$ and $\norm{\hbeta_{T,ose}-\bbeta_0}{2}=O_P(\sqrt{\frac{s}{n}})$.
\item[b. \textbf{(sparsity)}] $\hbeta_{2,T, ose}=0$ with probability tending to $1$ as $n\to \infty$.
\end{itemize}
\end{thm}
}

\subsection{The penalized maximum likelihood estimation LDA (PMLE-LDA) classifier}
Now we can develop the PMLE-LDA classifier in this section. No matter using PMLE$_{ose}$ or PMLE$_{T, ose}$, we obtain the consistent estimates for $\bDelta$ and $\btheta$ denoted by $\hDelta$ and $\hat{\btheta}$. The estimation of $\bmu_1$ and $\bmu_2$ are $\hat{\bmu}_{1}=\bar{\bY}-\tau_2 \hat{\bDelta}$ and $\hat{\bmu}_{2}=\bar{\bY}+\tau_1 \hat{\bDelta}$. Besides, we have estimated covariance $\hat{\Sigma}=\Sigma(\htheta)$, where the $(i,j)$th element of $\hSigma$ is:
\begin{align}
\hat{\sigma}_{ij}=\gamma(\abs{s_j-s_i};\htheta)
\end{align}

When $p>n$, the error accumulated in estimate of each $\hat{\sigma}_{ij}$ may also cause problems in classification (see \cite{bickel2008covariance} and \cite{shao2011sparse}). For regularization of the covariance matrix, we use the tapered covariance matrix in classification function. Specifically, we define $\tSigma=\Sigma_T(\htheta)=\Sigma(\htheta) \circ \textbf{K}(w)$, where $\textbf{K}(w)$ is defined in section \ref{sec_tapper}. We then replace $\bmu_1,\bmu_2, \Sigma$ in LDA (\ref{BayesRule}) by $\hat{\bmu}_{1},\hat{\bmu}_{2}$ and $\tSigma$ for classification. Then the PMLE-LDA function is:
\begin{align}\label{PLDA}
\hat{\delta}_{PLDA}(\bX)=(\bX-\bar{\bY}-\frac{n_1-n_2}{2n}\hDelta)^T\inTSigma \hDelta
\end{align}
where $\bar{\bY}=\frac{1}{n}\sum_{k=1}^{2}\sum_{i=1}^{n_k}\bY_{ki}$.

The conditional misclassification rate for class 1 and class 2 are defined by (\ref{W1}) and (\ref{W2}) with $\hSigma$ replaced with $\tSigma$. Similarly we have the overall misclassification rate defined in (\ref{W}).

We need more assumptions for the covariance function $\gamma(h;\btheta)$ in Theorem \ref{ThErrPgeN}.
\begin{asu}\label{asu_covfunc_int2}
Let $d\ (d\ge 1)$ be the dimension of the domain $D$, i.e. $D\subset R^d$. Assume $\int_{1}^{\infty}h^d\gamma(h;\btheta)dh<\infty$ and $\int_{0}^{1}h^{d-1}\gamma(h;\btheta)dh<\infty$ for $\btheta\in \Xi$.
\end{asu}
 This requires that when  $h\to\infty$, $\gamma(h;\btheta) \sim h^{x} $ with $x<-(d+1)$ and  when $h\to 0$, $\gamma(h;\btheta)\sim h^x$ with $x>-d$.

\begin{asu}\label{asu_covfunc_bound}
Assume there exist a constant $M$ such that for any $h\ge 0$ and $\btheta\in\Xi$, $\parallel \frac{\partial \gamma(h;\btheta)}{\partial \btheta} \parallel_2\le M$.
\end{asu}
\begin{thm}\label{ThErrPgeN}
Assume $\htheta$, $\hDelta$ in (\ref{PLDA}) are estimated from Theorem \ref{Th_consist_PMLE} or \ref{Th_consist_PMLE_T}. Suppose assumptions A\ref{asu_covfunc}-A\ref{asu_eigen_sig} and A\ref{asu_covfunc_int2}-\ref{asu_covfunc_bound} hold. Assume $\frac{s}{n}\to 0$, $\frac{p}{n}\to C$ with $0<C\le \infty$, $C_p \to C_0$ with $0\le C_0\le \infty$, $\frac{C_p}{\sqrt{s/n}}\to 0$. Also, assume $w=O((\sqrt{np})^{\frac{\alpha}{d}})$ with $0<\alpha<1$, and $w^{-1}=O(p^{-\delta})$ with some $\delta>0$, where $d$ is the dimension of the domain. Then the classification error rate of $\hat{\delta}_{PLDA}$ is asymptotically sub-optimal, i.e. $W(\hat{\delta})\stackrel{P}{\to} 1-\Phi(\frac{\sqrt{C_0}}{2})$. Moreover,
\begin{itemize}
\item[(1)] If $C_p\to C_0<\infty$, $W(\hat{\delta})$ is asymptotically optimal, i.e. $\frac{W(\hat{\delta})}{W_{OPT}}\stackrel{P}{\to} 1$;
\item[(2)] If $C_p\to \infty$ and $C_p\kappa_{n,p}\to 0$, $W(\hat{\delta})$ is asymptotically optimal, i.e. $\frac{W(\hat{\delta})}{W_{OPT}}\stackrel{P}{\to} 1$,  where $\kappa_{n,p}=\max(\frac{w^d}{\sqrt{np}},\frac{1}{w},\sqrt{\frac{s}{n}})$.
\end{itemize}

\end{thm}

\begin{proof}
See Supplementary Materials.
\end{proof}

Since $W_{OPT}=1-\Phi(\frac{\sqrt{C_p}}{2})\to 1-\Phi (\frac{\sqrt{C_0}}{2})$, Theorem \ref{ThErrPgeN} shows that with moderate conditions, the error rate of the proposed classifier goes to the unique limit $W_{OPT}$. Moreover, if $C_p\to C_0<\infty$ or $C_p$ goes to $\infty$ with a moderate rate, then $W(\hat{\delta})$ goes to $1-\Phi (\frac{\sqrt{C_0}}{2})$ with the same rate as $W_{OPT}$.

\section{Numerical Study}
\label{sec-simulation}
We conduct extensive simulation study to evaluate the performance of our proposed method compared to available generic procedures. Then we illustrate the methodology with real neuroimgaing data from ADNI.

\subsection{Simulation Analysis}
Assume that the spatial domain of interest $D$ in $\mathbb{R}^2$ is a $u\times u$ square area. We can observe signal at each lattice. Then we have $p=u \times u$ features for classification. The mean effects of the signal for class $\mathcal{C}_1$ and $\mathcal{C}_2$ are $\bmu_1$ and $\bmu_2$. We assume that $\bmu_1=(\textbf{1}_{10},\textbf{0}_{p-10})$ and $\bmu_2=\textbf{0}_p$, where $\textbf{1}_k$ is a $k$ dimension vector with all the elements equal to $1$ and $\textbf{0}_k$ is a $k$ dimension vector with all the elements equal to $0$. For example, if $u=4$ (hence $p=16$), then the corresponding spatial domain $D$, and the mean structure of $\bmu_1$ and $\bmu_2$ are shown as in Figure \ref{fig_domain}. In the simulation setting, we construct various simulation scenarios by letting $u=6,\ 20,\ $ and $35$, making $p=36,\ 400,\ 1225$ respectively.

 \begin{figure}
 \begin{scriptsize}
       \minipage{0.33\textwidth}
      \begin{tikzpicture}
        \draw(0,0) grid (4,4);
        \node at (0.5,0.5) {\large ${1}$};
        \node at (0.5,1.5) {\large ${2}$};
        \node at (0.5,2.5) {\large ${3}$};
        \node at (0.5,3.5) {\large ${4}$};
        \node at (1.5,0.5) {\large ${5}$};
        \node at (1.5,1.5) {\large ${6}$};
        \node at (1.5,2.5) {\large ${7}$};
        \node at (1.5,3.5) {\large ${8}$};
        \node at (2.5,0.5) {\large ${9}$};
        \node at (2.5,1.5) {\large ${10}$};
        \node at (2.5,2.5) {\large ${11}$};
        \node at (2.5,3.5) {\large ${12}$};
        \node at (3.5,0.5) {\large ${13}$};
        \node at (3.5,1.5) {\large ${14}$};
        \node at (3.5,2.5) {\large ${15}$};
        \node at (3.5,3.5) {\large ${16}$};
      \end{tikzpicture}
      \endminipage\hfill
       \minipage{0.33\textwidth}
       \begin{tikzpicture}
        \draw(0,0) grid (4,4);
        \node at (0.5,0.5) {\large $\color{red}{1}$};
        \node at (0.5,1.5) {\large $\color{red}{1}$};
        \node at (0.5,2.5) {\large $\color{red}{1}$};
        \node at (0.5,3.5) {\large $\color{red}{1}$};
        \node at (1.5,0.5) {\large $\color{red}{1}$};
        \node at (1.5,1.5) {\large $\color{red}{1}$};
        \node at (1.5,2.5) {\large $\color{red}{1}$};
        \node at (1.5,3.5) {\large $\color{red}{1}$};
        \node at (2.5,0.5) {\large $\color{red}{1}$};
        \node at (2.5,1.5) {\large $\color{red}{1}$};
        \node at (2.5,2.5) {\large ${0}$};
        \node at (2.5,3.5) {\large ${0}$};
        \node at (3.5,0.5) {\large ${0}$};
        \node at (3.5,1.5) {\large ${0}$};
        \node at (3.5,2.5) {\large ${0}$};
        \node at (3.5,3.5) {\large ${0}$};
      \end{tikzpicture}
    \endminipage\hfill  	
	 \minipage{0.33\textwidth}
      \begin{tikzpicture}
        \draw(0,0) grid (4,4);
        \node at (0.5,0.5) {\large $\color{red}{0}$};
        \node at (0.5,1.5) {\large $\color{red}{0}$};
        \node at (0.5,2.5) {\large $\color{red}{0}$};
        \node at (0.5,3.5) {\large $\color{red}{0}$};
        \node at (1.5,0.5) {\large $\color{red}{0}$};
        \node at (1.5,1.5) {\large $\color{red}{0}$};
        \node at (1.5,2.5) {\large $\color{red}{0}$};
        \node at (1.5,3.5) {\large $\color{red}{0}$};
        \node at (2.5,0.5) {\large $\color{red}{0}$};
        \node at (2.5,1.5) {\large $\color{red}{0}$};
        \node at (2.5,2.5) {\large ${0}$};
        \node at (2.5,3.5) {\large ${0}$};
        \node at (3.5,0.5) {\large ${0}$};
        \node at (3.5,1.5) {\large ${0}$};
        \node at (3.5,2.5) {\large ${0}$};
        \node at (3.5,3.5) {\large ${0}$};
      \end{tikzpicture}
     \endminipage
     \caption{ Two dimensional domain example. Left: 2D domain with p=$4\times4$; middle: $\bmu_1$; right: $\bmu_2$.}\label{fig_domain}
\end{scriptsize}
\end{figure}

{
For the spatial covariance, we generate the error terms from stationary and isotropic Gaussian process with zero mean. A widely used spatial covariance function $Mat\acute{e}rn$ covariance was defined in (\ref{matern}). We use a special case of $Mat\acute{e}rn$ covariance function when $\nu=\half$, which is the exponential covariance function. In the simulation, we set the variance scale as $\sigma^2=1$, the nugget effect as $c=0.2$ and the range parameter as $r=1, 2, ...,8,9$. Larger value of $r$ means longer range spatial dependency. Let $h$ be the Euclidean distance between two sites on the domain $D$. Specifically, on the domain $D\in\  \mathbb{R}^2$, the distance between site $i$ with coordinate $s_i=(x_i,y_i)$ and site $j$ with coordinate $s_j=(x_j,y_j)$ is $h_{ij}=\sqrt{(x_i-x_j)^2+(y_i-y_j)^2}$.}

We generate 100 groups of training sets with $n_1=n_2=30$ according to different setting of $\bmu_1,\ \bmu_2$ and $\Sigma(\btheta_0)$. For each training set, we estimate the parameters $\bmu_1$, $\bmu_2$ and $\btheta_0$ by MLE, tapered MLE, PMLE (penalized) and tapered PMLE. We also generate 100 groups of testing data sets with $n_1=n_2=100$ to test the classification performance. The average classification error rate was calculated from the 100 groups of testing data sets.

We name our classification method proposed in this paper as PMLE-LDA. For each choice of $p$, we compare the classification performance of PMLE-LDA with MLE-LDA, PREG-LDA, FAIR (Feature Annealed Independence Rule) and NB (Naive Bayes) and RPEC(Random-projection ensemble classification).
More specifically, PMLE-LDA is the classifier defined in (\ref{PLDA}); MLE-LDA uses $\hat{\bmu}_{1MLE}$, $\hat{\bmu}_{2MLE}$ and $\Sigma(\hat{\btheta}_{MLE})$ in LDA function for classification;  PREG-LDA uses $\hDelta=\hbeta^{(0)}$ and $\htheta=\htheta^{(0)}$  in LDA function, where $\hbeta^{(0)}$ and $\htheta^{(0)}$ are estimated in the first and second step in the procedure in section \ref{sec_PMLE_intr}.  NB \citep{bickel2004some} uses sample mean $\hat{\bmu}_1$, $\hat{\bmu}_2$ and diagonal of sample covariance $\hat{\Sigma}$ in LDA. This method is also known as independent rule(IR). FAIR \citep{fan2008high} assumes independence between variables and utilizes t-test for variable selection in NB. RPEC \citep{cannings2017random} is a very general method which is not designed for a specific classifier (e.g. LDA). It uses random projections to project the feature vectors from any classifier into a lower dimensional space. {To compare with the truth, we uses TRUE to denote that true mean $\bmu_1,\ \bmu_2$ and true covariance $\Sigma(\btheta_0)$ are used in LDA for classification.}

We recall the basic differences of these methods as follows. RPEC, FAIR and NB are the classification methods without considering spatial dependency, while MLE-LDA, PREG-LDA and PMLE-LDA are methods considering spatial dependency. MLE and NB are two methods without feature selection, while PREG-LDA, PMLE-LDA and FAIR are classification methods with feature selection. Moreover, PREG-LDA selects features by penalized regression without considering spatial dependency, while PMLE-LDA selects features by penalized maximum likelihood estimation with spatial dependency incorporated.

{
We also compared our method with four more methods, which are $l_1$-logistic regression, $l_1$-FDA, DSDA (Direct sparse discriminant analysis) and CATCH (Covariate-adjusted tensor classification in high-dimensions). $l_1$-logistic regression is one of the most basic and popular methods to solve a classification problem. $l_1$-LDA is proposed for penalizing the discriminant vectors in Fisher's discriminant problem~\citep{witten2011penalized}. DSDA generalizes classical LDA and formulates high-dimensional LDA into a penalized least squares problem~\citep{mai2012direct}. CATCH takes advantage of the tensor structure to significantly reduce the number of parameters and hence alleviate computation complexity~\citep{pan2019covariate}. 
}

{
The classification performance of all the methods is shown in Table \ref{TCexp36}. $r=1$, $r=5$ and $r=9$ means weak, moderate and strong spatial dependence respectively. Among all the methods, PMLE-LDA outperforms all the others. We have the following conclusions. First, when spatial dependency is weak ($r=1$), all the methods with or without spatial dependency do not have much difference. But when spatial dependency is strong ($r=9$), the methods with spatial dependency (MLE-LDA, PREG-LDA, PMLE-LDA) outperform the methods without spatial dependency (FAIR, NB, RPEC,l1-logistic, l1-LDA, DSDA and CATCH). Second, when the number of feature is small ($p=36$), the methods with or without feature selection have similar performance. But when number of feature is large ($p=400$ and $p=1225$), the methods with feature selection outperforms the methods without feature selection. PREG-LDA and PMLE-LDA outperforms MLE-LDA. FAIR outperforms NB. Third, CATCH outperforms l1-logistic, l1-LDA and DSDA, because it honors the tensor structure and preserves more information. But CATCH does not consider the spatial dependency, so the performance is not as good as PMLE-LDA. In the end, PMLE-LDA outperforms PREG-LDA, which implies the selection procedure considering spatial dependency outperforms feature selection without considering spatial dependency.}

The parameter estimation results are shown in Table \ref{TPexp}. It shows the parameters are all consistently estimated. We also compared the average number of variables selected from PMLE-LDA, PREG-LDA and FAIR in Table \ref{TVexp36}. The tunning parameter $\lambda$ in PREG-LDA and PMLE-LDA is selected by 10 fold cross validation by minimizing the classification error rate. Table \ref{TVexp36} shows that FAIR tends to select the fewest features. PMLE tends to select more features than PREG. But when spatial dependency is strong, PMLE produces smaller variance for feature selection and thus smaller misclassification rate.

Additionally, we investigate the performance of classification, parameter estimation and feature selection of tapered MLE-LDA, tapered PMLE-LDA and tapered PREG-LDA. Note that the tapering technique is applied in parameter estimation. The performances of classification and feature selection are similar with the ones without tapering (See Table \ref{TCexp36} and \ref{TVexp36}). However, the tapering technique estimate a larger range parameter $r$ when the spatial dependency is strong. This is consistent with the characteristic of tapering technique. To save space, these tables are omitted here but are available in \cite{li2018high}.
\begin{table}[]
\centering
\begin{scriptsize}
\caption{Comparisons of classification accuracy rate for simulations.}
\label{TCexp36}
\small
\resizebox{\columnwidth}{!}{
\begin{tabular}{c|ccccccccccc}
\hline
        & TRUE & MLE  & PREG & PMLE & FAIR & NB  &  RPEC  &l1-logistic & l1-LDA & DSDA & CATCH\\
\hline
        & p=36 &      &      &      & 	&		 &  &&&&       \\
\hline
r=1	&0.884(0.02)	&0.838(0.03)	&0.838(0.04)	&0.839(0.05)	&0.807(0.04)	&0.837(0.04)& 0.822(0.03) & 0.810(0.04) & 0.83(0.03) & 0.800(0.05) & 0.830(0.04)\\
r=5	&0.911(0.02)	&0.879(0.02)	&0.881(0.02)	&0.881(0.04)	&0.722(0.05)	&0.752(0.05)& 0.846(0.03) & 0.830(0.04) & 0.750(0.06)  & 0.830(0.04) & 0.840(0.04)\\
r=9	&0.936(0.02)	&0.913(0.02)	&0.915(0.02)	&0.917(0.02)	&0.709(0.05)	&0.746(0.05)& 0.883(0.03) & 0.870(0.03) & 0.750(0.08)  & 0.860(0.04) & 0.870(0.03)\\
\hline
       & p=400 &      &      &      & 	&		 &   &&&&     \\
\hline
r=1	&0.915(0.02)	&0.739(0.03)	&0.841(0.05)	&0.824(0.05)	&0.833(0.04)	&0.738(0.04)&0.728(0.04) & 0.800(0.04) & 0.740(0.04) & 0.790(0.05) & 0.830(0.04) \\
r=5	&0.953(0.01)	&0.814(0.03)	&0.895(0.05)	&0.924(0.03)	&0.743(0.04)	&0.598(0.05)& 0.678(0.04) & 0.770(0.05) & 0.590(0.06) & 0.790(0.05) & 0.780(0.06)\\
r=9	&0.971(0.01)	&0.863(0.02)	&0.926(0.04)	&0.955(0.02)	&0.716(0.05)	&0.575(0.05)& 0.707(0.04) &  0.820(0.05) & 0.560(0.06) & 0.830(0.05) & 0.820(0.04)\\
\hline
	& p=1225 &      &      &      & 	&		 &   &&&&      \\
\hline
r=1	&0.915(0.02)	&0.653(0.03)	&0.829(0.05)	&0.765(0.07)	&0.825(0.05)	&0.658(0.05) & 0.650(0.04) & 0.780(0.04) &0.660(0.04) & 0.790(0.05) & 0.810(0.04)\\
r=5	&0.951(0.02)	&0.715(0.03)	&0.858(0.07)	&0.902(0.04)	&0.741(0.05)	&0.556(0.05)& 0.578(0.04) & 0.730(0.05) & 0.530(0.04) & 0.740(0.05) & 0.730(0.04)\\
r=9	&0.968(0.01)	&0.761(0.03)	&0.888(0.07)	&0.930(0.03)	&0.718(0.05)	&0.539(0.05)&0.583(0.04) & 0.760(0.05) & 0.520(0.03) & 0.780(0.05) & 0.760(0.05) \\ 
\hline
\end{tabular}
}
\end{scriptsize}
\end{table}

\begin{table}[]
\centering
\begin{scriptsize}
\caption{Comparisons of parameter estimation for simulations. }
\label{TPexp}
\begin{tabular}{cc|c|cc|cc|cc}
\hline

        & &     &\multicolumn{2}{c}{p=36}&    \multicolumn{2}{c}{p=400}&         \multicolumn{2}{c}{p=1225} \\
\hline
        & &TRUE & MLE & PMLE                  & MLE & PMLE     & MLE&  PMLE\\
\hline
r=1	&r&	1	&1.01(0.16)	&1.04(0.16)	&1(0.05)	    &1(0.05)		    &1.00(0.02)	&1.00(0.01)\\
	&c&	0.2	&0.19(0.12)	&0.19(0.12)	&0.2(0.04)	&0.19(0.03)		&0.20(0.01)	&0.20(0.02)\\
	&$\sigma$&	1	&0.97(0.03)	&1(0.04)	&0.97(0.01)	&1(0.01) 	&0.97(0.01)	&1.00(0.01)\\
\hline										
r=5	&r&	5	&5.08(0.75)	&5.09(0.75)	&5.03(0.27)	&5.04(0.27)		&4.99(0.17)	&5.00(0.17)\\
	&c&	0.2	&0.2(0.03)	&0.2(0.03)	&0.2(0.01)	&0.2(0.01)		&0.20(0.004)	&0.20(0.004)\\
	&$\sigma$&	1	&0.97(0.08)	&1.01(0.08)	&0.97(0.03)	&1.01(0.04)&0.97(0.02)&1.00(0.02)\\
\hline																				
r=9	&r&	9	&9.17(1.58)	&9.1(1.57)	&9.1(0.67)	&9.12(0.69)		&8.96(0.45)	&8.97(0.45)\\
	&c&	0.2	&0.2(0.02)	&0.2(0.02)	&0.2(0.01)	&0.2(0.01)		&0.20(0.01)	&0.20(0.01)\\
	&$\sigma$&	1	&0.97(0.1)	&1.01(0.1)	&0.97(0.05)&1.01(0.05)&0.96(0.03)&1.00(0.03)\\
\hline																					
\end{tabular}
\end{scriptsize}
\end{table}

\begin{table}[]
\centering
\begin{scriptsize}
\caption{Comparisons of number of selected features for simulations.}
\label{TVexp36}
\begin{tabular}{c|cccccc}

\hline
        &\multicolumn{2}{c}{PMLE}&    \multicolumn{2}{c}{PREG}&\multicolumn{2}{c}{FAIR} \\
\hline
p=36      & selectedN  & correctN  & selectedN & correctN  & selectedN  & correctN  \\
\hline
r=1	&20.77(6.21)	&9.47(1.71)	&18.92(7.03)	&9.81(0.8)	&6.45(3.98)	&5.48(2.36)\\
r=5	&19.97(5.51)	&9.51(1.52)	&19.97(7.06)	&10(0)	&3.23(1.48)	&3.08(1.47)\\
r=9	&19.8(5.04)	&9.72(1.07)	&19.45(7.57)	&10(0)	&2.81(1.38)	&2.68(1.34)\\
\hline
p=400      &   &   &  &    &  & \\
\hline
r=1	&84.8(55.3)	&9.2(1.5)	&42(55.2)	&9.2(1.4)		    &20.8(15.1)	&7.2(2.1)\\
r=5	&73.7(32.2)	&9.8(0.6)	&50.7(73.3)	&9.5(1.2)			&11.3(10.6)	&5.3(2.8)\\
r=9	&52.5(21.9)	&9.9(0.5)	&63.7(98.5)	&9.7(1.1)			&6.8(7.6)	&4.1(2.7)\\
\hline
p=1225      &   &   &  &    &  & \\
\hline
r=1	&181.3(205.3)	&7.85(2.83)	&37.48(59.38)	&8.59(1.85)	&31.07(22.28)	&6.96(1.73)\\
r=5	&174.5(126.3)	&9.48(1.53)	&76.88(191.66)	&8.89(1.84)	&26.15(20.55)	&6.42(2.83)\\
r=9	&115.1(58.9)  	&9.81(0.88)	&79.21(192.62)	&9.1(1.81)	&15.28(14.02)	&5.28(3.41)\\
\hline
\end{tabular}
\end{scriptsize}
\end{table}

{
Finally, we investigate the simulation results when covariance is mis-specified. 
More specifically, we use Gaussian covariance function (i.e. $Mat\acute{e}rn$ covariance when $\nu\to\infty$) to generate the data. Then we use exponential covariance function (i.e. $Mat\acute{e}rn$ covariance when $\nu=\half$) to estimate the structure and complete classification. Both of them are $Mat\acute{e}rn$ covariance with different smoothness parameters. Table \ref{TCmiscov36} shows the classification performance if the covariance are misspecified. We generate the data using Gaussian covariance function with $\sigma^2=1,\ c=0.2$, and $r=1,2,...,9$ in exponential covariance function.
It shows that with misspecified covariance, the PMLE-LDA classification method has the best performance,  even when the spatial dependency is strong ($r=5$ and $r=9$). Therefore, the proposed method is robust to the mis-specification of covariance. }

\begin{table}[]
\centering
\begin{scriptsize}
\caption{Comparisons of classification accuracy rate for simulations when covariance is mis-specified.}
\label{TCmiscov36}
\begin{tabular}{c|ccccccc}
\hline
        & TRUE & MLE  & PREG & PMLE & FAIR & NB & RPEC     \\
\hline
        & p=36 &      &      &      & 	&		&          \\
\hline
r=1	&0.887(0.02)	&0.841(0.03)	&0.84(0.04)	&0.843(0.05)	&0.826(0.05)	&0.854(0.05)&0.830(0.03)\\
r=5	&0.939(0.02)	&0.922(0.02)	&0.922(0.02)	&0.925(0.03)	&0.709(0.05)	&0.735(0.05)&0.895(0.03)\\
r=9	&0.969(0.01)	&0.956(0.01)	&0.954(0.02)	&0.959(0.02)	&0.705(0.05)	&0.736(0.05)&0.939(0.02)\\
\hline
        & p=400 &      &      &      & 	& &\\
\hline
r=1	&0.911(0.02)	&0.731(0.03)		&0.839(0.05)	&0.825(0.05)		&0.845(0.05)	 &0.753(0.03)&0.743(0.04)\\
r=5	&0.976(0.01)	&0.879(0.02)	    &0.938(0.03)	&0.967(0.02)		&0.726(0.04)	 &0.585(0.04)&0.669(0.04)\\
r=9	&0.99(0.01)	&0.936(0.02)		&0.968(0.02)	&0.981(0.01)	    &0.697(0.06)	&0.558(0.04)&0.748(0.05)\\
\hline
 & p=1225 &      &      &      & 	& &\\
\hline
r=1	&0.914(0.02)	&0.642(0.04)	&0.818(0.05)	&0.766(0.06)	&0.832(0.05)	&0.666(0.05)&0.654(0.04)\\
r=5	&0.977(0.01)	&0.778(0.03)	&0.895(0.07)	&0.952(0.02)	&0.725(0.05)	&0.543(0.05)&0.562(0.04)\\
r=9	&0.989(0.01)	&0.848(0.03)	&0.915(0.07)	&0.962(0.03)	&0.707(0.06)	&0.533(0.06)&0.573(0.04)\\
\hline
\end{tabular}
\end{scriptsize}
\end{table}

\section{Real data Application}
\label{sec-data}
{
Alzheimer’s disease (AD) is a neuro-degenerative disease and the most common form of dementia, affecting many millions around the world. Classification of AD patients is a crucial task in dementia research.  To apply our classification method, we obtain the data from the Alzheimer's disease Neuroimaging Initiative (ADNI) database (http:// www.loni.ucla.edu/ADNI), which was launched in 2004. ADNI aims to improve clinical trials for prevention and treatment of Alzheimer's disease (AD). With the interest of promoting consistency in data analysis, the ADNI Core has created standardized analysis sets of the structured MRI scans comprising only image data that have passed quality control (QC) assessments. The assessments were conducted at the Aging and Dementia Imaging Research laboratory at the Mayo Clinic \citep{jack2008alzheimer}. In this study, we used T1-weighted MRI images from the collection of standardized datasets. The description of the standardized MRI imaging from ADNI can be found in \url{http://adni.loni.usc.edu/methods/mri-analysis/adni-standardized-data/} and \cite{wyman2013standardization}.}

According to \cite{jack2008alzheimer}, the images were generated using magnetization prepared rapid gradient echo (MPRAGE) or equivalent protocols with varying resolutions (typically 1.0 $\times$ 1.0 mm in plane spatial resolution and 1.2 mm thick sagittal slices with $256\times256\times166$ voxels). The images were then pre-processed according to a number of steps detailed in \cite{jack2008alzheimer} and \url{http://adni.loni.usc.edu/methods/mri-analysis/mri-pre-processing/}, which corrected gradient non-linearity, intensity inhomogeneity and phantom-based distortion. In addition, the pre-processed imaging were processed by FreeSurfer for cortical reconstruction and volumetric segmentation by Center for Imaging of Neurodegnerative Diseases, UCSF.

In this paper, we obtain images from ADNI-1 subjects obtained using 1.5 T scanners at screening visits. We use the first time point if there are multiple images of the same subject acquired at different times. 187 subjects diagnosed as Alzheimer's disease at screening visits and 227 healthy subjects at screening visits are contained in this study. The total number of subjects is 414. Details of the subjects can be found in Table \ref{T_subject_all}. {The authors used ADNI data in their previous research work \citep{zhang2019analysis,li2019early}. Please refer these papers for other information about the data.}

\begin{table}[]
\centering
\caption{Subjects characteristics}
\label{T_subject_all}
\begin{tabular}{l l l l }
\hline
    &    AD      &   NL   &   p-value\\
    \hline
n	             &$187$				&	$227$		&			\\
Age (Mean$\pm$sd)&	$75.28\pm7.55$	& $75.80\pm4.98$& $0.4168$\\
Gender (F/M)    	&$88/99$&	$110/117$	&$0.813$\\
MMSE (Mean$\pm$sd)&	$23.28\pm2.04$	&$29.11\pm1.00$&	$<1e-15$\\
\hline
\end{tabular}
\begin{flushleft}
\small{Key: AD, subjects with Alzheimer's disease ; NL, healthy subjects; Age, baseline age; MMSE, baseline Mini-Mental State Examination.}
\end{flushleft}
\end{table}

After retrieving the pre-processed imaging data from ADNI, an R package ANTsR is applied for imaging registration. Then we use ``3dresample'' command by AFNI software (\cite{cox1996afni}) to adjust the resolution and reduce the total number of voxels in the images to $18\times22\times18$ voxels. Take $x$ axis and $y$ axis for horizontal plane, $x$ axis and $z$ axis for coronal plane and $y$ axis and $z$ axis for sagittal plane.

After removing the voxels with zero signal for most of the subjects (more than 409 subjects), we have $1971$ voxels left in use. The distance between each pair of voxels can be calculated by their coordinates. For example, there are two voxels $s_1$, $s_2$ with coordinate $s_1=(x_1,y_1,z_1)$ and $s_2=(x_2,y_2,z_2)$. Then the Euclidean distance between $s_1$ and $s_2$ is defined by: $d (s_1,s_2)=\sqrt{(x_1-x_2)^2+(y_1-y_2)^2+(z_1-z_2)^2}$. Other distances can also be used in our method. 

We randomly sample 100 from the 187 AD subjects and 100 from the 227 health subjects as the training set. Then there are 87 AD subjects and 127 healthy subjects left. The testing set includes the 87 AD subjects and a random sample of 87 from the 127 healthy subjects. Details of the subjects in the training and testing set are provided in Table \ref{T_subject_testtraing}.

\begin{table}[]
\centering
\caption{Subjects characteristics of training and testing set}
\label{T_subject_testtraing}
\begin{tabular}{l l l l l}
\hline
   & &    training set     &   testing set  &   p-value\\
    \hline
AD&n	             &$100$				    &	$87$		    &			\\
  &Age (Mean$\pm$sd) &	$75.64\pm7.39$	& $74.85\pm7.75$& $0.478$\\
  &Gender (F/M)    	 &  $47/53$         &  $41/46$ 	    &$0.999$\\
  &MMSE (Mean$\pm$sd)&	$23.22\pm2.08$	&  $23.36\pm2.01$&	$0.649$\\
\hline
NL&n	             &$100$				&	$87$		&			\\
  &Age (Mean$\pm$sd)&	$75.99\pm5.39$	& $75.34\pm4.56$& $0.3723$\\
  &Gender (F/M)    	&$42/58$&	$50/37$	&$0.05$\\
  &MMSE (Mean$\pm$sd)&	$29.06\pm1.04$	&$29.09\pm1.01$&	$0.8307$\\
\hline
\end{tabular}
\begin{flushleft}
\small{Key: AD, subjects with Alzheimer's disease ; NL, healthy subjects; Age, baseline age; MMSE, baseline Mini-Mental State Examination.}
\end{flushleft}
\end{table}

We assume the exponential correlation structure among voxels. Then we apply the PMLE-LDA method proposed in this research for classification. First, the parameter are estimated by PMLE: $r=61.66, c=0.954, \sigma^2=223.09$ and 5 voxels are selected for classification from training data. Then we plug-in the estimates into the classification function and obtain classification results on the testing data.

The classification accuracy rate of PMLE-LDA is listed in Table \ref{TCReal}. We also list the classification accuracy from other methods. It shows the classification accuracy rate of our method is about $77.0\%$, which is superior to other comparable methods (MLE-LDA: $69.0\%$, PREG-LDA: $75.9\%$, FAIR: $56.9\%$, NB: $66.1\%$ and RPEC: $67.8\%$).

\begin{table}[]
\centering
\begin{scriptsize}
\caption{Classification performance for voxel level MRI data. We split the data into Training sets (200 samples) and testing sets (174 samples) as described.}
\label{TCReal}
\begin{tabular}{ccccccc}

\hline
                         & MLE      & PREG & PMLE  & FAIR &NB &RPEC\\
\hline
 Accuracy                &0.690&  0.759& 0.770 &  0.569 &0.661   & 0.678\\
 No. of training err     & 37  &  52   &  47   &  68    &  43 & 21\\
 No. of testing err      & 54  &  42   &  40   &  75    &  59 & 56\\
 No. of selected voxels  &1971 &  26    &  4    & 16       &1971& --\\
\hline
\end{tabular}
\end{scriptsize}
\end{table}

To show the robustness of the proposed method, we repeat the above procedure for 100 times and calculated the average accuracy and standard deviation in Table \ref{TCReal100}. Although we know that MRI data is noisy in general and heterogeneous across subjects, our method still provide reasonably higher classification rates with small standard deviations.  

\begin{table}[]
\centering
\begin{scriptsize}
\caption{Average classification performance for voxel level MRI data. We repeated the training/testing procedure for 100 times.}
\label{TCReal100}
\begin{tabular}{ccccccc}

\hline
            & MLE    & PREG  & PMLE  & FAIR  &NB   &RPEC\\
\hline
Accuracy 	&0.704	&0.704	&0.703	&0.577	&0.691& 0.695\\
(sd)        &0.036	&0.035	&0.035	&0.054	&0.034  & 0.032 \\
No. of selected voxels  &1971 &  1076.30    &  1043.35    & 6       &1971& --\\
(sd)        & --    & 35.61 & 35.5 & 0.19 & -- & --\\
\hline
\end{tabular}
\end{scriptsize}
\end{table}

\section{Conclusion and discussion}

{
The paper contains new developments for the classification problem of multivariate Gaussian variables with spatial structures. We generalize the classical LDA by assuming spatially dependent structures in the covariance and imposing sparsity on the feature difference. In particular, by using the $Mat\acute{e}rn$ covariance function, the $p\times p$ dimensional covariance is parameterized by only three univariate parameters. By utilizing the additional spatial location information and constructing the data-driven spatial correlation structure in the data, the new spatial LDA method is expected to be more efficient than other sparse LDA methods. Under this framework, we adopt Penalized Maximum Likelihood Estimation (PMLE) method to perform parameter estimation. Most importantly, we show in theory that the proposed method can not only provide consistent results of parameter estimation and feature selection, but also achieve an asymptotically optimal classifier for high dimensional data with spatial structures. 

Brain imaging data are usually matrix-variate or tensor-variate observations. In this paper, we adopt one type of data arrangements that vectorize the tensor data into vectors and stack covariates along with the long vector to apply vector methods. This will inevitably increase the dimension of covariance matrix and thus lead to high computational burden regarding matrix operations. We can increase the speed by avoiding matrix inverse calculation~\citep{bhattacharya2016fast} or directly estimating sample covariance instead of using MLE. In the future, we plan to extend our current computing strategy to tensor methods (e.g. tensor LDA) that significantly reduce the number of parameters and hence alleviate computation complexity~\citep{li2017parsimonious, pan2019covariate}.

In general, nonparametric methods do not have the assumption on data distribution. It overcomes the limitation of parametric LDA, which can not perform well in non-Gaussian data. 
But intrinsically, these two types of methods are demanded to solve the same kind of problem, which is to resemble the changing pattern of brain regions. Mostly, they use average trend to denote the deterioration of brain function and use network structure to denote synergies of the functional connectivity. To model the average trend, typical methods such as kernel function, spline, wavelet and Fourier transformation are exploited to transform the original coordinates to a new space with better properties. To model the network structure, typical methods such as nonparametric graphical model, random matrices and Fourier transform of images are exploited to decompose the nodes into new clusters relevant information. 
However, one disadvantage of nonparametric method is that the results are not straightforward to interpret, as they are wrapped up in a ``black box''. 
In the future, we plan to extend our current method to nonparametric models by using Fourier transformation. We will work more on how to include inverse transformation into the model so that the results are interpretable in the original space.

}

\section*{Remarks on the assumptions}
\label{remarks_asu}
\begin{itemize}
\item[\textbf{Remarks on A\ref{asu_eigen_sig}}]: The first part of A\ref{asu_eigen_sig} is the same as that in \cite{mardia1984maximum}.   We now verify that the covariance matrix derived from $Mat\acute{e}rn$ covariance function satisfy the first part of A\ref{asu_eigen_sig}. First, for symmetric matrix, we have
\begin{align*}
\lambda_{\max}(\Sigma)\le(\norm{\Sigma}{1})^{1/2}(\norm{\Sigma}{\infty})^{1/2}=\norm{\Sigma}{\infty}=\max_{i}\sum_{j}^{p}\gamma(h_{ij})
\end{align*}
Using the same notation in the proof of Lemma~\ref{lem_tapermatrix}, for each $i$
\begin{align}\label{remark(0)}
\sum_{j=1}^{p}\gamma(h_{ij})\le \sum_{m=0}^{\infty}\sum_{j\in B_m^i}r(h_{ij})\le K\rho\sum_{m=0}^{\infty}m^{d-1}\delta^d \max_{j\in B_m^i}r(h_{ij})\le K\rho\int_{0}^{\infty}h^{d-1}r(h)dh
\end{align}
Recall that $Mat\acute{e}rn$ covariance function has the following expansion at $h=0$:
\begin{align*}
\gamma(h;\sigma^2,c,\nu,r)=\sigma^2(1-c)(1-b_1h^{2\nu}+b_2h^2+O(h^{2+2\nu})) \ as\ h\to 0
\end{align*}
where $b_1$ and $b_2$ are explicit constants depending only on $\nu$ and $r$. Thus for $\epsilon>0$,
\begin{align}\label{remark(1)}
\int_{0}^{\epsilon}h^{d-1}\gamma(h)dh=O(\int_{0}^{\epsilon}h^{d-1}dh)=O(\epsilon^d/d)\to 0 \ as\ \epsilon\to 0
\end{align}

Also, since $K_{\nu}(h)\propto e^{-h}h^{-\half}(1+O(\frac{1}{h}))$ as $h\to \infty$, there exist a constant $K$, for any $C$ sufficiently large, we have:
\begin{align}\label{remark(2)}
\int_{C}^{\infty}h^{d-1}\gamma(h)dh\le K \int_{0}^{\infty} h^{d-1+v-\half} e^{-h}dh=\Gamma(d+v-\half)<\infty
\end{align}

\ref{remark(1)} and \ref{remark(2)} lead to $\int_{0}^{\infty}h^{d-1}\gamma(h)dh<\infty$. Let $p\to\infty$ in \ref{remark(0)}, we have $\lim sup_{p\to \infty} \lambda_{\max}(\Sigma)<\infty$ if $\Sigma$ is derived from $Mat\acute{e}rn$ covariance function.

Now consider the second part of A\ref{asu_eigen_sig}. A\ref{asu_domain} assumes increasing domain framework. \cite{bachoc2016smallest} showed that under A\ref{asu_domain} and some weak assumptions on the matrix covariance function, if the spectral density of the covariance function is positive, the smallest eigenvalue of the covariance matrix is asymptotically bounded away from zero. Most standard covariance function such as $Mat\acute{e}rn$ covariance function satisfy those assumptions hence satisfy the second part of A\ref{asu_eigen_sig}.

\item[\textbf{Remarks on A\ref{asu_covnorm_inv} and A\ref{asu_tij_lim}}]: A\ref{asu_covnorm_inv} and A\ref{asu_tij_lim} are the same as the assumptions in \cite{mardia1984maximum}. $\norm{\Sigma_k}{F}=\sum_{i,j=1}^{p}\gamma_k^2(h_{ij};\btheta)$, where $\gamma_k(h_{ij};\btheta)=\der{\gamma(h_{ij};\btheta)}{\theta_k}$, $k=1,2,...,q$ and $\btheta$ is a $k$ dimensional parameter. We now verify that $Mat\acute{e}rn$ covariance function satisfy A\ref{asu_covnorm_inv} for fixed $\nu$. For $Mat\acute{e}rn$ covariance function with fixed $\nu$, we have
\begin{align} \label{remark_matern_der}
\der{\gamma(h)}{\sigma^2}&=\frac{2^{1-\nu}}{\Gamma(\nu)}(h/r)^{\nu}K_{\nu}(h/r)(1-c)\\\nonumber
\der{\gamma(h)}{c}&=-\sigma^2\frac{2^{1-\nu}}{\Gamma(\nu)}(h/r)^{\nu}K_{\nu}(h/r)\\\nonumber
\der{\gamma(h)}{(1/r)}&=\sigma^2(1-c)\frac{2^{1-\nu}}{\Gamma(\nu)}h(h/r)^{\nu}(2\frac{\nu}{h/r}K_{\nu}(h/r)-K_{\nu-1}(h/r))
\end{align}
It is easy to show that for each $k$, exist a constant $\epsilon>0$ independent of $n,p$, for each $i$, there's $j$ such that $\gamma_k(h_{ij})>c$. As a result, $\norm{\Sigma_k}{F}=\sum_{i,j=1}^{p}\gamma_k^2(h_{ij};\btheta)\ge \sum_{i=1}^{p}\epsilon=p\epsilon$. Therefore $\norm{\Sigma_k}{F}^{-1}=O_p(p^{-1})$.

\item[\textbf{Remarks on A\ref{asu_cov_star}}]: We can also verify that the $Mat\acute{e}rn$ covariance function with fixed $\nu$ satisfy A\ref{asu_star_coveigen} and A\ref{asu_star_coveigen2}. Similar to the procedure in the remarks on A\ref{asu_eigen_sig}, it is sufficient to verify for any $\btheta\in \Theta$, $\gamma_k(h;\btheta)$ and $\gamma_{kj}(h;\btheta)$ belong to the function space:
\begin{align*}
\Im=\{f(x):\int_{0}^{\infty}f(x)x^{d-1}dx<\infty\}
\end{align*}
where $d\ge 1$ is the dimension of the domain. We have the first-order partial derivative of $Mat\acute{e}rn$ covariance function in (\ref{remark_matern_der}). The second-order partial derivative of $Mat\acute{e}rn$ covariance function is as follows:
\begin{align}\label{remark_matern_der2}
&\frac{\partial^2{\gamma(h)}}{(\partial^2 \sigma^2)}=0\\\nonumber
&\frac{\partial^2{\gamma(h)}}{\partial \sigma^2 \partial c}=-\frac{2^{1-\nu}}{\Gamma(\nu)}(h/r)^{\nu}K_{\nu}(h/r)\\\nonumber
&\frac{\partial^2{\gamma(h)}}{\partial \sigma^2 \partial (1/r)}=(1-c)\frac{2^{1-\nu}}{\Gamma(\nu)}h(h/r)^{\nu}(2\frac{\nu}{h/r}K_{\nu}(h/r)-K_{\nu-1}(h/r))\\\nonumber
&\frac{\partial^2{\gamma(h)}}{\partial c^2}=0\\\nonumber
&\frac{\partial^2{\gamma(h)}}{\partial c \partial (1/r)}=-\sigma^2\frac{2^{1-\nu}}{\Gamma(\nu)}h(h/r)^{\nu}(2\frac{\nu}{h/r}K_{\nu}(h/r)-K_{\nu-1}(h/r))\\\nonumber
&\frac{\partial^2{\gamma(h)}}{\partial^2(1/r)}=\sigma^2(1-c)\frac{2^{1-\nu}}{\Gamma(\nu)}[(h/r)^{\nu-2}h^2K_{\nu}(h/r)(2\nu-1)2\nu\\\nonumber
&\  \ \ \ \ \ \ \ -(4\nu+1)(h/r)^{\nu-1}h^2K_{\nu+1}(h/r)-(h/r)^{\nu}h^2K_{\nu+2}(h/r)]
\end{align}
Note that the covariance function and its first-order and second-order partial derivatives are linear combinations of a Bessel function of $h$ times a polynomial of $h$. Similar to proving $\int_{0}^{\infty}h^{d-1}\gamma(h)dh<\infty$ in (\ref{remark(0)}), we have $\gamma_k(h;\btheta)\in \Im$ and  $\gamma_{kj}(h;\btheta)\in \Im$. Hence A\ref{asu_star_coveigen} and A\ref{asu_star_coveigen2} are satisfied. By similar procedure, we can verify that $Mat\acute{e}rn$ covariance function also satisfy A\ref{asu_covfunc_int} and A\ref{asu_covfunc_int2}.

\end{itemize}

\clearpage
\bibliography{Li_Reference}
\bibliographystyle{cbe}

\clearpage
\newpage

\title{Supplementary Material for ``High Dimensional Classification for Spatially Dependent Data with Application to Neuroimaging''}

\author{{Yingjie Li$^\ast$, Liangliang Zhang$^\dag$ and Tapabrata Maiti$^\ast$}\\
{\small \em $^\ast$Department\ of\ Statistics\ and\ Probability,\
  Michigan\ State\ University,\ East Lansing,\ Michigan,\ U.S.A.}\\
{\small \em $^\dag$Department\ of\ Biostatistics,\ University\ of\ Texas\ MD\ Anderson\ Cancer\ Center,\ Houston,\ Texas,\ U.S.A.}\\
{\small \em liangliangzhang.stat@gmail.com}}

\date{}
\maketitle

\setcounter{page}{1}

\setcounter{section}{0}

\section{Proofs for classification using MLE}
\begin{lemma}\label{lemma_sumEps}
Let $\beps$ be $p$-dimensional vectors and $\beps\sim N(0,\Sigma)$, where $\Sigma$ is a $p\times p$ positive definite covariance matrix. For $m$-dimension vector $\bu$ with $\norm{\bu}{2}=\sqrt{\bu^T\bu}=C$ and $p\times m$ matrix $\bX_i$, we have:
\begin{align}
\abs{\beps^T\bX u}=O_p(\sqrt{tr(\bX^T\Sigma\bX)}\norm{\bu}{2})
\end{align}
\end{lemma}

\begin{proof}
Since $E(\beps^T\bX)=0$,
\begin{align}
E(\beps^T\bX \bu)^{2}\le& \left[E(\beps^T\bX \bX^T\beps)\right]^{1/2}\norm{\bu}{2}
=\left[E(tr(\beps^T\bX \bX^T\beps))\right]^{1/2}\norm{\bu}{2}\\\nonumber
=&tr(\bX^T\Sigma\bX)\norm{\bu}{2}
\end{align}
By Chebyshev's inequality, for any $M$
\begin{align}
P(\frac{\beps^T\bX \bu}{\sqrt{\norm{\bu}{2}^2tr(\bX^T\Sigma\bX)}}>M)\le \frac{E(\beps^T\bX u)^{2}}{M^2 \norm{\bu}{2}^2tr(\bX^T\Sigma\bX)}=\frac{1}{M^2}
\end{align}
Thus for any $\epsilon>0$, exits $M$ large enough such that
\begin{align}
P(\frac{\abs{\beps^T\bX \bu}}{\sqrt{\norm{\bu}{2}^2c(n)tr(\bX^T\Sigma\bX)}}>M)<\epsilon
\end{align}
This lead to $\abs{\beps^T\bX u}=O_p(\sqrt{tr(\bX^T\Sigma\bX)}\norm{\bu}{2})$.
\end{proof}

\begin{lemma}\label{lemma_sumEps_sq}
Let $\beps_i (i=1,2,...,n)$ be $p$-dimensional vectors and $\beps_i\sim N(0,c(n)\Sigma)$, where $c(n)$ is a function of $n$ and $\Sigma$ is a $p\times p$ positive definite covariance matrix with $\lambda(\Sigma)<\infty$. For a $p\times p$ matrix $A$, we have
$$\sum_{i=1}^{n}\left[\beps_i^T A \beps_i-c(n)tr(A\Sigma)\right]=O_p(c(n)\sqrt{n}\norm{A}{F})$$
\end{lemma}
\begin{proof}
Since $E(\beps_i^T A \beps_i)=tr(c(n)A\Sigma)$, we have
\begin{align}
E(\sum_{i=1}^{n}\beps_i^T A\beps_i-c(n)tr(A\Sigma))^2=\sum_{i=1}^{n}E(\beps_i^T A\beps_i-c(n)tr(A\Sigma))^2=\sum_{i=1}^{n}E(\beps_i^T A\beps_i)^2-nc^2(n)tr^2(A\Sigma)
\end{align}
Let $B=c(n)\Sigma^{\half}A \Sigma^{\half}$, then exit orthogonal matrix $Q$ such that $B=Q^T \Lambda Q$ where $\Lambda=diag(\lambda_i)$ and $\lambda_i$ are eigenvalues of $B$. Let $\tilde{\beps}_i=\sqrt{c(n)Q\Sigma^{-\half}\beps_i}$, then $\tilde{\beps}_i\sim N(0, I_{p\times p})$ where $I_{p\times p}$ is identity matrix. Then
\begin{align}
E(\beps^T_i A\beps_i)^2&=E(\tilde{\beps}_i^T \Lambda \tilde{\beps}_i)^2=E(\sum_{i=1}^{p}\lambda_j\tilde{\beps}_{ij}^2)^2\\\nonumber
&=E(\sum_{i=1}^{p}\lambda_j^2\tilde{\beps}_{ij}^4+\sum_{j,k=1}^{p}\lambda_j\lambda_k \tilde{\beps}_{ij}^2 \tilde{\beps}_{ij}^2)=2\sum_{i=1}^{p}\lambda_j^2+(\sum_{i=1}^{p}\lambda_j)^2\\\nonumber
&=2tr(B^T B)+tr^2(B)\\\nonumber
&=c(n)^2[2tr(\Sigma A \Sigma A)+tr^2(A\Sigma)]
\end{align}
Hence
\begin{align}
E(\sum_{i=1}^{n}\beps_i^T A\beps_i-c(n)tr(A\Sigma))^2=2n c(n)^2 tr(\Sigma A^T \Sigma A)\le 2n c(n)^2 \lambda^2_{\max}(\Sigma)tr(A^T A)
\end{align}
By Chebyshev's inequality, for any $M$ we have
\begin{align}
P(\frac{\sum_{i=1}^{n}\beps_i^T A\beps_i-c(n)tr(A\Sigma)}{\sqrt{n c(n)^2 \norm{A}{F}^2}}>M)\le \frac{2n c(n)^2 tr(\Sigma A^T \Sigma A)}{M^2 n c(n)^2 \norm{A}{F}^2}\le \frac{2\lambda^2_{\max}(\Sigma)}{M^2}
\end{align}
Then for any $\epsilon>0$ exits $M$ large enough such that
\begin{align}
p(\frac{\sum_{i=1}^{n}\beps_i^T A\beps_i-c(n)tr(A\Sigma)}{\sqrt{n c(n)^2 \norm{A}{F}^2})}>M)<\epsilon
\end{align}
which means $\sum_{i=1}^{n}\beps_i^T A \beps_i-c(n)tr(A\Sigma)=O_p(c(n)\sqrt{n}\norm{A}{F})$
\end{proof}

\begin{proof}[\textbf{Proof of Theorem \ref{ThconsistMLE}}]
Take derivative with respect to $\bmu_k$ ($k=1,2$) with the function $L(\btheta,\bmu_1,\bmu_2)$ defined in \ref{loglik}. Considering $\inSigma(\btheta)$ is nonsingular, we have $\bmu_k=\hmu_{kMLE}=\bar{\bY}_{k\cdot}$.

For $\htheta_{MLE}$, we first consider the case of $p/n\to 0$. It is sufficient to prove that for any given $\epsilon>0$, there is a large constant $C$ such that for large $p$ and $n$, the smallest rate of convergence $\eta_{n,p}$ is $\sqrt{\frac{1}{np}}$ such that we have
\begin{align}
P(\sup_{\norm{\bu}{2}=C}L(\btheta_0+\bu \eta_{n,p},\hmu_1,\hmu_2)<L(\btheta_0,\hmu_1,\hmu_2))>1-\epsilon
\end{align}
where $\bu\in \IR^q$. This implies that there exists a local maximum for the function $L$ in the neighborhood of $\btheta_0$ with the radius at most proportional to $\eta_{n,p}$.

\begin{align}
&L(\btheta_0+\bu\eta_{n,p},\hmu_1,\hmu_2)- L(\btheta_0,\hmu_1,\hmu_2)\\\nonumber
&=(\der{L}{\btheta}(\btheta_0))^T\bu\eta_{n,p} + \half \bu^T (\frac{\partial ^2 L}{\partial \btheta \partial \btheta^T}(\btheta^*))\bu \eta_{n,p}^2\\\nonumber
&=-\frac{n}{2}\bu^T T(\btheta_0) \bu \eta_{n,p}^2+ (\der{L}{\btheta}(\btheta_0))^T\bu\eta_{n,p} + \half \bu^T (\frac{\partial ^2 L}{\partial \btheta \partial \btheta^T}(\btheta^*)+mnT(\btheta_0))\bu \eta_{n,p}^2\\\nonumber
&=(I)+(II)+(III)
\end{align}
where $T(\btheta_0)$ is a $q\times q$ matrix with its $(i,j)th$ element $t_{ij}(\btheta_0)=tr(\inSigma\Sigma_i\inSigma\Sigma_j)$.


From A\ref{asu_covnorm_inv} and A\ref{asu_tij_lim}, $t_{ij}=a_{ij}(t_{ii})^{\half}(t_{jj})^{\half}\ge a_{ij}\lambda^{-2}_{\min}(\Sigma)\norm{\Sigma_i}{F}\norm{\Sigma_j}{F}$. There exists a constant M such that
\begin{align}
(I)=-\frac{n}{2}\bu^T T \bu \eta_{n,p}^2=-\frac{n}{2}\sum_{i,j=1}^{q}t_{ij}u_iu_j\eta_{n,p}^2\le-\frac{Mnp}{2}\eta_{n,p}^2\norm{\bu}{2}^2
\end{align}

\begin{align}
(II)&=-\half\sum_{k=1}^{2}\sum_{i=1}^{n_k}\sum_{j=1}^{q} \left[(\bY_{ki}-\hmu_k)^T\Sigma^j(\bY_{ki}-\hmu_k)-(\frac{n_k-1}{n_k})tr(\Sigma\Sigma^j)\right]u_{\theta_j}\eta_{n,p}\\\nonumber
&+\half\sum_{j=1}^{q}tr(\Sigma\Sigma^j)u_{\theta_j}\eta_{n,p}\\\nonumber
&=(1)+(2)
\end{align}
Because $\bY_{ki}-\hmu_k\sim N(0, \frac{n_k-1}{n_k}\Sigma)$, by lemma\ref{lemma_sumEps_sq},
\begin{align}
\abs{(1)}=O_p(\sum_{k=1}^{2}\frac{n_k-1}{\sqrt{n_k}}\norm{\Sigma^j}{F}\norm{\bu}{2}\eta_{n,p})=O_p(\sqrt{n}\norm{\Sigma^j}{F}\norm{\bu}{2}\eta_{n,p})=O_p(\sqrt{np}\norm{\bu}{2}\eta_{n,p})
\end{align}
The last equality is because from A\ref{asu_star_coveigen}, $\norm{\Sigma^i}{F}^2\le \norm{\Sigma^i}{2}^2=p\lambda_{\max}^2(\Sigma^i)=O(p)$.

Also by A\ref{asu_star_coveigen}, $tr(\Sigma\Sigma^j)=tr(\Sigma^{1/2}\Sigma^j\Sigma^{1/2})\le \lambda_{\max}(\Sigma^j)tr(\Sigma)$. Then noticing that $tr(\Sigma)=O(p)$, and $p/n\to 0$
\begin{align*}
\abs{(2)}=\abs{\half \sum_{j=1}^{q}tr(\Sigma\Sigma^j)u_j\eta_{n,p}}=O_p(tr(\Sigma)\eta_{n,p}\norm{\bu}{2})=O_p(p\eta_{n,p}\norm{\bu}{2})
\end{align*}
Thus
\begin{align*}
(II)=O_p((\sqrt{np}+p)\eta_{n,p}\norm{\bu}{2})
\end{align*}
\begin{itemize}
\item If $p/n\to 0$, $(II)=O_p(\sqrt{np})\eta_{n,p}$. By choosing sufficient large $C=\norm{\bu}{2}$, the minimal rate of $\eta_{n,p}$ to have $(II)$ be dominated by $(I)$ is $\eta_{n,p}=O_p(\sqrt{\frac{1}{np}})$;
\item If $p/n \to C $ with $0<C\le \infty$, $(II)=O_p(p\eta_{n,p})$. Then the minimal rate of $\eta_{n,p}$ to have $(II)$ dominated by $(I)$ is $\eta_{n,p}=O_p(\sqrt{\frac{1}{n}})$
\end{itemize}
Since
\begin{align*}
\frac{\partial L}{\partial \theta_j \partial \theta_l}(\btheta)=\frac{n}{2}\left[ tr(\Sigma^{j}(\btheta)\Sigma(\btheta))+tr(\Sigma^j(\btheta)\Sigma_l(\btheta))\right]-\frac{1}{2}\sum_{k=1}^{2}\sum_{i=1}^{n_k}(\bY_{ki}-\hmu_k)^T\Sigma^{jl}(\btheta) (\bY_{ki}-\hmu_k)
\end{align*}
$(III)$ can be written as
\begin{align*}
(III)=&-\half\sum_{j,l=1}^{q}\left[\sum_{k=1}^{2}\sum_{i=1}^{n_k} \big((\bY_{ki}-\hmu_k)^T\Sigma^{jl}(\btheta^*)(\bY_{ki}-\hmu_k)-(\frac{n_k-1}{n_k})tr(\Sigma(\btheta_0)\Sigma^{jl}(\btheta^*))\big)      \right]u_{\theta_j}u_{\theta_l}\eta_{n,p}^2\\\nonumber
&+\frac{n}{2n_1n_2}\sum_{j,l=1}^{q}tr(\Sigma^{jl}(\btheta^*)\Sigma(\btheta_0)) u_{\theta_j}u_{\theta_l}\eta_{n,p}^2\\\nonumber
& +\frac{n}{2} \sum_{j,l=1}^{q}\left[ tr(\Sigma^{jk}(\btheta^*)\Sigma(\btheta^*)) - tr(\Sigma^{jk}(\btheta^*)\Sigma(\btheta_0))\right] u_{\theta_j}u_{\theta_l}\eta_{n,p}^2\\\nonumber
&+ \frac{n}{2} \sum_{j,l=1}^q\left[ tr(\Sigma^{j}(\btheta^*)\Sigma_l(\btheta^*)) - tr(\Sigma^{j}(\btheta_0)\Sigma_l(\btheta_0))\right] u_{\theta_j}u_{\theta_l}\eta_{n,p}^2\\\nonumber
= & (3)+(4)+(5)+(6)
\end{align*}
By lemma\ref{lemma_sumEps_sq} and A\ref{asu_star_coveigen2},
\begin{align*}
\abs{(3)}=O_p(\sqrt{n}\norm{\Sigma^{jl}(\theta^*)}{F}\eta_{n,p}^2)=O_p(\sqrt{np}\eta_{n,p}^2)
\end{align*}

For (4), by A\ref{asu_star_coveigen2}
\begin{align*}
tr(\Sigma^{jl}(\btheta^*)\Sigma(\btheta_0))=O_p(p)
\end{align*}
thus $\abs{(4)}=O_p(\frac{p}{n}\eta_{n,p}^2)$
It is easy to see (3), (4) are dominated by $(I)$.
For (5),
\begin{align}
|(5)|=\left|{1\over4}n\sum_{j,k=1}^{q} tr(\Sigma^{kj}(\btheta^*)(\Sigma(\btheta_0)-\Sigma(\btheta^*)))  u_{\theta_k}u_{\theta_j}\eta_{n,p}^2\right|.
\end{align}
Let $d_{il}(\btheta^*)$ be the $i,l$th entry of matrix $\Sigma^{kj}(\btheta^*)$, $\gamma_{il}(\btheta) $ be the $i,l$th entry of $\Sigma(\btheta)$, then by A\ref{asu_covfunc} and A\ref{asu_star_coveigen2}

\begin{align}
tr(\Sigma^{kj}(\btheta^*)(\Sigma(\btheta_0)-\Sigma(\btheta^*))) &=\sum_{i,l=1}^{p}d_{il}(\btheta^*)(\gamma_{li}(\btheta_0)-\gamma_{li}(\btheta^*))\\\nonumber
&\le \sum_{i,l=1}^{p}|d_{il}(\btheta^*)| {\parallel \frac{\partial \gamma_{il}(\theta^*)}{\partial \btheta}\parallel_2} {\parallel \btheta_0-\btheta^*\parallel_2}\\\nonumber
&\le \sum_{i,l=1}^{p}d_{il}(\btheta_2^*) M \eta_{n,p} \\\nonumber
&\le M \eta_{n,p} p\norm{\Sigma^{kj}(\btheta^*)}{F} \\\nonumber
&=O_p(\sqrt{p^3}\eta_{n,p})
\end{align}
Hence
\begin{align}
\abs{(5)}=O_P(n \sqrt{p^3}\eta^3)=O_p((p^{3/2}n\eta_{n,p})\eta_{n,p}^2))
\end{align}
\begin{itemize}
\item If $p/n\to 0$ and $\eta_{n,p}=O_p(\frac{1}{\sqrt{np}})$, (5) is dominated by $(I)$.
\item If $p/n\to C$ with $0<C\le \infty$, $eta_{n,p}=O_p(\frac{1}{n})$ and $\sqrt{p}/n\to 0$, (5) is dominated by $(I)$.
\end{itemize}

For (6), let $t_{ij}(\theta^*)=tr(\Sigma^{-1}(\btheta^*)\Sigma_i(\btheta^*)\Sigma^{-1}(\btheta^*)\Sigma_j(\btheta^*))$, by A\ref{asu_star_covdif3},
\begin{align}
|(6)|&\le n\sum_{k,j=1}^{q}\norm{\der{ t_{ij}(\btheta^*)}{\btheta}}{2} \norm{\btheta^*-\btheta_0}{2} u_{\theta_k}u_{\theta_j}\eta_{n,p}^2\\\nonumber
&=O_P(np\eta_{n,p}^3)
\end{align}
While $\eta_{n,p}=O_p(\frac{1}{\sqrt{np}})$ or $\eta_{n,p}=O_p(\frac{1}{n})$, $(6)$ is also dominated by $(I)$. Hence $(III)$ is dominated by $(I)$. This completes the proof.
\end{proof}

\begin{proof}[\textbf{Proof of Theorem \ref{ThErrpLessn}}]
We start with $W_1(\hat{\delta}_{MLE})=1-\Phi(\Psi_1)$, where $W_1(\hat{\delta}_{MLE})$ is the conditional misclassification rate defined in \ref{W1} and $\Psi_1$ is defined in \ref{Psi1}. The idea is to prove $\lim\inf_{n,p\to0}\Psi_1\to \frac{\sqrt{C_0}}{2}$.

From Therom \ref{ThconsistMLE}, we have $\parallel \htheta-\btheta \parallel_2=O_p(\frac{1}{\sqrt{np}})$.
Recall that $\Sigma=\Sigma(\btheta)=\big[ \gamma(h_{ij};\btheta)\big]_{i,j=1}^{p}$ and $\hSigma=\Sigma(\htheta)=\Big[ \gamma(h_{ij};\htheta)\Big]_{i,j=1}^{p}$.

By A\ref{asu_covfunc}, we have:
\begin{align}
\max_{i,j}|\gamma(h_{ij};\btheta)-\gamma(h_{ij};\htheta)|\le M  \parallel \btheta-\htheta \parallel_2
\end{align}
Thus there exist $\epsilon>0$ and matrix $E=\big[ e_{ij}\big]_{i,j=1}^{p}$ such that
\begin{align}
\hSigma=\Sigma+\epsilon E
\end{align}
where $\epsilon=O_p(\frac{1}{\sqrt{np}})$ and $E$ is a $p\times p$ matrix with absolute values of all entries less than $1$, i.e. $\abs{e_{ij}}\le 1$ for any $i,j=1,2,...,p$. As a result, for large $p$ and $n$, the inverse of $\Sigma$ can be written as:
\begin{align}\label{pr_hSigInv}
\hSigma^{-1}=\Sigma^{-1}-\epsilon\Sigma^{-1}E\Sigma^{-1}+O(\epsilon^2)E_2
\end{align}
where $\E_2$ is a $p\times p$ matrix with all entries less than $1$, see \cite{meyer2001matrix}.

Now we consider the denominator of \ref{Psi1}. We first claim the denominator can be written as:
\begin{align}\label{deno1}
\hDelta^T(\inhSigma\Sigma \inhSigma)\hDelta=\hDelta^T\inSigma\hDelta(1+o_p(1)).
\end{align}
Because by \ref{pr_hSigInv}, we have
\begin{align}
\hSigma^{-1}\Sigma\hSigma^{-1}&=(\Sigma^{-1}-\epsilon\Sigma^{-1}E\Sigma^{-1}+O(\epsilon^2)E_2)\Sigma(\Sigma^{-1}-\epsilon\Sigma^{-1}E\Sigma^{-1}+O(\epsilon^2)E_2)\\\nonumber
&=\inSigma-2\epsilon A +\epsilon^2 A \Sigma A+O(\epsilon^2)E_2+O(\epsilon^3)E E_2 \inSigma+O(\epsilon^4)E_2E_2
\end{align}
where $A=\inSigma E \inSigma$.

Also, noticing that $\epsilon=O({1 \over {\sqrt{np}}})$, $k_1\le \lambda_{\min}(\Sigma)\le \lambda_{\max}(\Sigma)\le k_2$ and $\lambda_{\max}(E)\le tr(E)\le p$, we have:
\begin{align}\label{pr_sanwid_1}
\frac{\hDelta^T(\epsilon A)\hDelta}{\hDelta^T \inSigma \hDelta} = \frac{y^T \epsilon E y}{y^T
\Sigma y}\le \frac{\epsilon \lambda_{\max}(E)}{\lambda_{\min}(\Sigma)}\le\epsilon p/\lambda_{\min}(\Sigma)=O(\sqrt{p\over n})
\end{align}
where $y^T=\hDelta^T\inSigma$. Similarly, we have:
\begin{align}\label{pr_sanwid_2}
\frac{\hDelta^T (\epsilon^2 A \Sigma A) \hDelta}{\hDelta^T \inSigma \hDelta}& \le  \epsilon^2 \frac{ \lambda_{\max}^2(E)}{\lambda_{\min}^2(\Sigma)}\le \epsilon^2 p^2/\lambda_{\min}^2(\Sigma)=O({p\over n})
\end{align}
\begin{align}\label{pr_sanwid_3}
\frac{\hDelta^T (O(\epsilon^2)E_2) \hDelta}{\hDelta^T \inSigma \hDelta}& \le O(\epsilon^2)\lambda_{\max}(E_2){\lambda_{\min}(\Sigma)}\le O(\epsilon^2)p\lambda_{\max}(\Sigma)=O({1\over n})
\end{align}

\begin{align}\label{pr_sanwid_4}
\frac{\hDelta^T (O(\epsilon^3)E E_2 \inSigma) \hDelta}{\hDelta^T \inSigma \hDelta} \le O(\epsilon^3) \frac{\lambda_{\max}(E_2)\lambda_{\max}(E)\lambda_{\max}(\inSigma)}{\lambda_{\min}(\Sigma)}\le O(\epsilon^3) p^2\frac{\lambda_{\max}(\Sigma)}{\lambda_{\min}(\Sigma)}=O(\sqrt{p \over n}{1\over n})
\end{align}

\begin{align}\label{pr_sanwid_5}
\frac{\hDelta^T (O(\epsilon^4)E_2 E_2) \hDelta}{\hDelta^T \inSigma \hDelta} \le O(\epsilon^4)p^2\lambda_{\max}(\Sigma)=O_p({1\over n^2})
\end{align}
Since ${p\over n} \to 0$ as $n\to \infty$ and $p\to \infty$, (\ref{deno1}) is derived by combining (\ref{pr_sanwid_1})- (\ref{pr_sanwid_5}).

Now we investigate $\hDelta^T\inSigma\hDelta$ and claim that:
\begin{align}\label{pr_delta_sigma_delta}
\hDelta^T\inSigma\hDelta=\Delta^T\inSigma\Delta(1+o_p(1))+\frac{np}{n_1n_2}(1+o_p(1))
\end{align}

Recall $\hmu_1=\bar{\bY}_{1\cdot}=\frac{1}{n_1}\sum _{i=1}^{n_1}\bY_{1i}$, which is normally distributed as $\mathcal{N}(\bmu_1, \frac{1}{n_1}\Sigma)$. Also, $\hmu_2=\bar{\bY}_{2\cdot}=\frac{1}{n_2}\sum _{i=1}^{n_2}\bY_{2i}$, which is normally distributed as ${N}(\bmu_2, \frac{1}{n_2}\Sigma)$. Let $\hmu_1=\bmu_1+\hat{\beps}_1$ and $\hmu_2=\bmu_2+\hat{\beps}_2$ where $\hat{\beps}_1\sim \mathcal{N}(0,\frac{1}{n_1}\Sigma)$ and $\hat{\beps}_2\sim {N}(0,\frac{1}{n_2}\Sigma)$. Then we have:

\begin{align*}
\hDelta^T\inSigma\hDelta=\bDelta^T\inSigma\bDelta+2\Delta^T\inSigma(\hat{\beps}_1-\hat{\beps}_2)+(\hat{\beps}_1-\hat{\beps}_2)^T\inSigma (\hat{\beps}_1-\hat{\beps}_2)
\end{align*}

Noticing $\hat{\beps}_1-\hat{\beps}_2 \sim N(0,\frac{n}{n_1n_2}\Sigma)$, by Chebyshev's inequality, for any $\epsilon_0>0$

\begin{align}\label{pr_delta_sigma_err}
P(\frac{\bDelta^T\inSigma(\hat{\beps}_1-\hat{\beps}_2)}{\bDelta^T\inSigma\bDelta}>\epsilon_0)
&\le \frac{ E\big( \bDelta^T\inSigma(\hat{\epsilon}_1-\hat{\epsilon}_2) \big)^2 }{(\epsilon_0 \bDelta^T\inSigma\bDelta)^2} \\\nonumber
&=\frac{n}{n_1n_2\bDelta^T\inSigma\bDelta}\\\nonumber
&\le \frac{n}{n_1n_2C_p}=\frac{1}{\pi(1-\pi)nC_p}\to0
\end{align}
It goes to $0$ because $nC_p\to \infty$. Then
\begin{align*}
\bDelta^T\inSigma(\hat{\beps}_1-\hat{\beps}_2)=o_p(\bDelta^T\inSigma\bDelta)
\end{align*}

Then we consider the third term in \ref{pr_delta_sigma_delta}. Let $\tilde{\beps}=\sqrt{\frac{n_1n_2}{n}}\Sigma^{1\over2}(\hat{\beps}_1-\hat{\beps}_2)$. Then $\tilde{\beps}\sim\mathcal{N}(0,I_{p\times p})$. Now for any $\eps_0>0$

\begin{align*}
P(|\frac{(\hat{\beps}_1-\hat{\beps}_2)^T\inSigma (\hat{\beps}_1-\hat{\beps}_2)-np/n_1n_2}{np/n_1n_2}|>\epsilon_0)&=P(|\frac{\tilde{\beps}^T\tilde{\beps}-p}{p}|>\epsilon_0)\\\nonumber
&\le \frac{ E\big( \tilde{\beps}^T\tilde{\beps} \big)^2}{\epsilon_0^2 p^2}\\\nonumber
&={2\over p} {1\over\epsilon^2} \to 0
\end{align*}
as $p\to \infty$.
Then
\begin{align}\label{deno4}
(\hat{\beps}_1-\hat{\beps}_2)^T\inSigma (\hat{\beps}_1-\hat{\beps}_2)=\frac{np}{n_1 n_2}(1+o_p(1))
\end{align}

Then \ref{pr_delta_sigma_delta} followed. Now \ref{deno1} and \ref{pr_delta_sigma_delta} yield:
\begin{align}\label{deno_r}
\hDelta^T(\inhSigma\Sigma \inhSigma)\hDelta=\bDelta^T\Sigma\bDelta(1+o_p(1))+\frac{np}{n_1n_2}(1+o_p(1))
\end{align}

Now we consider the nominator of \ref{Psi1}.
\begin{align}\label{nominator}
(\bmu_1-\hmu)^T\hSigma\hDelta=&{1\over 2}\Big( \bDelta^T\inhSigma \bDelta+\heps_2^T\inhSigma\heps_2-\heps_1^T\inhSigma\heps_1-2\bDelta^T\inhSigma\heps_2\Big)\\\nonumber
&={1\over 2}((1)+(2)-(3)-(4))
\end{align}

\begin{align*}
(1)=\bDelta^T (\inSigma-\epsilon E+O(\epsilon^2)E_2)\bDelta
\end{align*}
By the assumption that $k_1\le\lambda_{\min}(\Sigma)\le\lambda_{\max}\le k_2$, $\lambda_{\max}(E)\le p$ and $\lambda_{\max}(E_2)\le p$, we have
$$\frac{\bDelta^T(\epsilon E)\bDelta}{\bDelta^T \inSigma \bDelta}\to 0$$ and
$$\frac{\bDelta^T(O(\epsilon^2) E_2)\bDelta}{\bDelta^T \inSigma \bDelta}\to 0$$

thus
\begin{align}\label{nominator1}
(1)=\bDelta^T \inSigma \bDelta (1+o_p(1))
\end{align}

By the same argument, we have $(2)=\heps_2^T\inSigma\heps_2(1+o_p(1))$ and $(3)=\heps_1^T\inSigma\heps_1(1+o_p(1))$. Since $\hat{\beps}_1\sim N(0, \frac{1}{n_1}\Sigma)$ and $\hat{\beps}_1\sim N(0, \frac{1}{n_1}\Sigma)$,similar to the proof of (\ref{deno4}), we have:
\begin{align}\label{nominator2}
(2)=\frac{p}{n_2}(1+o_p(1))\ \text{and}\ (3)=\frac{p}{n_1}(1+o_p(1))
\end{align}
Now we consider term (4) in \ref{nominator}.
\begin{align*}
(4)=&\bDelta^T(\inSigma-\epsilon \inSigma E\inSigma+O(\epsilon^2)E_2)\heps_2\\\nonumber
=&\bDelta^T\inSigma \heps_2+ \epsilon \bDelta^T \inSigma E\inSigma \heps_2+ O(\epsilon^2)\bDelta^T E_2 \heps_2
\end{align*}
Similar to the proof of \ref{pr_delta_sigma_err}, all the three terms in (4) are small order of $\bDelta^T\inSigma \bDelta$. Thus we have
\begin{align}\label{nominator4}
(4)=o_p(\bDelta^T\inSigma \bDelta)
\end{align}

Now the nominator can be written as:
\begin{align}\label{nominator_r}
(\bmu_1-\hmu)^T\inhSigma\hDelta= \half \left( \bDelta^T\inSigma \bDelta(1+o_p(1))+\frac{p}{n_1n_2}(n_1-n_2)(1+o_p(1)) \right)
\end{align}

\ref{deno_r} and \ref{nominator_r} yield
\begin{align*}
W_1(\hat{\delta}_{MLE})= 1-\Phi\Big( \frac{\bDelta^T \Sigma^{-1}\bDelta(1+o_p(1))+{p\over{n_1n_2}}(n_1-n_2)(1+o_p(1))}{2\sqrt{\bDelta^T \Sigma^{-1}\bDelta(1+o_p(1))+{{np}\over{n_1n_2}} (1+o_p(1))}} \Big)
\end{align*}

By the same argument, we have:
\begin{align*}
W_2(\hat{\delta}_{MLE})= \Phi\Big( \frac{-\bDelta^T \Sigma^{-1}\bDelta(1+o_p(1))+{p\over{n_1n_2}}(n_1-n_2)(1+o_p(1))}{2\sqrt{\bDelta^T \Sigma^{-1}\bDelta(1+o_p(1))+{{np}\over{n_1n_2}} (1+o_p(1))}} \Big)
\end{align*}

Since ${p\over n}\to 0$ and $C_p=\bDelta^T \Sigma^{-1}\bDelta \to C_0$ with $0\le C_0\le \infty$, we have
\begin{align*}
W(\hat{\delta}_{MLE})
=&{\half}\Big( 1-\Phi\left( \frac{C_p(1+o_p(1))+{p\over{n_1n_2}}(n_1-n_2)(1+o_p(1))}{2\sqrt{C_p(1+o_p(1))+{{np}\over{n_1n_2}} (1+o_p(1))}} \right)\\\nonumber
&+\Phi\left( \frac{-C_p(1+o_p(1))+{p\over{n_1n_2}}(n_1-n_2)(1+o_p(1))}{2\sqrt{C_p(1+o_p(1))+{{np}\over{n_1n_2}} (1+o_p(1))}} \right) \Big)\\\nonumber
\to &1-\Phi(\frac{\sqrt{C_0}}{2})
\end{align*}
as $p\to \infty$ and $n\to \infty$.
If $C_p\to C_0<\infty$, $1-\Phi(\frac{\sqrt{C_0}}{2})>0$. Thus $\hat{\delta}_{MLE}$ is asymptotically optimal. Now we check the asymptotically optimal when $C_p\to \infty$. From the inequality
\begin{align}\label{inequality}
\frac{x}{1+x^2}e^{-\frac{x^2}{2}}\le\Phi(-x)\le \frac{1}{x}e^{-\frac{x^2}{2}},\ x>0
\end{align}
we have
\begin{align*}
\frac{xy}{1+x^2}e^{-\frac{x^2-y^2}{2}} \le \frac{W(\hat{\delta}_{MLE})}{\Phi(-\frac{\sqrt{C_p}}{2})}\le \frac{1+y^2}{xy}e^{-\frac{x^2-y^2}{2}}
\end{align*}
where $x=\frac{C_p(1+o_p(1))\pm{p\over{n_1n_2}}(n_1-n_2)(1+o_p(1))}{2\sqrt{C_p(1+o_p(1))+{{np}\over{n_1n_2}} (1+o_p(1))}}$ and $y=\frac{\sqrt{C_p}}{2}$.
It is easy to check that $\frac{xy}{1+x^2}\to 1$ and $\frac{1+y^2}{xy}\to 1$ as $C_p\to \infty$. Also $x^2-y^2\to 0$ if $C_p(p/n)\to 0$.
This completes the proof.
\end{proof}

\begin{proof}[\textbf{Proof of Theorem \ref{ThErrp>n}}]
The misclassification rate of $\delta_{\hmu}$ is:
\begin{align*}
W(\delta_{\hmu})=\half(W_1(\delta_{\hmu})+W_2(\delta_{\hmu}))
\end{align*}
where
\begin{align*}
W_1(\delta_{\hmu})=1-\Phi(\Psi_1)\ \text{and}\ W_2(\delta_{\hmu})=\Phi(\Psi_2)
\end{align*}
where $\Psi_1$ and $\Psi_2$ is defined by \ref{Psi1} and \ref{Psi2} with $\hSigma$ replaced by $\Sigma$. We start with
\begin{align*}
\Psi_1=\frac{(\bmu_1-\hmu)^T\inSigma(\hmu_1-\hmu_2)}{\sqrt{(\hmu_1-\hmu_2)^T\inSigma(\hmu_1-\hmu_2)}}
\end{align*}
The denominator is (\ref{pr_delta_sigma_delta}) and it can be represented as:
\begin{align*}
\hDelta \inSigma \hDelta=\bDelta^T\Sigma\bDelta(1+o_p(1))+\frac{np}{n_1n_2}(1+o_p(1))
\end{align*}
The nominator is:
\begin{align*}
(\bmu_1-\hmu)^T\inSigma \hDelta=\half(\bDelta^T \inSigma \bDelta+\heps_1^T\inSigma \heps_1-\heps_2^T\inSigma \heps_2-2\bDelta^T\inSigma\heps_2)
\end{align*}
By the similar procedure in the proof of (\ref{nominator_r}), it can be represented by
\begin{align*}
(\bmu_1-\hmu)^T\inSigma \hDelta=\half \left( \bDelta^T\inSigma \bDelta(1+o_p(1))+\frac{p}{n_1n_2}(n_1-n_2)(1+o_p(1)) \right)
\end{align*}
Thus we have:
\begin{align}\label{pr_W1_ineq}
W_1(\hat{\delta}_{\hmu})&= 1-\Phi\Big( \frac{C_p(1+o_p(1))+{p\over{n_1n_2}}(n_1-n_2)(1+o_p(1))}{2\sqrt{C_p(1+o_p(1))+{{np}\over{n_1n_2}} (1+o_p(1))}} \Big)\\\nonumber
\end{align}
Similarly, we have
\begin{align*}
W_2(\hat{\delta}_{\hmu})&= \Phi\Big( \frac{-C_p(1+o_p(1))+{p\over{n_1n_2}}(n_1-n_2)(1+o_p(1))}{2\sqrt{C_p(1+o_p(1))+{{np}\over{n_1n_2}} (1+o_p(1))}} \Big)\\\nonumber
\end{align*}
(i) If $\frac{C_p}{p/n}\to \infty$. Then
\begin{align*}
\frac{\pm C_p(1+o_p(1))+{p\over{n_1n_2}}(n_1-n_2)(1+o_p(1))}{2\sqrt{C_p(1+o_p(1))+{{np}\over{n_1n_2}} (1+o_p(1))}}
&=\frac{\pm C_p(1 \pm {p\over{n_1n_2C_p}}(n_1-n_2)(1+o_p(1)))}{2\sqrt{C_p(1+{{np}\over{n_1n_2C_p}}(1+o_p(1))}}\\\nonumber
&=\frac{\pm\sqrt{C_p}(1 \pm {p\over{n_1n_2C_p}}(n_1-n_2)(1+o_p(1)))}{2\sqrt{(1+{{np}\over{n_1n_2C_p}}(1+o_p(1))}}\\\nonumber
\to & \frac{\pm\sqrt{C_0}}{2}
\end{align*}
which yields $W(\hat{\delta}_{\hmu})\to 0$ since $\frac{p}{n}\to C$ with $0<C<\infty$ and $C_p\to C_0=\infty$.
Now we show that $\frac{W(\hat{\delta}_{\hmu})}{W_{OPT}}\to \infty$ in probability.

Noticing the fact that
\begin{align*}
\frac{x}{1+x^2}e^{-\frac{x^2}{2}}\le\Phi(-x)\le \frac{1}{x}e^{-\frac{x^2}{2}},\ x>0
\end{align*}
we have
\begin{align*}
\frac{W_{OPT}}{\Phi\left(-\frac{x}{2}\right)}=\frac{\Phi(-\frac{\sqrt{C_p}}{2})}{\Phi\left(-\frac{x}{2}\right)} \le
\frac{4+x^2}{x\sqrt{C_p}}e^{-\frac{1}{8}(C_p-x^2)}
\end{align*}
where $x=\frac{C_p(1+o_p(1) \pm\frac{p(n_1-n_2)}{n_1n_2}(1+o_p(1)))}{\sqrt{C_p(1+o_p(1))+\frac{np}{n_1n_2}(1+o_p(1))}}$

\begin{align*}
\frac{4+x^2}{x\sqrt{C_p}}=\frac{4}{x\sqrt{C_p}}+\frac{x}{\sqrt{C_p}} \to \text{a constant}
\end{align*}
because
\begin{align*}
\frac{1}{x\sqrt{C_p}}\to 0
\end{align*}
and
\begin{align*}
\frac{x}{\sqrt{C_p}}=
\begin{cases}
\to 1 & \text{if } c=\infty,\\
\to a_0 & \text{if } c<\infty
\end{cases}
\end{align*}
where $a_0=\frac{\sqrt{c}+1/\sqrt{c}}{\sqrt{c+1/(\pi(1-\pi))}}$.
Also
\begin{align*}
C_p-x^2=\frac{C_p^2o_p(1)+C_p \frac{(n\pm (n_1-n_2))p}{n_1n_2}(1+o_p(1))+\frac{p^2(n_1-n_2)^2}{n_1^2n_2^2}}{C_p(1+o_p(1))+\frac{np}{n_1n_2}(1+o_p(1))}\to\infty
\end{align*}
Thus we have
\begin{align}
\frac{\Phi(-\frac{\sqrt{C_p}}{2})}{\Phi\left(-\frac{x}{2}\right)}\to 0
\end{align}
As a result
\begin{align*}
\frac{W(\hat{\delta}_{\hmu})}{W_{OPT}}\to \infty
\end{align*}

(ii)While $\frac{C_p}{p/n}\to c$ with $0<c<\infty$
\begin{align*}
\frac{\pm C_p(1+o_p(1))+{p\over{n_1n_2}}(n_1-n_2)(1+o_p(1))}{2\sqrt{C_p(1+o_p(1))+{{np}\over{n_1n_2}} (1+o_p(1))}}
&=\frac{\pm C_p(1 \pm {p\over{n_1n_2C_p}}(n_1-n_2)(1+o_p(1)))}{2\sqrt{C_p(1+{{np}\over{n_1n_2C_p}}(1+o_p(1))}}\\\nonumber
&=\frac{\pm\sqrt{C_p}(1 \pm {p\over{n_1n_2C_p}}(n_1-n_2)(1+o_p(1)))}{2\sqrt{(1+{{np}\over{n_1n_2C_p}}(1+o_p(1))}}\\\nonumber
&\to \frac{\pm\sqrt{C_0}(1\pm\frac{1}{c}\frac{2\pi-1}{\pi(1-\pi)})}{2\sqrt{1+\frac{1}{c}\frac{1}{\pi(1-\pi)}}}
\end{align*}
Since $C_p\to C_0$,
\begin{align*}
W_1(\hat{\delta}_{\hmu})\to 1-\Phi(\frac{\sqrt{C_0}(1+\frac{1}{c}\frac{2\pi-1}{\pi(1-\pi)})}{2\sqrt{1+\frac{1}{c}\frac{1}{\pi(1-\pi)}}})
\end{align*}
and
\begin{align*}
W_2(\hat{\delta}_{\hmu})\to 1-\Phi(\frac{\sqrt{C_0}(1-\frac{1}{c}\frac{2\pi-1}{\pi(1-\pi)})}{2\sqrt{1+\frac{1}{c}\frac{1}{\pi(1-\pi)}}})
\end{align*}

If $\frac{p}{n}\to C$ with $0<C<\infty$, then $0<C_0<\infty$. Since $\Phi(x)$ is convex function in the sense that $\half(\Phi(x+\epsilon)+\Phi(x-\epsilon))\le\Phi(x)$ for any $x>0$ and $x>\epsilon>0$,
\begin{align*}
lim_{P} W(\hat{\delta}_{\hmu})=lim_{P} \half(W_1(\hat{\delta}_{\hmu})+W_2(\hat{\delta}_{\hmu}))\ge 1-\Phi(\frac{\sqrt{C_0}}{2\sqrt{1+\frac{1}{c}\frac{1}{\pi(1-\pi)}}})>1-\Phi(\frac{\sqrt{C_0}}{2})
\end{align*}
where $lim_{P}$ means converge in probability with $p\to\infty$ and $n\to \infty$.

If $\frac{p}{n}\to \infty$, then $C_0=\infty$. Hence $W(\hat{\delta}_{\hmu})=\half(W_1(\hat{\delta}_{\hmu})+W_2(\hat{\delta}_{\hmu}))\to 0$. By similar argument in (i), we have $\frac{W(\hat{\delta}_{\hmu})}{W_{OPT}}\to \infty.$

(iii)While $\frac{C_p}{p/n}\to 0$,
\begin{align*}
\frac{\pm C_p(1+o_p(1))+{p\over{n_1n_2}}(n_1-n_2)(1+o_p(1))}{2\sqrt{C_p(1+o_p(1))+{{np}\over{n_1n_2}} (1+o_p(1))}}
&=\frac{\sqrt{{p\over n}}(\pm C_p/({p\over n})+{\frac{n(n_1-n_2)}{n_1n_2}}(1+o_p(1)))}{2\sqrt{(C_p/({p\over n})+{{n^2}\over{n_1n_2}}(1+o_p(1))}}\\\nonumber
& \begin{cases}
   \to  \infty & \text{if } n_1>n_2,\\
    \to -\infty & \text{if } n_1<n_2.
    \end{cases}
\end{align*}
Which yields $W(\hat{\delta}_{\hmu})\to \frac{1}{2}$.
\end{proof}

\begin{proof}[\textbf{Proof of Corollary \ref{CorErrp>n}}]
Noticing that when $n_1=n_2=n/2$, $\heps_1\sim N(0,\frac{1}{n_1}\Sigma)$ and $\heps_2\sim N(0,\frac{1}{n_1}\Sigma)$. Let $\tilde{\epsilon}_i=\sqrt{n_i}\Sigma^{\frac{1}{2}}\hat{\epsilon}_i$ for $i=1,2$. Then $\tilde{\epsilon}_i\sim N(0,I_p)$ and
\begin{align*}
E(\heps_1^T\inSigma\heps_1-\heps_2^T\inSigma\heps_2)^2=&E(\tilde{\epsilon}_1^T\tilde{\epsilon}_1-\tilde{\epsilon}_2^T\tilde{\epsilon}_2)^2/N_1^2
=\sum_{j=1}^{p}(\tilde{\epsilon}_{1j}^2-\tilde{\epsilon}_{2j}^2)^2/n_1^2=6p/n_1^2
\end{align*}

Hence we have:
\begin{align*}
\heps_1^T\inSigma\heps_1-\heps_2^T\inSigma\heps_2=O_p(\frac{\sqrt{p}}{n})
\end{align*}

Similar to the proof of (\ref{pr_W1_ineq}) we have:
\begin{align*}
W_1(\hat{\delta}_{\hmu})&\le 1-\Phi\Big( \frac{C_p(1+o_p(1))+{\frac{\sqrt{p}}{n}}(1+o_p(1))}{2\sqrt{C_p(1+o_p(1))+{{4p}\over{n}} (1+o_p(1))}} \Big)
\end{align*}
and
\begin{align*}
W_2(\hat{\delta_{\hmu}})&\le \Phi\Big( \frac{-C_p(1+o_p(1))+{\frac{\sqrt{p}}{n}}(1+o_p(1))}{2\sqrt{C_p(1+o_p(1))+{{4p}\over{n}} (1+o_p(1))}} \Big)
\end{align*}

\begin{align*}
\frac{\pm C_p(1+o_p(1))+{\frac{\sqrt{p}}{n}}(1+o_p(1))}{2\sqrt{C_p(1+o_p(1))+{{4p}\over{n}} (1+o_p(1))}}=\frac{\pm C_p/\sqrt{p\over n}(1+o_p(1))+{\frac{1}{\sqrt{n}}}(1+o_p(1))}{2\sqrt{C_p/{p\over n}(1+o_p(1))+4+o_p(1))}}\\\nonumber
\begin{cases}
\to \pm \infty & \text{if } \frac{C_p}{\sqrt{p/n}}\to \infty \\
\to \pm \frac{c}{4} & \text{if } \frac{C_p}{\sqrt{p/n}}\to c \text{\ and } p/n\to\infty \\
\to \pm \frac{c}{2\sqrt{4+c/\sqrt{C}}} & \text{if } \frac{C_p}{\sqrt{p/n}}\to c \text{\ and } p/n\to C<\infty \\
\to 0 & \text{if } \frac{C_p}{\sqrt{p/n}}\to 0 \\
\end{cases}
\end{align*}
The proof of $\frac{W(\hat{\delta}_{\hmu})}{W_{OPT}}\to \infty$ is the same as that in the proof of Theorem \ref{ThErrp>n}(1). This completes the proof.

\end{proof}

\section{Proofs for consistency of one-step PMLE}
\begin{proof}[\textbf{Proof of Theorem \ref{Th_consist_PMLE}}]
In the algorithm, we estimate $\hbeta^{(0)}$ first. Then $\htheta^{(0)}$ is estimated by fixing $\bbeta=\hbeta^{(0)}$. Then update $\bbeta=\hbeta^{(1)}$ by fixing $\btheta=\htheta^{(0)}$. $\btheta=\htheta^{(1)}$ is updated in the last step by fixing $\bbeta=\hbeta^{(1)}$. So the idea is to prove the theorem in the following sequence: (a) The consistency and sparsity of $\bbeta^{(0)}$; (b) The consistency of $\htheta^{(0)}$; (c) The consistency and sparsity of $\hbeta^{(1)}$; (d) The consistency of $\htheta^{(1)}$.
\begin{itemize}
\item[(a)] We first prove  $\norm{\hbeta^{(0)}-\bbeta_0}{2}=O_p(\sqrt{\frac{s}{n}})$ and $\hbeta_2^{(0)}=0$ with probability tending to 1, where $\hbeta_2^{(0)}$ is the $p-s$ dimension sub-vector of $\hbeta^{(0)}=(\hbeta_1^{(0)T}, \hbeta_2^{(0)T})^T$. The proof of (a) is the same as the proof of (c), except the loss function is defined as $R(\bbeta)$ which is negative of the penalized MLE function (\ref{PMLE_G}) with covariance matrix $\dSigma$ replaced by $diag_{n-1}(I_p)$. Then the parameters are estimated by minimize the loss function. We omit the proof here and illustrate the details in (c).

\item[(b)] Second we prove $\norm{\htheta^{(0)}-\btheta_0}{2}=O_p(\sqrt{\frac{1}{np}})$. Fixing $\bbeta = \hbeta^{0}$, write
\begin{align*}
F(\btheta,\hbeta^{(0)};\bZ)=-\frac{np}{2}log(2\pi)-\half log\abs{\dot{\Sigma}(\btheta)}-\half (\bZ-\bX\hbeta^{(0)})^T\dot{\Sigma}^{-1}(\btheta)(\bZ-\bX\hbeta^{(0)})
\end{align*}
It is sufficient to prove for any given $\epsilon>0$ the smallest convergence rate of $\eta_{n,p}$ is $\sqrt{\frac{1}{np}}$ such that we have
\begin{align*}
P(\sup_{\norm{\bu}=C}F(\btheta_0+\bu\eta_{n,p},\hbeta^{(0)};\bZ)<F(\btheta_0,\hbeta^{(0)};\bZ))>1-\epsilon
\end{align*}
This implies there exists a local maximum for the function $Q(\btheta,\hbeta^{(0)};\bZ)$ of $\btheta$ in the neighborhood of $\btheta_0$ with the radius at most proportional to $\eta_{n,p}$.

By Taylor's expansion, $\dot{\Sigma}(\btheta_0+\bu\eta_{n,p})-\dot{\Sigma}(\btheta_0)=\sum_{j=1}^{q}\der{\dot{\Sigma}(\btheta^*)}{\theta_j}u_{\theta_j}\eta_{n,p}$, where $\btheta^*$ is between $\btheta_0+\bu\eta_{n,p}$ and $\btheta_0$. Denote $\dot{\Sigma}^j(\btheta^*)=\der{\dot{\Sigma}(\btheta^*)}{\theta_j}$, then
\begin{align*}
&F(\btheta_0+\bu\eta_{n,p},\hbeta^{(0)};\bZ)-F(\btheta_0,\hbeta^{(0)};\bZ)\\\nonumber
=&\left[ F(\btheta_0+\bu\eta_{n,p},\bbeta_0)-F(\btheta_0,\bbeta_0)\right]-\sum_{j=1}^{q}(\bZ-\bX\bbeta_0)^T\dot{\Sigma}^j(\btheta^*)\bX(\hbeta^{(0)}-\bbeta_0)u_{\theta_j}\eta_{n,p}\\\nonumber
&-\half\sum_{j=1}^{q}(\hbeta^{(0)}-\bbeta_0)^TX^T \dot{\Sigma}^j(\btheta^*) X (\hbeta^{(0)}-\bbeta_0)u_{\theta_j}\eta_{n,p}\\\nonumber
=&(I)+(II)+(III)
\end{align*}
where
\begin{align*}
(I)=& F(\btheta_0+\bu\eta_{n,p},\bbeta_0)-F(\btheta_0,\bbeta_0)\\\nonumber
=&-\frac{n-1}{2}\bu^T T \bu \eta_{n,p}^2+ \left(\der{F}{\btheta}(\btheta_0)\right)^T\bu \eta_{n,p}+ \frac{1}{2}\bu^T (\frac{\partial^2 F}{\partial \btheta \partial \btheta^T}(\btheta^*)+(n-1)T) \bu \eta_{n,p}^{2}\\\nonumber
=&(1)+(2)+(3)\\\nonumber
(II)=&\sum_{j=1}^{q}(\hbeta^{(0)}-\bbeta_0)^T\bX^T\dot{\Sigma}^j(\btheta^*)(\bZ-\bX\bbeta_0)u_j\eta_{n,p}\\\nonumber
(III)=&-\half \sum_{j=1}^{q}(\hbeta^{(0)}-\bbeta_0)^T \bX^T \dot{\Sigma}^j(\btheta^*)\bX(\hbeta^{(0)}-\bbeta_0)u_j\eta_{n,p}
\end{align*}
We consider (I) first. $T$ in $(1)$ is $q\times q$ matrix with its $(i,j)$th element as $t_{ij}(\btheta_0)$, where $t_{ij}(\btheta)=tr(\inSigma(\btheta)\Sigma_i(\btheta)\inSigma(\btheta)\Sigma_j(\btheta))$. By A\ref{asu_covnorm_inv} and A\ref{asu_tij_lim} and using the similar argument in proving the bound of $(I)$ in theorem \ref{ThconsistMLE}, there exist a constant $K$, such that $(1)=-\frac{n}{2}\sum_{i,j=1}^{q}t_{ij}(\btheta_0)u_iu_j\eta_{n,p}^2\le -Knp\eta_{n,p}^2\norm{\bu}{2}^2$ with probability tending to 1. In regarding to (2),
\begin{align*}
\der{F}{\theta_j}(\btheta_0)=\frac{n-1}{2}tr(\Sigma^j(\btheta_0)\Sigma(\btheta_0))-\half(\bZ-\bX\bbeta)^T\dSigma^j(\btheta_0)(\bZ-\bX\bbeta)
\end{align*}
Notice that $\bZ-\bX\bbeta\sim N(0,\dSigma(\btheta))$ and $tr(\dSigma^j\dSigma)=(n-1)tr(\Sigma^j\Sigma)$. By lemma $\ref{lemma_sumEps_sq}$,
\begin{align*}
(2)=&O_p(\sqrt{tr(\dSigma\dSigma^j\dSigma\dSigma^j)}\eta_{n,p}\norm{\bu}{2})=O_p(\sqrt{(n-1)tr(\Sigma\Sigma^j\Sigma\Sigma^j)}\eta_{n,p}\norm{\bu}{2})\\\nonumber
=&O_p(\sqrt{(n-1)\norm{\Sigma^j}{F}^2}\eta_{n,p}\norm{\bu}{2})\norm{\bu}{2})
\end{align*}
By A\ref{asu_eigen_sig}, A\ref{asu_star_coveigen}, $(2)=O_p(\sqrt{np}\eta_{n,p})$.

Then we consider (3). For any $j,k=1,2,...,q$
\begin{align*}
\frac{\partial^2(F)}{\partial \theta_j \partial \theta_k}(\btheta^*)=\frac{n-1}{2}(tr(\Sigma^{jk}(\btheta^*)\Sigma(\btheta^*)-t_{jk}(\btheta^*))-\half (\bZ-\bX\bbeta)^T\dSigma^{jk}(\btheta^*)(\bZ-\bX\bbeta)
\end{align*}
Thus $(3)$ could be written as:
\begin{align*}
(3)=&\sum_{j,k=1}^{q}\frac{n-1}{2}(tr(\Sigma^{jk}(\btheta^*)\Sigma(\btheta_0))-\half(\bZ-\bX\bbeta)^T \Sigma^{jk}(\btheta^*)(\bZ-\bX\bbeta))u_{\theta_j}u_{\theta_k}\eta_{n,p}^2\\\nonumber
&+\sum_{j,k=1}^{q}\frac{n-1}{2}(tr(\Sigma^{jk}(\btheta^*)\Sigma(\btheta^*))-tr(\Sigma^{jk}(\btheta^*)\Sigma(\btheta_0)))u_{\theta_j}u_{\theta_k}\eta_{n,p}^2\\\nonumber
&+\sum_{j,k=1}^{q}\frac{n-1}{2}(t_{jk}(\btheta_0)-t_{jk}(\btheta^*))u_{\theta_j}u_{\theta_k}\eta_{n,p}^2\\\nonumber
=&(i)+(ii)+(iii)
\end{align*}
By lemma \ref{lemma_sumEps_sq} and A\ref{asu_star_coveigen2}, $(i)=O_p(\sqrt{(n-1)tr(\Sigma\Sigma^{jk}\Sigma\Sigma^{jk})}\eta_{n,p}^2)=O_p(\sqrt{np}\eta_{n,p}^2)$.
Similar to the deriving the order of (5) and (6) in the proof of Theorem\ref{ThconsistMLE}, $(ii)=O_p(n\sqrt{p^3}\eta_{n,p}^3)$ and $(iii)=O_p(np\eta_{n,p}^3)$. By choosing large $C=\norm{\bu}{2}$, the minimal rate that $(2)$ and $(3)$ are dominated by $(1)$ is $\eta_{n,p}=O_p(\frac{1}{\sqrt{np}})$.

Now we consider $(II)$. Denote $B=\sum_{i=1}^{n-1}\bX_i$. Then by lemma \ref{lemma_sumEps}, for any $j=1,2,...,q$,
\begin{align*}
&(\hbeta^{(0)}-\bbeta_0)^T\bX^T\dot{\Sigma}^j(\btheta^*)(\bZ-\bX\bbeta_0)=O_p(\sqrt{tr(\bX^T \dSigma^j(\btheta^*)\dSigma(\btheta^*)\dSigma^j(\btheta^*)\bX)}\norm{\hbeta^{(0)}-\bbeta_0}{2})\\\nonumber
=&O_p(\sqrt{tr(\sum_{i=1}^{n-1}\bX^T_j\Sigma^{j*}\Sigma^*\Sigma^{j*}\bX_i)+tr(B^T \Sigma^{j*}\Sigma^*\Sigma^{j*}B)}\norm{\hbeta^{(0)}-\bbeta_0}{2})
\end{align*}
where $\Sigma^*=\Sigma(\btheta^*)$ and $\Sigma^{j*}=\Sigma^j(\btheta^*)$. Notice that $\sum_{i=1}^{n-1}\bX_i^T\bX_i=(\frac{n_1n_2}{n}-\frac{n_1^2}{n^2})I_{p\times p}$ and $(\sum_{i=1}^{n-1}\bX_i)^T(\sum_{i=1}^{n-1}\bX_i)=\frac{n_1^2}{n^2}I_{p\times p}$. Then by A\ref{asu_star_coveigen}
\begin{align*}
tr(\sum_{i=1}^{n-1}\bX^T_j\Sigma^{j*}\Sigma^*\Sigma^{j*}\bX_i)=(\frac{n_1n_2}{n}-\frac{n_1^2}{n_2})tr(\Sigma^{j*}\Sigma^*\Sigma^{j*})\le \lambda_{\max}(\Sigma)\frac{n_1n_2}{n}\norm{\Sigma^{j*}}{F}^2=O_p(np).
\end{align*}
Similarly $tr(B^T\Sigma^{j*}\Sigma^*\Sigma^{j*}B)=O_p(np)$.

Since $\norm{\bbeta^{(0)}-\bbeta_0}{2}=O_p(\sqrt{\frac{s}{n}})$, $(II)=O_p(\sqrt{ps}\eta_{n,p})$

For (III), by A\ref{asu_star_coveigen}, for any $j=1,2,...,q$ we also have:
\begin{align*}
&(\hbeta^{(0)}-\bbeta_0)^T \bX^T \dot{\Sigma}^j(\btheta^*)\bX(\hbeta^{(0)}-\bbeta_0)\\\nonumber
=& (\hbeta^{(0)}-\bbeta_0)^T \bX^T diag_{n-1}{\Sigma}^j(\btheta^*) (\tilde{I}_{n-1,p}+\tilde{J}_{n-1,p})\bX (\hbeta^{(0)}-\bbeta_0)\\\nonumber
\le& \lambda_{\max}(\Sigma^{j*})((\hbeta^{(0)}-\bbeta_0)^T\left(\sum_{i=1}^{n-1}\bX_i^T\bX_i +(\sum_{i=1}^{n-1}\bX_i)^T (\sum_{i=1}^{n-1}\bX_i)\right)(\hbeta^{(0)}-\bbeta_0))\\\nonumber
=&O_p(n\norm{\hbeta^{(0)}-\bbeta_0}{2}^2)=O_p(s)
\end{align*}
Thus $(III)=O_p(s\eta_{n,p})$. Both $(II)$ and $(III)$ are dominated by $(I)$ while $\eta_{n,p}=O_p(\sqrt{\frac{1}{np}})$. This concludes the proof  of $\norm{\htheta^{(0)}-\btheta_0}{2}=O_p(\sqrt{\frac{1}{np}})$

\item[(c)]Write $\hbeta^{(1)}=(\hbeta_1^{(1)},\hbeta_2^{(2)})^T$. Then we prove $\norm{\hbeta^{(1)}-\bbeta_0}{2}=O_p(\sqrt{\frac{s}{n}})$ and $\hbeta_2^{(1)}=0$ with probability tending to 1, where $\hbeta^{1}_{(1)}$ formed by elements in $supp(\bbeta_0)$ and $\hbeta_2^{(1)}$ is a $p-s$ sub-vector of $\hbeta^{(1)}$.
Let $\eta_{n,p}=O_p(\norm{\htheta^{(0)}-\btheta_0}{2})=O_p(\sqrt{\frac{1}{np}})$. We use two steps to prove the consistency and sparsity.

\textit{step 1} We first prove consistency on s-dimensional space. Define the log-likelihood function for $\bbeta^1$ as
\begin{align*}
\bar{Q}(\htheta^{(0)},\bbeta^1)=-\frac{np}{2}log(2\pi)-\half log\abs{\dSigma(\htheta^{(0)})}-\half(\bZ-\bX^1\bbeta^1)^T \indSigma(\htheta^{(0)})(\bZ-\bX^1\bbeta^1)-n\sum_{j=1}^{n}P_{\lambda}(\abs{\beta_j})
\end{align*}
where $\bbeta^1$ is sub-vector of $\bbeta_0=(\bbeta^{1T},\bbeta^{2T})^T$ formed by elements in supp $\bbeta_0$. We first $\norm{\hbeta^{(1)}_1-\bbeta^{1}}{2}=O_p(\sqrt{\frac{s}{n}})$. It is sufficient to prove that for any $\epsilon>0$, the smallest rate of $\xi_{n,p}$ is $\sqrt{\frac{s}{n}}$ such that we have:
\begin{align*}
P(\sup_{\norm{\bu}{2}=C}\bar{Q}(\htheta^{(0)},\bbeta^1+u\xi_{n,p})<\bar{Q}(\htheta^{(0)},\bbeta^1))>1-\epsilon
\end{align*}
where $\bu\in\IR^s$. This implies that with probability tending to 1, there is a local maximizer $\hbeta^{(1)}_1$ of the function $\bar{Q}$ in the neighborhood of $\bbeta_0^1$ with the radius of $\bbeta_0^1$ at most proportional to $\xi_{n,p}$.

\begin{align*}
&\bar{Q}(\htheta^{(0)},\bbeta_0^1+u\xi_{n,p})-\bar{Q}(\htheta^{(0)},\bbeta_0^1)\\\nonumber
=&-\half \bu^T \bX^{1T}\indSigma(\htheta^{(0)})\bX^1\bu\xi_{n,p}^2-(\bZ-\bX^1\bbeta_0^1)\indSigma(\htheta^{(0)})\bX^1\bu\xi_{n,p}\\\nonumber
&-np\sum_{j=1}^{m}\big( p_{\lambda_{n,p}}^{'}(|\beta_{0j}|)sgn(\beta_j)u_{\beta_j}\xi_{n,p}+ p_{\lambda_{n,p}}^{''} (|\beta_{0j}|)u_{\beta_j}^{2}\xi_{n}^2)(1+o(1))\big)\\\nonumber
=&(I)+(II)+(III)
\end{align*}
By Taylor's expansion, $\dSigma^{-1}(\htheta^{(0)})=\dSigma^{-1}(\btheta_0)+\sum_{j=1}^{q}\dSigma^j(\btheta*)$, where $\btheta*$ is a $q$ dimension vector between $\btheta_0$ and $\htheta^{(0)}$. Therefore,
\begin{align*}
(I)=&-\half \bu^T \bX^{1T}\indSigma(\btheta_0)\bX^1\bu\xi_{n,p}^2-\half \sum_{j=1}^{q}\bu^T \bX^{1T}\dSigma^j(\btheta^*)\bX^1\bu\xi_{n,p}^2u_j\eta_{n,p}\\\nonumber
=&(1)+(2)
\end{align*}
and
\begin{align*}
(II)=&-(\bZ-\bX^1\bbeta_0^1)\indSigma(\htheta^{(0)})\bX^1\bu\xi_{n,p}+\sum_{j=1}^{q}(\bZ-\bX^1\bbeta_0^1)\dSigma^j(\btheta^*)\bX^1\bu\xi_{n,p}u_j\eta_{n,p}\\\nonumber
=&(3)+(4)
\end{align*}

Since $\lambda_{\min}(\inSigma)>0$ and notice $\sum_{i=1}^{n-1}\bX_i^{1T}{\bX_i^1}=(\frac{n_1n_2}{n}-\frac{n_1^2}{n^2})I_s$, $\left( \sum_{i=1}^{n-1}\bX_{i=1}^{1T}\right) \left( \sum_{i=1}^{n-1}\bX_{i=1}^{1} \right)=\frac{n_1^2}{n^2}I_{s}$
\begin{align*}
(1)=&\half \bu^T \bX^{1T}diag_{n-1}(\inSigma)(\tilde{I}_{n-1,p}+\tilde{J}_{n-1,p}) \bX^1 \bu \xi_{n,p}^2\\\nonumber
=&-\half \bu^T \sum_{i=1}^{n-1}\bX_i^{1T}\inSigma\bX_i^1\bu \xi_{n,p}^2-\half \bu^T \sum_{i=1}^{n-1}\bX_i^{1T}\inSigma\sum_{i=1}^{n-1}\bX_i^1\bu \xi_{n,p}^2\\\nonumber
\le &-\half \bu^T \left(\sum_{i=1}^{n-1}\bX_i^{1T}\bX_i^1+\sum_{i=1}^{n-1}\bX_i^{1T}\sum_{i=1}^{n-1}\bX_i^1\right) \bu \xi_{n,p}^2\lambda_{\min}(\inSigma)\\\nonumber
= &-\half\frac{n_1n_2}{n}\norm{\bu}{2}^2\xi_{n,p}^2\lambda_{\min}(\inSigma)\\\nonumber
\le &-\half\frac{\pi(1-\pi)}{\lambda_{\max}(\Sigma)}n\xi_{n,p}^2\norm{\bu}{2}^2
\end{align*}
By similar argument and A\ref{asu_star_coveigen}, $(2)=O_p(n\xi_{n,p}^2\eta_{n,p})=o_p((1))$  while $\eta_{n,p}=O_p(\sqrt{\frac{1}{np}})$.
For (II), by lemma \ref{lemma_sumEps} and A\ref{asu_eigen_sig},
\begin{align*}
(3)=&O_p(\sqrt{tr(\bX^1\indSigma(\btheta_0)\bX^1)}\norm{\bu}{2}\xi_{n,p})\\\nonumber
=& O_p(\sqrt{\lambda_{\max}(\Sigma)\frac{n_1n_2}{n}tr(I_{s\times s})}\norm{\bu}{2}\xi_{n,p})\\\nonumber
=& O_p(\sqrt{ns}\norm{\bu}{2}\xi_{n,p}))
\end{align*}
Similarly, $(4)=O_p(\sqrt{ns}\norm{\bu}{2}\xi_{n,p}\eta_{n,p})=o_p((3))$. So we have $(II)=O_p(\sqrt{ns}\norm{\bu}{2}\xi_{n,p}))$.

$(III)=(5)+(6)$, where
\begin{align}
(5)&=-n\sum_{j=1}^{s} p_{\lambda_{n,p}}^{'}(|\beta_{0j}|)sgn(\beta_j)u_{\beta_j}\xi_{n,p}\\\nonumber
(6)&=-n \sum_{j=1}^{s} p_{\lambda_{n,p}}^{''} (|\beta_{0j}|)u_{\beta_j}^{2}\xi_{n}^2)(1+o(1))
\end{align}
Since $a_{n,p}=O_p(\frac{1}{\sqrt{n}})$ by A\ref{asu_penfunc1},
\begin{align}
|(5)|&\le n\sqrt{s} a_{n,p}\norm{\bu}{2}=O_p(\sqrt{ns}\xi_{n,p}\norm{\bu}{2})\\\nonumber
\end{align}
By A\ref{asu_penfunc2}
\begin{align}
|(6)|&\le 2n \xi_{n,p}^2 \sum_{j=1}^{s}p^{''}(\beta_{0j})u_{\beta_j}^2\le2n\xi_{n,p}^2b_{n,p}\norm{\bu}{2}^2\\\nonumber
&=o_p(n\xi_{n,p}^2)
\end{align}

By choosing large $C=\norm{\bu}{2}$, the smallest rate of $\xi_{np}$ that $(II)$ and $(III)$ are dominated by $(I)$ is $\xi_{n,p}=O_p(\sqrt{\frac{s}{n}})$. This completes the proof that $\norm{\hbeta^{(1)}_1-\bbeta_0^1}{2}=O_p(\sqrt{\frac{s}{n}})$.

\textit{step 2}. in step 2 we prove that the vector $\hbeta=(\hbeta_1^{(1)},0)$ is a strict local maximizer on $\IR^d$. It is sufficient to prove for any given $\bbeta\in \IR^d$ satisfying $\norm{\bbeta-\bbeta_0}{2}=O_p(\sqrt{\frac{s}{n}})$, we have $Q(\bbeta^s,\htheta^{(0)})\ge Q(\bbeta,\htheta^{(0)})$, where $\bbeta=(\bbeta^{1T},\bbeta^{2T})^T$ and $\bbeta^s=(\bbeta^{1T},0^T)^T$.

Let $\epsilon=C\sqrt{\frac{s}{n}}$, it is sufficient to prove for $j=s+1,s+2,...,p$:
\begin{align}\label{step2dQ}
&\der{Q(\bbeta,\htheta^{(0)})}{\beta_j}<0 \  for\   0<\beta_j<\epsilon\\\nonumber
&\der{Q(\bbeta,\htheta^{(0)})}{\beta_j}>0 \  for\   -\epsilon<\beta_j<0
\end{align}
\begin{align}
\der{Q(\bbeta,\htheta^{(0)})}{\bbeta_j}=&(\bZ-\bX\bbeta_0)^T\dSigma^{-1}(\htheta^{(0)})\bX_{j}+\sum_{l=1}^{p}\bX_{l}^T\indSigma(\htheta^{(0)})\bX_{j}(\beta_{l}-\beta_{0l})\\\nonumber
&-nP'_{\lambda}(|\beta_j|)sgn(\beta_j)\\\nonumber
=&(I)+(II)+(III)
\end{align}
where $\bX_{j}$ is the $j$th column of $\bX$.

We first consider (I). By Taylor's expansion,
\begin{align*}
(I)=&(\bZ-\bX\bbeta_0)^T\indSigma(\btheta_0)\bX_j+\sum_{k=1}^{q}(\bZ-\bX\bbeta_0)^T\dSigma^k(\btheta^*)\bX_j u_k\eta_{n,p}\\\nonumber
=&(5)+(6)
\end{align*}

Notice $\sum_{i=1}^{n-1}\bX_{ij}^T\bX_{ij}=\frac{n_1n_2}{n}-\frac{n_1^2}{n^2}$ and $(\sum_{i=1}^{n-1}\bX_{ij})^T (\sum_{i=1}^{n-1}\bX_{ij})=\frac{n_1^2}{n^2}$.
\begin{align*}
(5)=&(\bZ-\bX\bbeta_0)^T diag_{n-1}(\inSigma)(\tilde{I}_{n-1,p}+\tilde{J}_{n-1,p})\bX_j=O_p(\sqrt{\bX_j\indSigma\bX_j})\\\nonumber
=&O_p(\sqrt{\sum_{i=1}^{n-1}\bX_{ij}^T\inSigma\bX_{ij}+\sum_{i=1}^{n-1}\bX_{ij}^T\inSigma\sum_{i=1}^{n-1}\bX_{ij}})
\end{align*}
where $\bX_{ij}$ is the $j$th column of $\bX_i$.

Noticing $\sum_{i=1}^{n-1}\bX_{ij}^T\bX_{ij}=\frac{n_1n_2}{n}-\frac{n_1^2}{n^2}$, $(\sum_{i=1}^{n-1}\bX_{ij})^T(\sum_{i=1}^{n-1}\bX_{ij})=\frac{n_1^2}{n^2}$ and $\lambda_{\max}(\inSigma)\le \infty$, we have $(5)=O_p(\sqrt{n})$. Similarly, $(6)=O_p(\sqrt{n}\eta_{n,p})=o_p((3))$, which is dominated by (5) if $\eta_{n,p}=o(1)$.

For (II), by Taylor's expansion,
\begin{align*}
(II)=&\sum_{l=1}^{p}\bX_{l}^T\indSigma(\theta_0)\bX_{j}(\beta_{l}-\beta_{0l})+\sum_{k=1}^{q}\sum_{l=1}^{p}\bX_{l}^T\dSigma^k(\theta^{*})\bX_{j}(\beta_{l}-\beta_{0l})(\theta_k^*-\theta_{0k})\\\nonumber
=&(7)+(8)
\end{align*}
\begin{align*}
(7)=&\sum_{l=1}^{p}\bX_{l}^T diag_{n-1}(\inSigma)(\tilde{I}_{n-1,p}+\tilde{J}_{n-1,p}) \bX_{j}(\beta_{l}-\beta_{0l})\\\nonumber
=&\sum_{l=1}^{p}\left( \sum_{i=1}^{n-1}\bX_{il}^T\inSigma\bX_{ij} +   \sum_{i=1}^{n-1}\bX_{il}^T\inSigma\sum_{i=1}^{n-1}\bX_{ij}  \right) (\beta_{l}-\beta_{0l})
\end{align*}
Notice
\begin{align*}
&\sum_{i=1}^{n-1}\bX_{il}^T\inSigma\bX_{ij}\le \sum_{i=1}^{n-1}(\bX_{il}^T\inSigma\bX_{il}))^{1\over 2}(\bX_{ij}^T\inSigma\bX_{ij}))^{1\over 2}\\\nonumber
\le& \sum_{i=1}^{n-1} \lambda_{\max}(\inSigma)(\bX_{il}^T\bX_{il}))^{1\over 2}(\bX_{ij}^T\bX_{ij}))^{1\over 2}\\\nonumber
=&\lambda_{\max}(\inSigma)(\frac{n_1n_2}{n}-\frac{n_1^2}{n^2})=O_p(n).
\end{align*}
Also, let $B_l=\sum_{i=1}^{n-1}\bX_{il}$. Then $B_l^TB_l=\left ( \frac{n_1}{n} \right )^2$.
\begin{align*}
\sum_{i=1}^{n-1}\bX_{il}^T\inSigma\sum_{i=1}^{n-1}\bX_{ij}=B_l^T\inSigma B_j\le (B_l^T\inSigma B_l)^{1/2}(B_j^T\inSigma B_j)^{1/2}\le  \lambda_{\max}(\inSigma)\frac{n_1^2}{n^2}
\end{align*}

Then
\begin{align}
(7)=O_p(n\norm{\bbeta-\bbeta_0})= O_P(\sqrt{ns})
\end{align}
Similarly, $(8)=O_p(\sqrt{ns}\eta_{n,p})$
Thus
\begin{align}\label{step2Q}
\der{Q(\bbeta)}{\beta_{j}}=n\lambda_{n,p} \Big( O_P(\frac{\sqrt{s}}{\sqrt{n}\lambda_{n,p}}))+\frac{P_{\lambda_{n,p}}^{'}(|\beta_j|)}{\lambda_{n,p}}sgn(\beta_j) \Big)
\end{align}

By assumption \ref{asu_sparsity_1} and \ref{asu_sparsity_2}, the sign of \ref{step2Q} is determined by $\beta_j$, hence \ref{step2dQ} followed. This implies $\hbeta^{(1)}$ should satisfy sparse property and completes the proof of step 2.

\item[(d)]Lastly we prove $\norm{\htheta^{(1)}-\btheta_0}{2}=O_p(\sqrt{\frac{1}{np}})$. Since $\norm{\hbeta^{(1)}-\bbeta_0}{2}=O_p(\sqrt{\frac{s}{n}})$, it is the same as the proof of (b). We omit the detail here and this completes the proof.
\end{itemize}
\end{proof}

\section{Proofs for consistency for one-step PMLE with tappering}
\begin{lemma}\label{lem_tapermatrix}
Assume A\ref{asu_domain}, A\ref{asu_covfunc}, A\ref{asu_tapperRange} and A\ref{asu_covfunc_int} hold, we have:
\begin{itemize}
\item [(1)]$\norm{\Sigma(\btheta)-\Sigma(\btheta)_T}{\infty}=O(\frac{1}{p^{\delta_0}})$;
\item [(2)]$\norm{\Sigma_k(\btheta)-\Sigma_{k,T}(\btheta)}{\infty}=O(\frac{1}{p^{\delta_0}})$;
\item [(3)]$\norm{\Sigma_{jk}(\btheta)-\Sigma_{jk,T}(\btheta)}{\infty}=O(\frac{1}{p^{\delta_0}})$.
\end{itemize}
The matrix norm $\norm{\cdot}{\infty}$ for the $p\times p$ matrix $A=[a_{ij}]_{i,j=1}^{p}$ is defined as the maximum of row sumation, i.e. $\norm{A}{\infty}=\max_{i}\sum_{j=1}^{p}\abs{a_{ij}}$
\end{lemma}
\begin{proof}
We show (1) in detail and omit the details for (2) and (3), as similar arguments can be applied.
\begin{align}
\norm{\Sigma(\btheta)-\Sigma_T(\btheta)}{\infty}=\max_{i}\sum_{j=1}^{p}\abs{\gamma(h_{ij};\btheta)K_T(h_{ij},w_{p})-\gamma(h_{ij};\btheta)}
\end{align}
where $h_{ij}=\norm{s_i-s_j}{2}$ is the distance between site $s_i$ and $s_j$.
For any $i=1,2,...,p$,
\begin{align}
&\sum_{j=1}^{p}\abs{\gamma(h_{ij};\btheta)K_T(h_{ij},w_{p})-\gamma(h_{ij};\btheta)}\\\nonumber
&\le \sum_{h_{ij}<w_{p}}\abs{\gamma(h_{ij};\btheta)K_T(h_{ij},w_{p})-\gamma(h_{ij};\btheta)} +\sum_{h_{ij}\ge w_{p}}\gamma_0(\btheta,h_{ij})\\\nonumber
&=(I)+(II)
\end{align}

Let $A^i=\{j:h_{ij}>w_{p}\}$ and $B_{m}^{i}=\{j:(m-1)\delta\le h_{ij}<m\delta\}$, where $\Delta$ is independent of $n$ and $p$. Then $A^i\subset \bigcup_{m=\lfloor\frac{w_{p}}{\Delta}\rfloor}^{\infty}B_{m}^{i}$. Let $V(R)$ be the volume of a d-dimensional baa of radius $R$. Then the volume of $B_m^i$ is $B_m^i=V(m\delta)-V((m-1)\delta)=f_{d-1}(m)\delta^d$, where $f_{d-1}(m)$ is a polynomial function of $m$ with degree of $d-1$. By A\ref{asu_domain}, the number of sites in any unit subset of $D \subset \IR^d$ is bounded, say $\rho$. Let $\#\{A\}$ denote the cardinality of a discrete set $A$. Then we have $\#\{B_{m}^{i}\}\le f_{d-1}(m)\delta^d\rho$. Then exist a constant $K$ such that $f_{d-1}(m)\le Km^{d-1}$ Then
\begin{align}
(II)=&\sum_{h_{ij}\ge w_{p}}\abs{\gamma(\btheta,h_{ij})}\le\sum_{m=\lfloor\frac{w_{p}}{\delta}\rfloor}^{\infty}\sum_{j\in B_{m}^{i}}\abs{\gamma(\btheta,h_{ij})}\\\nonumber
&\le K\rho \sum_{m=\lfloor\frac{w_{p}}{\delta}\rfloor}^{\infty} m^{d-1}\delta^{d} \max_{j\in B_m^i}\abs{\gamma(\btheta,h_{ij})}\\\nonumber
&\le K\rho\int_{w_{p}}^{\infty}x^{d-1}\abs{\gamma(\btheta,x)}dx \le  \frac{K\rho}{w_{p}}\int_{0}^{\infty}x^d \abs{\gamma(\btheta,x)}dx
\end{align}

Let $A_2^i=\{j:h_{ij}\le w_{p}\}$. Then $A_{2}^{i}\subset \bigcup_{m=1}^{\lfloor \frac{w_{p}}{\delta}\rfloor+1}$
\begin{align}
(I)=&\sum_{h_{ij}<w_{p}}\abs{\gamma(h_{ij};\btheta)-\gamma(h_{ij};\btheta)K_T(h_{ij},w_{p})}\\\nonumber
&=2\sum_{h_{ij}<w_{p}} \abs{\gamma(h_{ij};\btheta)}\frac{h_{ij}}{w_{p}}\\\nonumber
&\le\frac{2}{w_{p}}\sum_{m=1}^{\lfloor \frac{w_{p}}{\delta}\rfloor+1}\sum_{j\in B_{m}^{i}}h_{ij}\abs{\gamma(\btheta,h_{ij})}\\\nonumber
&\le \frac{2K\rho}{w_{p}}\sum_{m=1}^{\lfloor \frac{w_{p}}{\delta}\rfloor+1}(m\delta)^{d-1}
\delta \max_{j\in B^i_m}h_{ij}\abs{\gamma(\btheta,h_{ij})}\\\nonumber
&\le \frac{2K\rho}{w_{p}}\int_{0}^{\infty}x^d \abs{\gamma(\btheta,x)}dx\\\nonumber
&\le \frac{2K\rho}{w_{p}}\int_{0}^{\infty}x^d \gamma_0(\btheta,x)dx
\end{align}
$w_p$ has the same order as $p^{1/2}$ by A\ref{asu_tapperRange}. By A\ref{asu_covfunc_int}, both $(I)$ and $(II)=O(1/p^{1/2})$. This completes the proof.
\end{proof}

\begin{lemma}\label{lem_taper_cov_asu}
Assume A\ref{asu_eigen_sig}-A\ref{asu_cov_star}, A \ref{asu_tapperRange}, A\ref{asu_covfunc_int} hold, we have
\begin{itemize}
\item[(a)] $\lim_{p\to \infty}\lambda_{\min}(\Sigma_T)>0$, $\lim_{p\to \infty}\lambda_{\max}(\Sigma_T)< \infty$;

\item[(b)] There exists an open subset $\omega$ that contains the true parameter $\btheta_0$ such that for all $\btheta^*\in \omega$, we have:
\begin{itemize}
 \item[(i)] $-\infty< \lim_{p\to \infty}\lambda_{\min}(\Sigma_{T}^k(\btheta^*)) <\lim_{p\to \infty}\lambda_{\max}(\Sigma_{T}^k(\btheta^*))<\infty$;
 \item[(ii)] $-\infty< \lim_{p\to \infty}\lambda_{\min}(\Sigma_{T}^{kj}(\btheta^*)) <\lim_{p\to \infty}\lambda_{\max}(\Sigma_{T}^{kj}(\btheta^*))<\infty$;
 \item[(iii)]$\abs{\der{t_{ij,T(\btheta^*)}}{\theta_k}}=O(p)$ for all $k=1,2,...,q$.
\end{itemize}
\end{itemize}
\end{lemma}

\begin{proof}
\begin{itemize}
\item[(a)] Let $K_T=[K(h_{ij},w)]_{i,j=1}^{p}$ be the tappering covariance. By eigenvalue inequalities of Schur product:
\begin{align}
\min_{1\le i\le p}a_{ii}\lambda_{\min}(\Sigma) \le \lambda(\Sigma \circ K_T)\le \max_{1\le i\le p}a_{ii}\lambda_{\max}(\Sigma)
\end{align}
where $a_{ij}$ are the $(i,j)th$ entry of matrix $K_T$. By A\ref{asu_eigen_sig}, $\lambda_{\min}(\Sigma_T)>0$ and $\lim_{p\to \infty}\lambda_{\max}(\Sigma_T)< \infty$

\item[(b)] Since $\Sigma_{T}^k=\Sigma^{k}\circ K_T$ and $\Sigma^{kj}_{T}=\Sigma^{kj}\circ K_T$
$[d](i)$ and $[d](ii)$ hold by A \ref{asu_star_coveigen2}.
For $[2](iii)$, since $t_{ij,T}(\btheta)=tr(\inSigma_T\Sigma_{i,T}\inSigma_T\Sigma_{j,T})$
\begin{align}
\der{t_{ij,T}(\btheta)}{\theta_l}=&tr(\inSigma_T\Sigma_{l,T}\inSigma_T\Sigma_{i,T}\inSigma_T\Sigma_{j,T})+tr(\inSigma_T\Sigma_{il,T}\inSigma_T\Sigma_{j,T})\\\nonumber
&+tr(\inSigma_T\Sigma_{i,T}\inSigma_T\Sigma_{l,T}\inSigma_T\Sigma_{j,T})+tr(\inSigma_T\Sigma_{i,T}\inSigma_T\Sigma_{jl,T})\\\nonumber
=&(1)+(2)+(3)+(4)
\end{align}
Then (1) can be written as:
\begin{align}
&tr(\inSigma_T\Sigma_{l,T}\inSigma_T\Sigma_{i,T}\inSigma_T\Sigma_{j,T})\\\nonumber
=&tr((\inSigma_T-\inSigma)\Sigma_{l,T}\inSigma_T\Sigma_{i,T}\inSigma_T\Sigma_{j,T})+tr(\inSigma(\Sigma_{l,T}-\Sigma_{l})\inSigma_T\Sigma_{i,T}\inSigma_T\Sigma_{j,T})+ \\\nonumber
&tr(\inSigma\Sigma_{l}(\inSigma_T-\inSigma)\Sigma_{i,T}\inSigma_T\Sigma_{j,T})+tr(\inSigma\Sigma_{l}\inSigma(\Sigma_{i,T}-\Sigma_{i})\inSigma_T\Sigma_{j,T})+\\\nonumber
&tr(\inSigma\Sigma_{l}\inSigma\Sigma_{i}(\inSigma_T-\inSigma)\Sigma_{j,T}) +tr(\inSigma\Sigma_{l}\inSigma\Sigma_{i}\inSigma(\Sigma_{j,T}-\Sigma_{j}))+ \\\nonumber
&tr(\inSigma\Sigma_{l}\inSigma\Sigma_{i}\inSigma\Sigma_{j})
\end{align}
Define $\norm{\cdot}{s}$ for matrix A by $\norm{A}{s}=\max_{i}\{\abs{\lambda_i(A)},i=1,2,...,p\}$, where $\lambda_i(A)$ is the $i$th eigenvalue of matrix A. Notice
$$\norm{\inSigma_T-\inSigma}{s}\le \norm{\inSigma}{s}\norm{\Sigma-\Sigma_T}{s}\norm{\inSigma_T}{s}.$$

Since $\lambda_{\min}(\inSigma)=1/\lambda_{\max}(\Sigma)>0$, $\norm{\inSigma}{s}\le \lambda_{\max}(\inSigma)<\infty$. Also $\norm{\inSigma_T}{s}<\infty$, $\norm{\Sigma_{j,T}}{s}<\infty$ for all $j=1,2,...,q$. Hence $\norm{\inSigma_T-\inSigma}{s}=O_p(p^{-\delta_0})$. Then
\begin{align}
&\abs{tr((\inSigma_T-\inSigma)\Sigma_{l,T}\inSigma_T\Sigma_{i,T}\inSigma_T\Sigma_{j,T})}\\\nonumber
&\le p \norm{((\inSigma_T-\inSigma)\Sigma_{l,T}\inSigma_T\Sigma_{i,T}\inSigma_T\Sigma_{j,T})}{s}\\\nonumber
&\le  p \norm{(\inSigma_T-\inSigma)}{s} \norm{\Sigma_{l,T}}{s}\norm{\inSigma_T}{s}^2\norm{\Sigma_{i,T}}{s}\norm{\Sigma_{j,T}}{s}\\\nonumber
&=O(p/p^{\delta_0})=O(p^{1-\delta_0})
\end{align}
By similar argument, the first six terms in $(1)$ all have the order of $O(p^{1-\delta_0})$. Apply the same argument on $(2)-(4)$ we have:
\begin{align}
\der{t_{ij,T}(\btheta)}{\theta_l}=\der{t_{ij}(\btheta)}{\theta_l}+O(p^{1-\delta_0})
\end{align}
By A\ref{asu_star_covdif3}, $\der{t_{ij,T}(\btheta)}{\theta_l}=O_p(p)$. This completes the proof.
\end{itemize}
\end{proof}

\begin{proof}[\textbf{Proof of Theorem \ref{Th_consist_PMLE_T}}]
From lemma 4, all regularity conditions for $\Sigma_T$ are satisfied. The proof of \ref{Th_consist_PMLE_T} is similar to that of Theorem \ref{Th_consist_PMLE}. By replacing $\Sigma$ by $\Sigma_T$ and replacing A\ref{asu_eigen_sig}-A\ref{asu_star_covdif3} by the results in lemma 4, the results in Theorem \ref{Th_consist_PMLE_T} follows.
\end{proof}

\section{Proofs for classification using PMLE-LDA}
\begin{lemma}\label{lem_taperEstMatrix}
Let $\htheta$ be the estimate of $\btheta_0$ and $\norm{\htheta-\htheta_0}{2}=O_p(\frac{1}{\sqrt{np}})$. Define
\begin{align*}
\tSigma=\Sigma_T(\htheta)=\Sigma(\htheta)\circ K(w)
\end{align*}
where $K(w)$ is defined in section \ref{sec_tapper}. Assume A\ref{asu_domain}, A\ref{asu_covfunc} and A\ref{asu_tapperRange} and A\ref{asu_covfunc_int} hold,then
$$\norm{\tSigma-\Sigma}{2}=O_p(c_n) \ \text{and}\ \norm{\tSigma^{-1}-\inSigma}{2}=O_p(c_n)$$
where $c_{n,p}=max(\frac{w^d}{\sqrt{np}},\frac{1}{w})$
\end{lemma}
\begin{proof}
\begin{align*}
\norm{\tSigma-\Sigma}{2}=&\norm{\Sigma(\htheta)\circ K(w)-\Sigma(\btheta_0)}{2}\\\nonumber
\le&\max_{i}\sum_{j=1}^{p}\abs{r(\htheta;h_{ij})K_T(h_{ij},w)-r(\btheta_0;h_{ij})}
\end{align*}
where $K_T(h,w)=[(1-h/w)_{+}]^2$. For any $i=1,2,...,p$,
\begin{align*}
&\sum_{j=1}^{p}\abs{r(\htheta;h_{ij})K_T(h_{ij},w)-r(\btheta_0;h_{ij})}\\\nonumber
\le & \sum_{h_{ij}<w}\abs{r(\htheta;h_{ij})K_T(h_{ij},w)-r(\btheta_0;h_{ij})}+\sum_{h_{ij}\ge w}\abs{r(\btheta_0,h_{ij})}\\\nonumber
\le & \sum_{h_{ij}<w}\abs{(r(\htheta;h_{ij})-r(\btheta_0;h_{ij}))K_T(h_{ij},w)} +\sum_{h_{ij}<w}\abs{r(\theta_0;h_{ij})K_T(h_{ij},w)-r(\btheta_0;h_{ij})}+\sum_{h_{ij}\ge w}\abs{r(\btheta_0,h_{ij})}\\\nonumber
=&(I)+(II)+(III)
\end{align*}
From the same proof procedure of lemma \ref{lem_tapermatrix}, we have $(II)=O_p(1/w)$ and $(III)=O_p(1/w)$. From A\ref{asu_domain} and A\ref{asu_covfunc}(iii), we have
\begin{align*}
(I)\le& \sum_{h_{ij}<w}\abs{r(\htheta;h_{ij})-r(\btheta_0;h_{ij})}\le\sum_{k=1}^{q}\sum_{h_{ij}<w}\abs{r_k(\btheta^*;h_{ij})}\abs{\hat{\theta}_k^*-\theta_{0k}}\\\nonumber
\le& M w^d \rho \norm{\htheta-\btheta_0}{2}
\end{align*}

Therefore $(I)=O_p(\frac{w^d}{\sqrt{np}})$. Combine (I), (II) and (III), $\norm{\tSigma-\Sigma}{2}=O_p(c_n)$.

Since $\Sigma(\btheta_0)$ and $\tSigma=\Sigma_T(\htheta)$ are positive definite, $\norm{\inSigma}{2}=\frac{1}{\lambda_{\min}(\Sigma)}<\infty$ and $\norm{\tSigma_T^{-1}}{2}=\frac{1}{\lambda_{\min}(\tSigma_T)}<\infty.$
\begin{align*}
\norm{\tSigma_T^{-1}-\inSigma}{2}=\norm{\inSigma(\Sigma-\tSigma_T)\tSigma_T^{-1}}{2}\le \norm{\inSigma}{2}\norm{\Sigma-\tSigma_T}{2}\norm{\tSigma_T^{-1}}{2}=O_p(c_n).
\end{align*}
 and this completes the proof.
\end{proof}

\begin{lemma}
\label{lem_sparseCov}
Assume A\ref{asu_covfunc_int} holds. Then $\max_{1\le i\le s}\sum_{k=s+1}^{p}\sigma_{ik}^2$ is bounded above by a constant.
\end{lemma}
\begin{proof}
Since A\ref{asu_covfunc_int} holds, by similar argument as proving lemma \ref{lem_tapermatrix}, we have:
\begin{align*}
\max_{1\le i\le s}\sum_{k=s+1}^{p}\sigma_{ik}^2=&\max_{1\le i\le s}\sum_{k=s+1}^{p}\gamma(h_{ik},\btheta_0)\le \max_{1\le i\le s}\left(\sum_{0<h_{ij}<1}\gamma(h_{ik},\btheta_0)+\sum_{h_{ij}\ge 1}\gamma(h_{ik},\btheta_0)\right)\\\nonumber
&\le \int_{0}^{1}h^{d-1}\gamma(h,\btheta_0)dh+\int_{1}^{\infty}h^{d-1}\gamma(h,\btheta_0)\\\nonumber
&\le  \int_{h=0}^{1}h^{d-1}\gamma_0(h,\btheta_0)dh+\int_{1}^{\infty}h^{d}\gamma_0(h,\btheta_0)\le \infty
\end{align*}
\end{proof}

\begin{proof}[\textbf{Proof of Theorem \ref{ThErrPgeN}}]
Suppose a new observation is from class $1$, the  conditional misclassification rate of $\hat{\delta}_{PMLE}$ for class $1$ is:
\begin{align}
W_1(\hat{\delta}_{PMLE})=1-\Phi(\frac{(\bmu_1-\bar{\bY}-\frac{n_1-n_2}{2n}\hDelta)^T\inTSigma\hDelta}{\sqrt{\hDelta^T\inTSigma\Sigma\inTSigma\hDelta}})
\end{align}
Where $\bar{\bY}=\sum_{k=1}^{2}\sum_{i=1}^{n_k}\bY_{ki}$. We first consider denominator. From lemma \ref{lem_taperEstMatrix}, $\norm{\Sigma-\tSigma}{2}=O_p(c_n)$ and $\norm{\inSigma-\tSigma^{-1}}{2}=O_p(c_n)$, where $c_n=\max(\frac{w^d}{\sqrt{np}},\frac{1}{w})$ and $w$ is the threshold distance $w$. Then
\begin{align}
\hDelta^T\inTSigma\Sigma\inTSigma\hDelta=&\hDelta^T\inTSigma\hDelta+\hDelta^T\inTSigma(\Sigma-\tSigma)\inTSigma\hDelta\\\nonumber
\le& \hDelta^T\inTSigma\hDelta+\frac{\norm{\Sigma-\tSigma}{2}}{\lambda_{\min}(\tSigma)}\hDelta^T\inTSigma\hDelta\\\nonumber
=& \hDelta^T\inTSigma\hDelta(1+O_p(c_n))\\\nonumber
=& (\hDelta^T\inSigma\hDelta +\hDelta^T(\inTSigma-\inSigma)\hDelta) (1+O_p(c_n))\\\nonumber
=& \hDelta^T\inSigma\hDelta(1+O_p(c_n))
\end{align}

Write
\begin{align}
\hDelta^T\inSigma\hDelta=\bDelta\inSigma\bDelta+2(\hDelta-\bDelta)^T\inSigma\bDelta+(\hDelta-\bDelta)^T\inSigma(\hDelta-\bDelta)\\\nonumber
\end{align}
From Theorem \ref{Th_consist_PMLE}, $\norm{\hDelta-\bDelta}{2}=O_P(\sqrt{\frac{s}{n}})$. Hence $(\hDelta-\bDelta)^T\inSigma(\hDelta-\bDelta)=O_P(\frac{s}{n})$.
Also the second term
\begin{align}
(\hDelta-\bDelta)^T\inSigma\bDelta\le(\bDelta^T\Sigma^{-1}\bDelta)^{\half} \left( (\hDelta-\bDelta)^T\inSigma(\hDelta-\bDelta)\right)^{\half}
\end{align}
Since $\frac{s}{n\bDelta^T\inSigma\bDelta}\to 0$, we have
\begin{align}\label{proof_thp>n_dSd}
\hDelta^T\inSigma\hDelta=&\bDelta^T\Sigma^{-1}\bDelta(1+O_P(\sqrt{\frac{s}{n\bDelta^T\inSigma\bDelta}})+O_P(\frac{s}{n\bDelta^T\inSigma\bDelta}))\\\nonumber
&=\bDelta^T\Sigma^{-1}\bDelta(1+O_P(\sqrt{\frac{s}{n\bDelta^T\inSigma\bDelta}}))
\end{align}

Let $D_{n,p}=\max(\sqrt{\frac{s}{n\bDelta^T\inSigma\bDelta}},c_n)$, the denominator can be represented by:
$\sqrt{\hDelta^T\inTSigma\Sigma\inTSigma\hDelta}=\sqrt{\bDelta\inSigma\bDelta(1+O_p(D_{n,p}))}$.

Now consider the nominator.
\begin{align}
&(\bmu_1-\bar{\bY}-\frac{n_1-n_2}{2n}\hDelta)^T\inTSigma \hDelta\\\nonumber
&=(\bmu_1-\bar{\bY})^T\inTSigma \hDelta+\frac{n_2-n_1}{2n}\hDelta^T \inTSigma  \hDelta \\\nonumber
&=(\bmu_1-\bar{\bY}-\frac{n_2}{n}\bDelta)^T\inTSigma \hDelta +\frac{n_2}{n}\bDelta^T\inTSigma \hDelta+\frac{n_2-n_1}{2n}\hDelta^T \inTSigma \hDelta\\\nonumber
&=(1)+(2)+(3)  \\\nonumber
\end{align}
We start from (3). By lemma \ref{lem_taperEstMatrix}, $\hDelta^T \inTSigma \hDelta=\hDelta^T \inSigma \hDelta(1+O_p(c_n))$. From \ref{proof_thp>n_dSd}, (3) can be represented by $(3)=\frac{n_2-n_1}{2n}\bDelta^T\inSigma\bDelta(1+O_P(D_{n,p}))$.

For (2), first we have $\bDelta^T\inTSigma \hDelta= \bDelta^T\inSigma \hDelta(1+O_p(c_n))$. Then combine \ref{proof_thp>n_dSd}
\begin{align}
\bDelta^T\inSigma \hDelta \le \left(\bDelta^T\inSigma \bDelta \right)^{\half} \left( \hDelta^T\inSigma \hDelta \right)^{\half}= \bDelta^T\inSigma \bDelta(1+O_P(\sqrt{\frac{s}{n\bDelta^T\inSigma\bDelta}}))
\end{align}

Then (2) can be represented by: $(2)=\frac{n_2}{n}\bDelta^T\inSigma \hDelta=\frac{n_2}{n}\bDelta^T\inSigma\bDelta(1+O_P(D_{n,p}))$.

Thus
\begin{align}
(2)+(3)=\half\bDelta^T\inSigma\bDelta(1+O_P(D_{n,p}))
\end{align}

Now we consider (1). Let $\hDelta=(\hDelta_1^T,\hDelta_2^T)^T$ where $\hDelta_1$ is $s$ dimension and $\hDelta_2$ is $p-s$ dimension. From Theorem \ref{Th_consist_PMLE}, with probability tending to 1, $\hDelta_2=0$ and $\norm{\hDelta-\bDelta}{2}=O_P(\frac{s}{n})$.
Let $\xi=\bmu_1-\bar{\bY}-\frac{n_2}{n}\bDelta=(\xi_1^T,\xi_0^T)^T$, where $\xi_1$ is $s$ dimension and $\xi_0$ is $p-s$ dimension. Then $\xi\sim N(0,\frac{1}{n}\Sigma)$

Write
\begin{align*}
\Sigma = \left(\begin{array}{ll}
\Sigma_1&\Sigma_{12}\\
\Sigma_{12}^T&\Sigma_{2}
\end{array}
\right),\  \
\inSigma = \left(\begin{array}{ll}
C_1&C_{12}\\
C_{12}^T&C_{2}
\end{array}
\right)
\end{align*}

\begin{align*}
\tSigma = \left(\begin{array}{ll}
\tSigma_1&\tSigma_{12}\\
\tSigma_{12}^T&\tSigma_{2}
\end{array}
\right), \  \
\inTSigma = \left(\begin{array}{ll}
\tilde{C}_1&\tilde{C}_{12}\\
\tilde{C}_{12}^T&\tilde{C}_{2}
\end{array}
\right)
\end{align*}
where $\Sigma_1$, $\tSigma_1$, $C_1$ and $\tilde{C}_1$ are $s\times s$ matrix. Then

\begin{align*}
C_{12}=-\inSigma_1\Sigma_{12}C_2\ \ \text{and} \ \ \tilde{C}_{12}=-\inTSigma_1\tSigma_{12}\tilde{C}_2.
\end{align*}

Write
\begin{align*}
(1)=&\xi^T\inTSigma \hDelta=\xi_1^T\tilde{C}_1\hDelta_1+\xi_0^T\tilde{C}_{12}\hDelta_1\\\nonumber
=&\xi_1^T\tilde{C}_1\hDelta_1-\xi_0\tilde{C}_2\tSigma_2\inTSigma_1\hDelta_1\\\nonumber
=&(i)+(ii)
\end{align*}

First we have: $\xi_1^T\tilde{C}_1\hDelta_1\le (\xi_1^T\tilde{C}_1\xi_1)^{1/2} (\hDelta_1^T\tilde{C}_1\hDelta_1)^{1/2}$. By lemma\ref{lem_taperEstMatrix},
\begin{align*}
\xi_1^T\tilde{C}_1\xi_1=\xi_1^T {C}_1 \xi_1(1+O_p(c_n))
\end{align*}
Since $\xi_1\sim N(0,\frac{1}{n}\Sigma_1)$, $E(\xi_1^T\inSigma_1\xi_1)=tr(\frac{1}{n}I_s)=\frac{s}{n}$. Then $ \xi_1^T\inSigma_1\xi_1=O_P(\frac{s}{n})$ and therefore $\xi_1^T C_1 \xi_1^T\le \xi_1^T \inSigma \xi_1^T \frac{\lambda_{\max}(C_1)}{\lambda_{\min}(\Sigma_1^{-1})}=O_p(\sqrt{\frac{s}{n}}). $
Hence $\xi_1^T \tilde{C}_1 \xi_1=O_P(\frac{s}{n})$.

Also,
\begin{align*}
\hDelta_1^T\tilde{C}_1\hDelta_1=\hDelta^T \inTSigma \hDelta=\bDelta^T\inSigma\bDelta (1+O_p(D_{n,p}))
\end{align*}
Hence $(i)$ in $(1)$ is: $(i)=\xi_1^T\tilde{C}_1 \hDelta_1=(\bDelta^T \inSigma \bDelta)^{\half} O_p(\sqrt{\frac{s}{n}})$.

The second term in $(1)$ is:
\begin{align*}
(ii)=\xi_0\tilde{C}_2\tSigma_{12}\inTSigma_1 \hDelta_1 \le (\hDelta_1 \inTSigma_{1} \hDelta_1)^{1/2} (\xi_0\tilde{C}_2^T\tSigma_{12}^T \inTSigma_{1}   \tSigma_{12}\tilde{C}_2)^{1/2}
\end{align*}
By lemma \ref{lem_taperEstMatrix}
\begin{align*}
\xi_0^T\tilde{C}_2^T\tSigma_{12}^T \inTSigma_{1}   \tSigma_{12}\tilde{C}_2\xi_0= \xi_0{C}_2^T\Sigma_{12}^T \inSigma_{1}   \Sigma_{12}{C}_2\xi_0(1+O_p(c_n))
\end{align*}

Since $\xi_0\sim N(0,\frac{1}{n}\Sigma_2)$,
\begin{align*}
E[\xi_0^T{C}_2^T\Sigma_{12}^T \inSigma_{1}   \Sigma_{12}{C}_2\xi_0] \le & \lambda_{\max}(\inSigma_1)E[\xi_0^T{C}_2^T\Sigma_{12}^T \Sigma_{12}{C}_2\xi_0]\\\nonumber
=&\lambda_{\max}(\inSigma_1)tr(E[\xi_0^T{C}_2^T\Sigma_{12}^T \Sigma_{12}{C}_2\xi_0])\\\nonumber
=&\lambda_{\max}(\inSigma_1)\frac{1}{n}tr({C}_2^T\Sigma_{12}^T \Sigma_{12}{C}_2\Sigma_2)\\\nonumber
\le& \lambda_{\max}(\inSigma_1)\lambda_{\max}(\Sigma_2)\lambda_{\max}^2(C_2) \frac{1}{n}tr(\Sigma_{12}\Sigma_{12}^T)
\end{align*}

\begin{align}
tr(\Sigma_{12}\Sigma_{12}^T)=\sum_{i=1}^{s}\sum_{k=s+1}^{p}\sigma_{ik}^2\le s \max_{1\le i\le s}\sum_{k=s+1}^{p}\sigma_{ik}^2
\end{align}

thus $\xi_0^T{C}_2^T\Sigma_{12}^T \inSigma_{1}   \Sigma_{12}{C}_2\xi_0=O_p(\frac{s}{n}d_{n,p})$, where $d_{n,p}=\max_{1\le i\le s}\sum_{k=s+1}^{p}\sigma_{ik}^2$. By lemma \ref{lem_sparseCov}, $d_{n,p}$ is bounded above by a constant. As a result:
\begin{align*}
\xi_0^T\tilde{C}_2^T\tSigma_{12}^T \inTSigma_{1}   \tSigma_{12}\tilde{C}_2\xi_0= xi_0^T{C}_2^T\Sigma_{12}^T \inSigma_{1}   \Sigma_{12}{C}_2\xi_0=O_p(\frac{s}{n}d_{n,p}) (1+O_p(\max(c_n,\sqrt{\frac{s}{n}})))
\end{align*}

Since
\begin{align}
\hDelta_1^T \inTSigma_1 \hDelta_1 \le \hDelta_1^T \tilde{C}_1 \hDelta_1\le \hDelta^T \inTSigma \hDelta=\bDelta^T \inSigma \bDelta(1+O_P(D_{n,p}))
\end{align}
Let $A_{n,p}=\max(\sqrt{\frac{s}{n\bDelta^T\inSigma\bDelta}},\frac{s}{n},c_n)$, we have
\begin{align}
(ii)=(\bDelta^T\inSigma\bDelta)^{\half}O_P(A_{n,p}))
\end{align}

Combining the approximation of $(i)$ and $(ii)$ in (1), we have
$(1)=\xi^T\inTSigma \hDelta=(\bDelta^T\inSigma\bDelta)^{\half}O_P(A_{n,p})$.

As a result, since $\frac{\sqrt{s/n}}{C_p}\to 0$ ,
\begin{align}
W_1(\hat{\delta}_{PLDA})=&1-\Phi(\frac{\half\bDelta^T\inSigma\bDelta(1+O_p(D_{n,p}))+\sqrt{\bDelta^T\inSigma\bDelta}O_p(A_{n,p})}{\sqrt{\bDelta^T\inSigma\bDelta(1+O_p(D_{n,p}))}})\\\nonumber
&=1-\Phi \left( \frac{\half C_p(1+O_p(D_{n,p})+\sqrt{C_p}O_p({A_{n,p}})}{\sqrt{C_p}(1+O_p(D_{n,p}))}\right)\\\nonumber
\end{align}

Similarly, we can derive:
\begin{align}
W_2(\hat{\delta}_{PLDA})=&\Phi(\frac{-\half\bDelta^T\inSigma\bDelta(1+O_p(D_{n,p}))+\sqrt{\bDelta^T\inSigma\bDelta}O_p(A_{n,p})}{\sqrt{\bDelta^T\inSigma\bDelta(1+O_p(D_{n,p}))}})\\\nonumber
&=\Phi \left( \frac{-\half C_p(1+O_p(D_{n,p})+\sqrt{C_p}O_p({A_{n,p}})}{\sqrt{C_p}(1+O_p(D_{n,p}))}\right)\\\nonumber
\end{align}
Since $D_{n,p}\to 0$ and $A_{n,p}\to 0$, both $W_1(\hat{\delta}_{PLDA})$ and $W_2(\hat{\delta}_{PLDA})$ go to $1-\Phi(\frac{\sqrt{C_0}}{2})$ as $n,p\to \infty$ with probability tending to 1 as $n, p, s \to \infty$. Then the approximate overall misclassification error rate is $W(\hat{\delta}_{PLDA})=\half(W_1(\hat{\delta}_{PLDA})+W_2(\hat{\delta}_{PLDA}))\to1-\Phi(\frac{\sqrt{C_0}}{2})$. This completes the proof of sub-optimal.

Now we show the asymptotically optimal of $W(\hat{\delta}_{PLDA})$.
If $C_p\to C_0 <\infty$, $\frac{W(\hat{\delta}_{PLDA})}{W_{OPT}}=\frac{W(\hat{\delta}_{PLDA})}{\Phi(-\frac{\sqrt{C_p}}{2})}\to 1$.

If $C_p \to \infty$,
\begin{align*}
\frac{x\sqrt{C_p}}{4+x^2}e^{-\frac{x^2-C_p}{8}}\le \frac{W(\hat{\delta}_{PLDA})}{\Phi(-\frac{\sqrt{C_p}}{2})}\le \frac{4+C_p}{x\sqrt{C_p}}e^{-\frac{x^2-C_p}{8}}
\end{align*}
where $x=\frac{C_p(1+O_p(D_{n,p}))\pm2\sqrt{C_p}O_p(A_{n,p})}{\sqrt{C_p(1+O_p(D_{n,p}))}}=\sqrt{C_p}(1+O_p(D_{n,p}\pm O(\frac{A_{n,p}}{\sqrt{C_p}})))$.

First  $\frac{x\sqrt{C_p}}{4+x^2}e^{-\frac{x^2-C_p}{8}}\to 1$ and $\frac{4+C_p}{x\sqrt{C_p}}\to 1$ as $C_p\to \infty$. Also,
\begin{align*}
x^2-C_p=C_p(O(D_{n,p}+O(\frac{A_{n,p}}{\sqrt{C_p}})))
\end{align*}
if $C_pc_n\to 0$ and $C_p\sqrt{\frac{s}{n}}\to 0$, we have $x^2-C_p\to 0$. Hence $\frac{W(\hat{\delta}_{PLDA})}{W_{OPT}}\to 1$. This completes the proof.
\end{proof}

\end{document}